\documentclass[a4paper,11pt]{article}
\pdfoutput=1
\usepackage{a4wide}
\usepackage[UKenglish]{babel}
\usepackage[noadjust]{cite}
\usepackage{amsmath}
\usepackage{amssymb}
\usepackage{amsthm} 
\usepackage{bm}
\usepackage{bbm}
\usepackage{mathtools}
\usepackage{mathrsfs}
\usepackage{csquotes}
\usepackage{todonotes}
\usepackage{tikz}
\usepackage{tikz-3dplot}
\usetikzlibrary{perspective}
\usepackage{graphicx,wrapfig}
\usepackage[title]{appendix}
\usepackage{bbm}
\usepackage{xcolor}
\usepackage{blindtext}
\usepackage{enumerate}
\usepackage{ifthen}
\usepackage{manfnt}

\newcommand{\cS}{\mathcal{S}}

\colorlet{minusTwoColour}{blue!200}
\colorlet{minusOneColour}{green!50}
\colorlet{zeroColour}{black!10}
\colorlet{plusOneColour}{yellow!80}
\colorlet{plusTwoColour}{orange!80}
\colorlet{plusThreeColour}{red!200}

\colorlet{aColour}{blue!50}
\colorlet{bColour}{green!40}
\colorlet{cColour}{purple!60}
\colorlet{dColour}{yellow!200}
\colorlet{eColour}{orange!80}
\colorlet{fColour}{red!200}

\pgfmathsetmacro{\bigShift}{4.5}
\pgfmathsetmacro{\distanceFactor}{1.2}

\newcommand\squareTensorA[2]{
\begin{tikzpicture}[baseline=#1 ex, local bounding box=C,scale=#2]
\draw[fill=zeroColour,scale=.7] (0,2) rectangle ++(1,1);
\draw[fill=zeroColour,scale=.7] (1,2) rectangle ++(1,1);
\draw[fill=plusOneColour,scale=.7] (2,2) rectangle ++(1,1);
\draw[fill=plusOneColour,scale=.7] (0,1) rectangle ++(1,1);
\draw[fill=zeroColour,scale=.7] (1,1) rectangle ++(1,1);
\draw[fill=minusOneColour,scale=.7] (2,1) rectangle ++(1,1);
\draw[fill=zeroColour,scale=.7] (0,0) rectangle ++(1,1);
\draw[fill=zeroColour,scale=.7] (1,0) rectangle ++(1,1);
\draw[fill=zeroColour,scale=.7] (2,0) rectangle ++(1,1);
\end{tikzpicture}
}

\pgfmathsetmacro{\opacityColouredBlocksFront}{.95}
\pgfmathsetmacro{\opacityColouredBlocksMiddle}{.65}
\pgfmathsetmacro{\opacityColouredBlocksBack}{.35}

\pgfmathsetmacro{\opacityColouredBlocksLayerOne}{1}
\pgfmathsetmacro{\opacityColouredBlocksLayerTwo}{.85}
\pgfmathsetmacro{\opacityColouredBlocksLayerThree}{.7}
\pgfmathsetmacro{\opacityColouredBlocksLayerFour}{.55}
\pgfmathsetmacro{\opacityColouredBlocksLayerFive}{.4}
\pgfmathsetmacro{\opacityColouredBlocksLayerSix}{.25}

\pgfmathsetmacro{\opacityZeroBlocksLayerOne}{.8}
\pgfmathsetmacro{\opacityZeroBlocksLayerTwo}{.6}
\pgfmathsetmacro{\opacityZeroBlocksLayerThree}{.4}
\pgfmathsetmacro{\opacityZeroBlocksLayerFour}{.35}
\pgfmathsetmacro{\opacityZeroBlocksLayerFive}{.3}
\pgfmathsetmacro{\opacityZeroBlocksLayerSix}{.25}

\pgfmathsetmacro{\border}{0.1}
\pgfmathsetmacro{\distanceShadowXY}{12}
\pgfmathsetmacro{\distanceShadowXZ}{-3}
\pgfmathsetmacro{\distanceShadowYZ}{8.5}
\pgfmathsetmacro{\blockseparation}{1}

\newcommand\mysquareXY[5]{
\fill[#1,opacity=#2,shift={(-#5,#4,-#3)},scale=.7](0,0,0)--++(0,-1,0)--++(0,0,-1)--++(0,1,0)--cycle;
}

\newcommand\mysquareXZ[5]{
\fill[#1,opacity=#2,shift={(-#5,#4,-#3)},scale=.7](0,0,0)--++(-1,0,0)--++(0,0,-1)--++(1,0,0)--cycle;
}

\newcommand\mysquareYZ[5]{
\fill[#1,opacity=#2,shift={(-#5,#4,-#3)},scale=.7](0,0,0)--++(-1,0,0)--++(0,-1,0)--++(1,0,0)--cycle;
}

\newcommand\mycube[6]{
\draw[line width=\border,fill=#1,opacity=#2,shift={(-#5,#4,-#3)},scale=#6](0,0,0)--++(0,-1,0)--++(0,0,-1)--++(0,1,0)--cycle;
\draw[line width=\border,fill=#1,opacity=#2,shift={(-#5,#4,-#3)},scale=#6](0,0,0)--++(-1,0,0)--++(0,0,-1)--++(1,0,0)--cycle;
\draw[line width=\border,fill=#1,opacity=#2,shift={(-#5,#4,-#3)},scale=#6](0,0,0)--++(-1,0,0)--++(0,-1,0)--++(1,0,0)--cycle;
}

\newcommand\mycubeNoBorder[5]{
\fill[#1,opacity=#2,shift={(-#5,#4,-#3)},scale=.7](0,0,0)--++(0,-1,0)--++(0,0,-1)--++(0,1,0)--cycle;
\fill[#1,opacity=#2,shift={(-#5,#4,-#3)},scale=.7](0,0,0)--++(-1,0,0)--++(0,0,-1)--++(1,0,0)--cycle;
\fill[#1,opacity=#2,shift={(-#5,#4,-#3)},scale=.7](0,0,0)--++(-1,0,0)--++(0,-1,0)--++(1,0,0)--cycle;
}

\newcommand\cubeTensorW[2]{
\begin{tikzpicture}[baseline=#1 ex, local bounding box=C,scale=#2]
\begin{scope}[3d view={110}{15}]
  \foreach \y in {0,...,5} 
	 \foreach \x in {5,...,0}
		 \mycubeNoBorder{zeroColour}{\opacityZeroBlocksLayerSix}{\distanceFactor*\x}{\distanceFactor*\y}{\distanceFactor*5};
 \foreach \y in {0,...,5} 
	 \foreach \x in {5,...,0}
		 \mycubeNoBorder{zeroColour}{\opacityZeroBlocksLayerFive}{\distanceFactor*\x}{\distanceFactor*\y}{\distanceFactor*4};
 \foreach \y in {0,...,5} 
	 \foreach \x in {5,...,0}
		 \mycubeNoBorder{zeroColour}{\opacityZeroBlocksLayerFour}{\distanceFactor*\x}{\distanceFactor*\y}{\distanceFactor*3};
 \foreach \y in {0,...,5} 
	 \foreach \x in {5,...,0}
		 \mycubeNoBorder{zeroColour}{\opacityZeroBlocksLayerThree}{\distanceFactor*\x}{\distanceFactor*\y}{\distanceFactor*2};
 \foreach \y in {0,...,5} 
	 \foreach \x in {5,...,0}
		 \mycubeNoBorder{zeroColour}{\opacityZeroBlocksLayerTwo}{\distanceFactor*\x}{\distanceFactor*\y}{\distanceFactor*1};
 \foreach \y in {0,...,5} 
	 \foreach \x in {5,...,0}
		 \mycubeNoBorder{zeroColour}{\opacityZeroBlocksLayerOne}{\distanceFactor*\x}{\distanceFactor*\y}{\distanceFactor*0};

\mycube{minusOneColour}{\opacityColouredBlocksBack}{\distanceFactor*2}{\distanceFactor*1}{\distanceFactor*2}{.7};
\mycube{plusOneColour}{\opacityColouredBlocksBack}{\distanceFactor*2}{\distanceFactor*0}{\distanceFactor*2}{.7};
\mycube{minusOneColour}{\opacityColouredBlocksBack}{\distanceFactor*1}{\distanceFactor*0}{\distanceFactor*2}{.7};
\mycube{plusOneColour}{\opacityColouredBlocksBack}{\distanceFactor*0}{\distanceFactor*0}{\distanceFactor*2}{.7};
\mycube{minusOneColour}{\opacityColouredBlocksFront}{\distanceFactor*1}{\distanceFactor*2}{\distanceFactor*0}{.7};
\mycube{plusOneColour}{\opacityColouredBlocksFront}{\distanceFactor*0}{\distanceFactor*2}{\distanceFactor*0}{.7};
\mycube{plusOneColour}{\opacityColouredBlocksFront}{\distanceFactor*2}{\distanceFactor*1}{\distanceFactor*0}{.7};
\mycube{minusOneColour}{\opacityColouredBlocksFront}{\distanceFactor*2}{\distanceFactor*0}{\distanceFactor*0}{.7};
\mycube{plusTwoColour}{\opacityColouredBlocksFront}{\distanceFactor*1}{\distanceFactor*0}{\distanceFactor*0}{.7};
\mycube{minusOneColour}{\opacityColouredBlocksFront}{\distanceFactor*0}{\distanceFactor*0}{\distanceFactor*0}{.7};
\end{scope}
\end{tikzpicture}
}

\newcommand\quartzExample[2]{
\begin{tikzpicture}[baseline=#1 ex, local bounding box=C,scale=#2]
\begin{scope}[3d view={110}{15}]
  \foreach \y in {0,...,5} 
	 \foreach \x in {5,...,0}
		 \mycubeNoBorder{zeroColour}{\opacityZeroBlocksLayerSix}{\distanceFactor*\x}{\distanceFactor*\y}{\distanceFactor*5};
 \foreach \y in {0,...,5} 
	 \foreach \x in {5,...,0}
		 \mycubeNoBorder{zeroColour}{\opacityZeroBlocksLayerFive}{\distanceFactor*\x}{\distanceFactor*\y}{\distanceFactor*4};
 \foreach \y in {0,...,5} 
	 \foreach \x in {5,...,0}
		 \mycubeNoBorder{zeroColour}{\opacityZeroBlocksLayerFour}{\distanceFactor*\x}{\distanceFactor*\y}{\distanceFactor*3};
 \foreach \y in {0,...,5} 
	 \foreach \x in {5,...,0}
		 \mycubeNoBorder{zeroColour}{\opacityZeroBlocksLayerThree}{\distanceFactor*\x}{\distanceFactor*\y}{\distanceFactor*2};
 \foreach \y in {0,...,5} 
	 \foreach \x in {5,...,0}
		 \mycubeNoBorder{zeroColour}{\opacityZeroBlocksLayerTwo}{\distanceFactor*\x}{\distanceFactor*\y}{\distanceFactor*1};
 \foreach \y in {0,...,5} 
	 \foreach \x in {5,...,0}
		 \mycubeNoBorder{zeroColour}{\opacityZeroBlocksLayerOne}{\distanceFactor*\x}{\distanceFactor*\y}{\distanceFactor*0};

\draw[black, dashed, opacity=.6] (0-.35,0-.35,0-.35) -- (-\distanceFactor*5-.35,\distanceFactor*0-.35,-\distanceFactor*0-.35)--
(-\distanceFactor*5-.35,\distanceFactor*4-.35,-\distanceFactor*0-.35)--
(-\distanceFactor*0-.35,\distanceFactor*4-.35,-\distanceFactor*0-.35)--
(0-.35,0-.35,0-.35);
\draw[black, dashed, opacity=.6] (0-.35,0-.35,-\distanceFactor*3-.35) -- (-\distanceFactor*5-.35,\distanceFactor*0-.35,-\distanceFactor*3-.35)--
(-\distanceFactor*5-.35,\distanceFactor*4-.35,-\distanceFactor*3-.35)--
(-\distanceFactor*0-.35,\distanceFactor*4-.35,-\distanceFactor*3-.35)--
(0-.35,0-.35,-\distanceFactor*3-.35);
\draw[black, dashed, opacity=.6] (0-.35,0-.35,0-.35) --(0-.35,0-.35,-\distanceFactor*3-.35);
\draw[black, dashed, opacity=.6] (-\distanceFactor*5-.35,\distanceFactor*0-.35,-\distanceFactor*0-.35)--(-\distanceFactor*5-.35,\distanceFactor*0-.35,-\distanceFactor*3-.35);
\draw[black, dashed, opacity=.6] (-\distanceFactor*5-.35,\distanceFactor*4-.35,-\distanceFactor*0-.35)--(-\distanceFactor*5-.35,\distanceFactor*4-.35,-\distanceFactor*3-.35);
\draw[black, dashed, opacity=.6] (-\distanceFactor*0-.35,\distanceFactor*4-.35,-\distanceFactor*0-.35)--(-\distanceFactor*0-.35,\distanceFactor*4-.35,-\distanceFactor*3-.35);

\mycube{minusOneColour}{\opacityColouredBlocksBack}{\distanceFactor*3}{\distanceFactor*4}{\distanceFactor*5}{.7};
\mycube{plusOneColour}{\opacityColouredBlocksBack}{\distanceFactor*0}{\distanceFactor*4}{\distanceFactor*5}{.7};
\mycube{plusOneColour}{\opacityColouredBlocksBack}{\distanceFactor*3}{\distanceFactor*0}{\distanceFactor*5}{.7};
\mycube{minusOneColour}{\opacityColouredBlocksBack}{\distanceFactor*0}{\distanceFactor*0}{\distanceFactor*5}{.7};
\mycube{plusOneColour}{\opacityColouredBlocksFront}{\distanceFactor*3}{\distanceFactor*4}{\distanceFactor*0}{.7};
\mycube{minusOneColour}{\opacityColouredBlocksFront}{\distanceFactor*0}{\distanceFactor*4}{\distanceFactor*0}{.7};
\mycube{minusOneColour}{\opacityColouredBlocksFront}{\distanceFactor*3}{\distanceFactor*0}{\distanceFactor*0}{.7};
\mycube{plusOneColour}{\opacityColouredBlocksFront}{\distanceFactor*0}{\distanceFactor*0}{\distanceFactor*0}{.7};

\draw(0,0,.75) node {$\ba$};
\draw(-\distanceFactor*5,\distanceFactor*4+.3,-\distanceFactor*3+.75) node {$\bb$};

\end{scope}
\end{tikzpicture}
}

\newcommand\WMinusOneQuartz[2]{
\begin{tikzpicture}[baseline=#1 ex, local bounding box=C,scale=#2]
\begin{scope}[3d view={110}{15}]
  \foreach \y in {0,...,5} 
	 \foreach \x in {5,...,0}
		 \mycubeNoBorder{zeroColour}{\opacityZeroBlocksLayerSix}{\distanceFactor*\x}{\distanceFactor*\y}{\distanceFactor*5};
 \foreach \y in {0,...,5} 
	 \foreach \x in {5,...,0}
		 \mycubeNoBorder{zeroColour}{\opacityZeroBlocksLayerFive}{\distanceFactor*\x}{\distanceFactor*\y}{\distanceFactor*4};
 \foreach \y in {0,...,5} 
	 \foreach \x in {5,...,0}
		 \mycubeNoBorder{zeroColour}{\opacityZeroBlocksLayerFour}{\distanceFactor*\x}{\distanceFactor*\y}{\distanceFactor*3};
 \foreach \y in {0,...,5} 
	 \foreach \x in {5,...,0}
		 \mycubeNoBorder{zeroColour}{\opacityZeroBlocksLayerThree}{\distanceFactor*\x}{\distanceFactor*\y}{\distanceFactor*2};
 \foreach \y in {0,...,5} 
	 \foreach \x in {5,...,0}
		 \mycubeNoBorder{zeroColour}{\opacityZeroBlocksLayerTwo}{\distanceFactor*\x}{\distanceFactor*\y}{\distanceFactor*1};
 \foreach \y in {0,...,5} 
	 \foreach \x in {5,...,0}
		 \mycubeNoBorder{zeroColour}{\opacityZeroBlocksLayerOne}{\distanceFactor*\x}{\distanceFactor*\y}{\distanceFactor*0}; 

            \mycube{plusOneColour}{\opacityColouredBlocksLayerSix}{\distanceFactor*3}{\distanceFactor*0}{\distanceFactor*5}{.7};\mycube{minusOneColour}{\opacityColouredBlocksLayerSix}{\distanceFactor*0}{\distanceFactor*0}{\distanceFactor*5}{.7};\mycube{minusOneColour}{\opacityColouredBlocksLayerSix}{\distanceFactor*3}{\distanceFactor*4}{\distanceFactor*5}{.7};\mycube{plusOneColour}{\opacityColouredBlocksLayerSix}{\distanceFactor*0}{\distanceFactor*4}{\distanceFactor*5}{.7};\mycube{plusOneColour}{\opacityColouredBlocksLayerThree}{\distanceFactor*2}{\distanceFactor*0}{\distanceFactor*2}{.7};\mycube{minusOneColour}{\opacityColouredBlocksLayerThree}{\distanceFactor*1}{\distanceFactor*0}{\distanceFactor*2}{.7};\mycube{plusOneColour}{\opacityColouredBlocksLayerThree}{\distanceFactor*0}{\distanceFactor*0}{\distanceFactor*2}{.7};\mycube{minusOneColour}{\opacityColouredBlocksLayerThree}{\distanceFactor*2}{\distanceFactor*1}{\distanceFactor*2}{.7};\mycube{minusOneColour}{\opacityColouredBlocksLayerOne}{\distanceFactor*3}{\distanceFactor*0}{\distanceFactor*0}{.7};\mycube{minusOneColour}{\opacityColouredBlocksLayerOne}{\distanceFactor*2}{\distanceFactor*0}{\distanceFactor*0}{.7};\mycube{plusTwoColour}{\opacityColouredBlocksLayerOne}{\distanceFactor*1}{\distanceFactor*0}{\distanceFactor*0}{.7};\mycube{plusOneColour}{\opacityColouredBlocksLayerOne}{\distanceFactor*2}{\distanceFactor*1}{\distanceFactor*0}{.7};\mycube{minusOneColour}{\opacityColouredBlocksLayerOne}{\distanceFactor*1}{\distanceFactor*2}{\distanceFactor*0}{.7};\mycube{plusOneColour}{\opacityColouredBlocksLayerOne}{\distanceFactor*0}{\distanceFactor*2}{\distanceFactor*0}{.7};\mycube{plusOneColour}{\opacityColouredBlocksLayerOne}{\distanceFactor*3}{\distanceFactor*4}{\distanceFactor*0}{.7};\mycube{minusOneColour}{\opacityColouredBlocksLayerOne}{\distanceFactor*0}{\distanceFactor*4}{\distanceFactor*0}{.7};

\end{scope}
\end{tikzpicture}
}

\newcommand\HollowCrystal[2]{
\begin{tikzpicture}[baseline=#1 ex, local bounding box=C,scale=#2]
\begin{scope}[3d view={110}{15}]

\foreach \x in {0,...,5}
    \foreach \y in {0,...,5}
        \mysquareXY{zeroColour}{.6}{\distanceFactor*\x}{\distanceFactor*\y}{\distanceShadowXY};
        
\mysquareXY{plusOneColour}{1}{\distanceFactor*0}{\distanceFactor*2}{\distanceShadowXY};
\mysquareXY{plusOneColour}{1}{\distanceFactor*1}{\distanceFactor*0}{\distanceShadowXY};
\mysquareXY{minusOneColour}{1}{\distanceFactor*1}{\distanceFactor*2}{\distanceShadowXY};
        
\foreach \x in {0,...,5}
    \foreach \z in {0,...,5}
        \mysquareXZ{zeroColour}{.6}{\distanceFactor*\x}{\distanceShadowXZ}{\distanceFactor*\z};
        
\mysquareXZ{plusOneColour}{1}{\distanceFactor*0}{\distanceShadowXZ}{\distanceFactor*2};
\mysquareXZ{plusOneColour}{1}{\distanceFactor*1}{\distanceShadowXZ}{\distanceFactor*0};
\mysquareXZ{minusOneColour}{1}{\distanceFactor*1}{\distanceShadowXZ}{\distanceFactor*2};
        
\foreach \y in {0,...,5}
    \foreach \z in {0,...,5}
        \mysquareYZ{zeroColour}{.6}{\distanceShadowYZ}{\distanceFactor*\y}{\distanceFactor*\z};
        
\mysquareYZ{plusOneColour}{1}{\distanceShadowYZ}{\distanceFactor*0}{\distanceFactor*2};
\mysquareYZ{plusOneColour}{1}{\distanceShadowYZ}{\distanceFactor*1}{\distanceFactor*0};
\mysquareYZ{minusOneColour}{1}{\distanceShadowYZ}{\distanceFactor*1}{\distanceFactor*2};

  \foreach \y in {0,...,5} 
	 \foreach \x in {5,...,0}
		 \mycubeNoBorder{zeroColour}{\opacityZeroBlocksLayerSix}{\distanceFactor*\x}{\distanceFactor*\y}{\distanceFactor*5};
 \foreach \y in {0,...,5} 
	 \foreach \x in {5,...,0}
		 \mycubeNoBorder{zeroColour}{\opacityZeroBlocksLayerFive}{\distanceFactor*\x}{\distanceFactor*\y}{\distanceFactor*4};
 \foreach \y in {0,...,5} 
	 \foreach \x in {5,...,0}
		 \mycubeNoBorder{zeroColour}{\opacityZeroBlocksLayerFour}{\distanceFactor*\x}{\distanceFactor*\y}{\distanceFactor*3};
 \foreach \y in {0,...,5} 
	 \foreach \x in {5,...,0}
		 \mycubeNoBorder{zeroColour}{\opacityZeroBlocksLayerThree}{\distanceFactor*\x}{\distanceFactor*\y}{\distanceFactor*2};
 \foreach \y in {0,...,5} 
	 \foreach \x in {5,...,0}
		 \mycubeNoBorder{zeroColour}{\opacityZeroBlocksLayerTwo}{\distanceFactor*\x}{\distanceFactor*\y}{\distanceFactor*1};
 \foreach \y in {0,...,5} 
	 \foreach \x in {5,...,0}
		 \mycubeNoBorder{zeroColour}{\opacityZeroBlocksLayerOne}{\distanceFactor*\x}{\distanceFactor*\y}{\distanceFactor*0};

\mycube{minusOneColour}{\opacityColouredBlocksLayerSix}{\distanceFactor*3}{\distanceFactor*0}{\distanceFactor*5}{.7};
\mycube{plusOneColour}{\opacityColouredBlocksLayerSix}{\distanceFactor*1}{\distanceFactor*0}{\distanceFactor*5}{.7};
\mycube{minusOneColour}{\opacityColouredBlocksLayerSix}{\distanceFactor*1}{\distanceFactor*2}{\distanceFactor*5}{.7};
\mycube{plusOneColour}{\opacityColouredBlocksLayerSix}{\distanceFactor*0}{\distanceFactor*2}{\distanceFactor*5}{.7};
\mycube{plusOneColour}{\opacityColouredBlocksLayerSix}{\distanceFactor*3}{\distanceFactor*4}{\distanceFactor*5}{.7};
\mycube{minusOneColour}{\opacityColouredBlocksLayerSix}{\distanceFactor*0}{\distanceFactor*4}{\distanceFactor*5}{.7};
\mycube{plusOneColour}{\opacityColouredBlocksLayerThree}{\distanceFactor*3}{\distanceFactor*0}{\distanceFactor*2}{.7};
\mycube{minusOneColour}{\opacityColouredBlocksLayerThree}{\distanceFactor*3}{\distanceFactor*1}{\distanceFactor*2}{.7};
\mycube{minusOneColour}{\opacityColouredBlocksLayerThree}{\distanceFactor*1}{\distanceFactor*4}{\distanceFactor*2}{.7};
\mycube{plusOneColour}{\opacityColouredBlocksLayerThree}{\distanceFactor*0}{\distanceFactor*4}{\distanceFactor*2}{.7};
\mycube{plusOneColour}{\opacityColouredBlocksLayerOne}{\distanceFactor*3}{\distanceFactor*1}{\distanceFactor*0}{.7};
\mycube{minusOneColour}{\opacityColouredBlocksLayerOne}{\distanceFactor*3}{\distanceFactor*4}{\distanceFactor*0}{.7};
\mycube{plusOneColour}{\opacityColouredBlocksLayerOne}{\distanceFactor*1}{\distanceFactor*4}{\distanceFactor*0}{.7};

\end{scope}
\end{tikzpicture}
}

\newcommand\SCrystal[2]{
\begin{tikzpicture}[baseline=#1 ex, local bounding box=C,scale=#2]
\begin{scope}[3d view={110}{15}]

\foreach \x in {0,...,2}
    \foreach \y in {0,...,2}
        \mysquareXY{zeroColour}{.6}{\distanceFactor*\x}{\distanceFactor*\y}{\distanceShadowXY};
        
\mysquareXY{eColour}{1}{\distanceFactor*0}{\distanceFactor*1}{\distanceShadowXY};
\mysquareXY{cColour}{1}{\distanceFactor*0}{\distanceFactor*2}{\distanceShadowXY};
\mysquareXY{fColour}{1}{\distanceFactor*1}{\distanceFactor*0}{\distanceShadowXY};
\mysquareXY{aColour}{1}{\distanceFactor*1}{\distanceFactor*2}{\distanceShadowXY};
\mysquareXY{dColour}{1}{\distanceFactor*2}{\distanceFactor*0}{\distanceShadowXY};
\mysquareXY{bColour}{1}{\distanceFactor*2}{\distanceFactor*1}{\distanceShadowXY};

\foreach \x in {0,...,2}
    \foreach \z in {0,...,2}
        \mysquareXZ{zeroColour}{.6}{\distanceFactor*\x}{\distanceShadowXZ}{\distanceFactor*\z};
        
\mysquareXZ{cColour}{1}{\distanceFactor*0}{\distanceShadowXZ}{\distanceFactor*1};
\mysquareXZ{eColour}{1}{\distanceFactor*0}{\distanceShadowXZ}{\distanceFactor*2};
\mysquareXZ{aColour}{1}{\distanceFactor*1}{\distanceShadowXZ}{\distanceFactor*0};
\mysquareXZ{fColour}{1}{\distanceFactor*1}{\distanceShadowXZ}{\distanceFactor*2};
\mysquareXZ{bColour}{1}{\distanceFactor*2}{\distanceShadowXZ}{\distanceFactor*0};
\mysquareXZ{dColour}{1}{\distanceFactor*2}{\distanceShadowXZ}{\distanceFactor*1};

\foreach \y in {0,...,2}
    \foreach \z in {0,...,2}
        \mysquareYZ{zeroColour}{.6}{5}{\distanceFactor*\y}{\distanceFactor*\z};
        
\mysquareYZ{dColour}{1}{5}{\distanceFactor*0}{\distanceFactor*1};
\mysquareYZ{fColour}{1}{5}{\distanceFactor*0}{\distanceFactor*2};
\mysquareYZ{bColour}{1}{5}{\distanceFactor*1}{\distanceFactor*0};
\mysquareYZ{eColour}{1}{5}{\distanceFactor*1}{\distanceFactor*2};
\mysquareYZ{aColour}{1}{5}{\distanceFactor*2}{\distanceFactor*0};
\mysquareYZ{cColour}{1}{5}{\distanceFactor*2}{\distanceFactor*1};

\mycube{eColour}{\opacityColouredBlocksLayerThree}{\distanceFactor*0}{\distanceFactor*1}{\distanceFactor*2}{.7};
\mycube{fColour}{\opacityColouredBlocksLayerThree}{\distanceFactor*1}{\distanceFactor*0}{\distanceFactor*2}{.7};
 \foreach \y in {0,...,2} 
	 \foreach \x in {2,...,0}
		 \mycubeNoBorder{zeroColour}{.3}{\distanceFactor*\x}{\distanceFactor*\y}{\distanceFactor*2};
 \mycube{cColour}{\opacityColouredBlocksLayerTwo}{\distanceFactor*0}{\distanceFactor*2}{\distanceFactor*1}{.7};
\mycube{dColour}{\opacityColouredBlocksLayerTwo}{\distanceFactor*2}{\distanceFactor*0}{\distanceFactor*1}{.7};
 \foreach \y in {0,...,2} 
	 \foreach \x in {2,...,0}
		 \mycubeNoBorder{zeroColour}{.4}{\distanceFactor*\x}{\distanceFactor*\y}{\distanceFactor*1};
 \foreach \y in {0,...,2} 
	 \foreach \x in {2,...,0}
		 \mycubeNoBorder{zeroColour}{.7}{\distanceFactor*\x}{\distanceFactor*\y}{\distanceFactor*0};

\mycube{aColour}{\opacityColouredBlocksLayerOne}{\distanceFactor*1}{\distanceFactor*2}{\distanceFactor*0}{.7};
\mycube{bColour}{\opacityColouredBlocksLayerOne}{\distanceFactor*2}{\distanceFactor*1}{\distanceFactor*0}{.7};

\end{scope}
\end{tikzpicture}
}

\newcommand\VCrystal[2]{
\begin{tikzpicture}[baseline=#1 ex, local bounding box=C,scale=#2]
\begin{scope}[3d view={110}{15}]

\foreach \x in {0,...,2}
    \foreach \y in {0,...,2}
        \mysquareXY{zeroColour}{.6}{\distanceFactor*\x}{\distanceFactor*\y}{\distanceShadowXY};
        
\mysquareXY{plusOneColour}{1}{\distanceFactor*0}{\distanceFactor*2}{\distanceShadowXY};
\mysquareXY{plusOneColour}{1}{\distanceFactor*1}{\distanceFactor*0}{\distanceShadowXY};
\mysquareXY{minusOneColour}{1}{\distanceFactor*1}{\distanceFactor*2}{\distanceShadowXY};

\foreach \x in {0,...,2}
    \foreach \z in {0,...,2}
        \mysquareXZ{zeroColour}{.6}{\distanceFactor*\x}{\distanceShadowXZ}{\distanceFactor*\z};
        
\mysquareXZ{plusOneColour}{1}{\distanceFactor*0}{\distanceShadowXZ}{\distanceFactor*2};
\mysquareXZ{plusOneColour}{1}{\distanceFactor*1}{\distanceShadowXZ}{\distanceFactor*0};
\mysquareXZ{minusOneColour}{1}{\distanceFactor*1}{\distanceShadowXZ}{\distanceFactor*2};

\foreach \y in {0,...,2}
    \foreach \z in {0,...,2}
        \mysquareYZ{zeroColour}{.6}{5}{\distanceFactor*\y}{\distanceFactor*\z};
        
\mysquareYZ{plusOneColour}{1}{5}{\distanceFactor*0}{\distanceFactor*2};
\mysquareYZ{plusOneColour}{1}{5}{\distanceFactor*1}{\distanceFactor*0};
\mysquareYZ{minusOneColour}{1}{5}{\distanceFactor*1}{\distanceFactor*2};

\mycube{plusOneColour}{\opacityColouredBlocksLayerThree}{\distanceFactor*0}{\distanceFactor*0}{\distanceFactor*2}{.7};
\mycube{minusOneColour}{\opacityColouredBlocksLayerThree}{\distanceFactor*1}{\distanceFactor*0}{\distanceFactor*2}{.7};
\mycube{plusOneColour}{\opacityColouredBlocksLayerThree}{\distanceFactor*2}{\distanceFactor*0}{\distanceFactor*2}{.7};
\mycube{minusOneColour}{\opacityColouredBlocksLayerThree}{\distanceFactor*2}{\distanceFactor*1}{\distanceFactor*2}{.7};

\foreach \y in {0,...,2} 
	 \foreach \x in {2,...,0}
		 \mycubeNoBorder{zeroColour}{.3}{\distanceFactor*\x}{\distanceFactor*\y}{\distanceFactor*2};
\foreach \y in {0,...,2} 
	 \foreach \x in {2,...,0}
		 \mycubeNoBorder{zeroColour}{.4}{\distanceFactor*\x}{\distanceFactor*\y}{\distanceFactor*1};
\foreach \y in {0,...,2} 
	 \foreach \x in {2,...,0}
		 \mycubeNoBorder{zeroColour}{.7}{\distanceFactor*\x}{\distanceFactor*\y}{\distanceFactor*0};

\mycube{minusOneColour}{\opacityColouredBlocksLayerOne}{\distanceFactor*0}{\distanceFactor*0}{\distanceFactor*0}{.7};
\mycube{plusTwoColour}{\opacityColouredBlocksLayerOne}{\distanceFactor*1}{\distanceFactor*0}{\distanceFactor*0}{.7};
\mycube{minusOneColour}{\opacityColouredBlocksLayerOne}{\distanceFactor*2}{\distanceFactor*0}{\distanceFactor*0}{.7};
\mycube{plusOneColour}{\opacityColouredBlocksLayerOne}{\distanceFactor*2}{\distanceFactor*1}{\distanceFactor*0}{.7};
\mycube{plusOneColour}{\opacityColouredBlocksLayerOne}{\distanceFactor*0}{\distanceFactor*2}{\distanceFactor*0}{.7};
\mycube{minusOneColour}{\opacityColouredBlocksLayerOne}{\distanceFactor*1}{\distanceFactor*2}{\distanceFactor*0}{.7};

\end{scope}
\end{tikzpicture}
}

\newcommand\SProjectionXY[2]{
\begin{tikzpicture}[baseline=#1 ex, local bounding box=C,scale=#2]
\draw[fill=zeroColour,scale=.7] (0,0) rectangle ++(1,1);
\draw[fill=eColour,scale=.7] (1,0) rectangle ++(1,1);
\draw[fill=cColour,scale=.7] (2,0) rectangle ++(1,1);
\draw[fill=fColour,scale=.7] (0,-1) rectangle ++(1,1);
\draw[fill=zeroColour,scale=.7] (1,-1) rectangle ++(1,1);
\draw[fill=aColour,scale=.7] (2,-1) rectangle ++(1,1);
\draw[fill=dColour,scale=.7] (0,-2) rectangle ++(1,1);
\draw[fill=bColour,scale=.7] (1,-2) rectangle ++(1,1);
\draw[fill=zeroColour,scale=.7] (2,-2) rectangle ++(1,1);
\end{tikzpicture}
}

\newcommand\SProjectionXZ[2]{
\begin{tikzpicture}[baseline=#1 ex, local bounding box=C,scale=#2]
\draw[fill=zeroColour,scale=.7] (0,0) rectangle ++(1,1);
\draw[fill=cColour,scale=.7] (1,0) rectangle ++(1,1);
\draw[fill=eColour,scale=.7] (2,0) rectangle ++(1,1);
\draw[fill=aColour,scale=.7] (0,-1) rectangle ++(1,1);
\draw[fill=zeroColour,scale=.7] (1,-1) rectangle ++(1,1);
\draw[fill=fColour,scale=.7] (2,-1) rectangle ++(1,1);
\draw[fill=bColour,scale=.7] (0,-2) rectangle ++(1,1);
\draw[fill=dColour,scale=.7] (1,-2) rectangle ++(1,1);
\draw[fill=zeroColour,scale=.7] (2,-2) rectangle ++(1,1);
\end{tikzpicture}
}

\newcommand\SProjectionYZ[2]{
\begin{tikzpicture}[baseline=#1 ex, local bounding box=C,scale=#2]
\draw[fill=zeroColour,scale=.7] (0,0) rectangle ++(1,1);
\draw[fill=dColour,scale=.7] (1,0) rectangle ++(1,1);
\draw[fill=fColour,scale=.7] (2,0) rectangle ++(1,1);
\draw[fill=bColour,scale=.7] (0,-1) rectangle ++(1,1);
\draw[fill=zeroColour,scale=.7] (1,-1) rectangle ++(1,1);
\draw[fill=eColour,scale=.7] (2,-1) rectangle ++(1,1);
\draw[fill=aColour,scale=.7] (0,-2) rectangle ++(1,1);
\draw[fill=cColour,scale=.7] (1,-2) rectangle ++(1,1);
\draw[fill=zeroColour,scale=.7] (2,-2) rectangle ++(1,1);
\end{tikzpicture}
}

\newcommand\lineTensorA[2]{
\begin{tikzpicture}[baseline=#1 ex, local bounding box=C,scale=#2]
\draw[fill=plusOneColour,scale=.7] (0,2) rectangle ++(1,1);
\draw[fill=minusOneColour,scale=.7] (0,1) rectangle ++(1,1);
\draw[fill=zeroColour,scale=.7] (0,0) rectangle ++(1,1);
\end{tikzpicture}
}

\newcommand\cubeTensorA[2]{
\begin{tikzpicture}[baseline=#1 ex, local bounding box=C,scale=#2]
\begin{scope}[3d view={110}{15}]

\mycube{minusOneColour}{\opacityColouredBlocksLayerTwo}{\distanceFactor*2}{\distanceFactor*2}{\distanceFactor*1}{.7};

\foreach \y in {0,...,2} 
	 \foreach \x in {2,...,0}
		 \mycubeNoBorder{zeroColour}{.3}{\distanceFactor*\x}{\distanceFactor*\y}{\distanceFactor*2};
\foreach \y in {0,...,2} 
	 \foreach \x in {2,...,0}
		 \mycubeNoBorder{zeroColour}{.4}{\distanceFactor*\x}{\distanceFactor*\y}{\distanceFactor*1};
\foreach \y in {0,...,2} 
	 \foreach \x in {2,...,0}
		 \mycubeNoBorder{zeroColour}{.7}{\distanceFactor*\x}{\distanceFactor*\y}{\distanceFactor*0};

\mycube{plusOneColour}{\opacityColouredBlocksLayerOne}{\distanceFactor*0}{\distanceFactor*0}{\distanceFactor*0}{.7};
\mycube{minusOneColour}{\opacityColouredBlocksLayerOne}{\distanceFactor*1}{\distanceFactor*0}{\distanceFactor*0}{.7};
\mycube{plusOneColour}{\opacityColouredBlocksLayerOne}{\distanceFactor*2}{\distanceFactor*0}{\distanceFactor*0}{.7};
\mycube{minusOneColour}{\opacityColouredBlocksLayerOne}{\distanceFactor*2}{\distanceFactor*1}{\distanceFactor*0}{.7};
\mycube{plusOneColour}{\opacityColouredBlocksLayerOne}{\distanceFactor*2}{\distanceFactor*2}{\distanceFactor*0}{.7};
\end{scope}
\end{tikzpicture}
}

\newcommand\rectangularTensorA[2]{
\begin{tikzpicture}[baseline=#1 ex, local bounding box=C,scale=#2]
\draw[fill=zeroColour,scale=.7] (0,2) rectangle ++(1,1);
\draw[fill=zeroColour,scale=.7] (1,2) rectangle ++(1,1);
\draw[fill=plusOneColour,scale=.7] (0,1) rectangle ++(1,1);
\draw[fill=zeroColour,scale=.7] (1,1) rectangle ++(1,1);
\draw[fill=zeroColour,scale=.7] (0,0) rectangle ++(1,1);
\draw[fill=zeroColour,scale=.7] (1,0) rectangle ++(1,1);
\end{tikzpicture}
}

\newcommand\squareTensorB[2]{
\begin{tikzpicture}[baseline=#1 ex, local bounding box=C,scale=#2]
\draw[fill=minusOneColour,scale=.7] (0,2) rectangle ++(1,1);
\draw[fill=zeroColour,scale=.7] (1,2) rectangle ++(1,1);
\draw[fill=plusOneColour,scale=.7] (2,2) rectangle ++(1,1);
\draw[fill=plusTwoColour,scale=.7] (0,1) rectangle ++(1,1);
\draw[fill=zeroColour,scale=.7] (1,1) rectangle ++(1,1);
\draw[fill=minusOneColour,scale=.7] (2,1) rectangle ++(1,1);
\draw[fill=minusOneColour,scale=.7] (0,0) rectangle ++(1,1);
\draw[fill=plusOneColour,scale=.7] (1,0) rectangle ++(1,1);
\draw[fill=zeroColour,scale=.7] (2,0) rectangle ++(1,1);
\end{tikzpicture}
}

\newcommand\squareTensorBNewProjection[2]{
\begin{tikzpicture}[baseline=#1 ex, local bounding box=C,scale=#2]
\draw[fill=plusOneColour,scale=.7] (0,2) rectangle ++(1,1);
\draw[fill=zeroColour,scale=.7] (1,2) rectangle ++(1,1);
\draw[fill=zeroColour,scale=.7] (2,2) rectangle ++(1,1);
\draw[fill=minusOneColour,scale=.7] (0,1) rectangle ++(1,1);
\draw[fill=zeroColour,scale=.7] (1,1) rectangle ++(1,1);
\draw[fill=zeroColour,scale=.7] (2,1) rectangle ++(1,1);
\draw[fill=plusOneColour,scale=.7] (0,0) rectangle ++(1,1);
\draw[fill=minusOneColour,scale=.7] (1,0) rectangle ++(1,1);
\draw[fill=zeroColour,scale=.7] (2,0) rectangle ++(1,1);
\end{tikzpicture}
}

\newcommand\cubeTensorB[2]{
\begin{tikzpicture}[baseline=#1 ex, local bounding box=C,scale=#2]
\begin{scope}[3d view={110}{15}]

\mycube{plusOneColour}{\opacityColouredBlocksLayerThree}{\distanceFactor*0}{\distanceFactor*0}{\distanceFactor*2}{.7};
\mycube{minusOneColour}{\opacityColouredBlocksLayerThree}{\distanceFactor*1}{\distanceFactor*0}{\distanceFactor*2}{.7};
\mycube{plusOneColour}{\opacityColouredBlocksLayerThree}{\distanceFactor*2}{\distanceFactor*0}{\distanceFactor*2}{.7};
\mycube{minusOneColour}{\opacityColouredBlocksLayerThree}{\distanceFactor*2}{\distanceFactor*1}{\distanceFactor*2}{.7};

\foreach \y in {0,...,2} 
	 \foreach \x in {2,...,0}
		 \mycubeNoBorder{zeroColour}{.3}{\distanceFactor*\x}{\distanceFactor*\y}{\distanceFactor*2};

\mycube{plusOneColour}{\opacityColouredBlocksLayerTwo}{\distanceFactor*2}{\distanceFactor*2}{\distanceFactor*1}{.7};
   
\foreach \y in {0,...,2} 
	 \foreach \x in {2,...,0}
		 \mycubeNoBorder{zeroColour}{.4}{\distanceFactor*\x}{\distanceFactor*\y}{\distanceFactor*1};
\foreach \y in {0,...,2} 
	 \foreach \x in {2,...,0}
		 \mycubeNoBorder{zeroColour}{.7}{\distanceFactor*\x}{\distanceFactor*\y}{\distanceFactor*0};

\mycube{minusTwoColour}{\opacityColouredBlocksLayerOne}{\distanceFactor*0}{\distanceFactor*0}{\distanceFactor*0}{.7};
\mycube{plusOneColour}{\opacityColouredBlocksLayerOne}{\distanceFactor*0}{\distanceFactor*2}{\distanceFactor*0}{.7};
\mycube{plusThreeColour}{\opacityColouredBlocksLayerOne}{\distanceFactor*1}{\distanceFactor*0}{\distanceFactor*0}{.7};
\mycube{minusOneColour}{\opacityColouredBlocksLayerOne}{\distanceFactor*1}{\distanceFactor*2}{\distanceFactor*0}{.7};
\mycube{minusTwoColour}{\opacityColouredBlocksLayerOne}{\distanceFactor*2}{\distanceFactor*0}{\distanceFactor*0}{.7};
\mycube{plusTwoColour}{\opacityColouredBlocksLayerOne}{\distanceFactor*2}{\distanceFactor*1}{\distanceFactor*0}{.7};
\mycube{minusOneColour}{\opacityColouredBlocksLayerOne}{\distanceFactor*2}{\distanceFactor*2}{\distanceFactor*0}{.7};

\end{scope}
\end{tikzpicture}
}

\newcommand\fourDimCrystal[2]{
\begin{tikzpicture}[baseline=#1 ex, local bounding box=C,scale=#2]
\begin{scope}[shift={(5,5,5)},3d view={110}{15},scale=.5]
\mycube{minusOneColour}{\opacityColouredBlocksLayerTwo}{\distanceFactor*2}{\distanceFactor*2}{\distanceFactor*1}{.7};

\foreach \y in {0,...,2} 
	 \foreach \x in {2,...,0}
		 \mycubeNoBorder{zeroColour}{.3}{\distanceFactor*\x}{\distanceFactor*\y}{\distanceFactor*2};
\foreach \y in {0,...,2} 
	 \foreach \x in {2,...,0}
		 \mycubeNoBorder{zeroColour}{.4}{\distanceFactor*\x}{\distanceFactor*\y}{\distanceFactor*1};
\foreach \y in {0,...,2} 
	 \foreach \x in {2,...,0}
		 \mycubeNoBorder{zeroColour}{.7}{\distanceFactor*\x}{\distanceFactor*\y}{\distanceFactor*0};

\mycube{plusOneColour}{\opacityColouredBlocksLayerOne}{\distanceFactor*0}{\distanceFactor*0}{\distanceFactor*0}{.7};
\mycube{minusOneColour}{\opacityColouredBlocksLayerOne}{\distanceFactor*1}{\distanceFactor*0}{\distanceFactor*0}{.7};
\mycube{plusOneColour}{\opacityColouredBlocksLayerOne}{\distanceFactor*2}{\distanceFactor*0}{\distanceFactor*0}{.7};
\mycube{minusOneColour}{\opacityColouredBlocksLayerOne}{\distanceFactor*2}{\distanceFactor*1}{\distanceFactor*0}{.7};
\mycube{plusOneColour}{\opacityColouredBlocksLayerOne}{\distanceFactor*2}{\distanceFactor*2}{\distanceFactor*0}{.7};
\end{scope}

\begin{scope}[shift={(3,3,3)},3d view={110}{15},scale=.75]
\foreach \y in {0,...,2} 
	 \foreach \x in {2,...,0}
		 \mycubeNoBorder{zeroColour}{.3}{\distanceFactor*\x}{\distanceFactor*\y}{\distanceFactor*2};
\foreach \y in {0,...,2} 
	 \foreach \x in {2,...,0}
		 \mycubeNoBorder{zeroColour}{.4}{\distanceFactor*\x}{\distanceFactor*\y}{\distanceFactor*1};
\foreach \y in {0,...,2} 
	 \foreach \x in {2,...,0}
		 \mycubeNoBorder{zeroColour}{.7}{\distanceFactor*\x}{\distanceFactor*\y}{\distanceFactor*0};
\end{scope}

\begin{scope}[3d view={110}{15},scale=1]
\mycube{plusOneColour}{\opacityColouredBlocksLayerThree}{\distanceFactor*0}{\distanceFactor*0}{\distanceFactor*2}{.7};
\mycube{minusOneColour}{\opacityColouredBlocksLayerThree}{\distanceFactor*1}{\distanceFactor*0}{\distanceFactor*2}{.7};
\mycube{plusOneColour}{\opacityColouredBlocksLayerThree}{\distanceFactor*2}{\distanceFactor*0}{\distanceFactor*2}{.7};
\mycube{minusOneColour}{\opacityColouredBlocksLayerThree}{\distanceFactor*2}{\distanceFactor*1}{\distanceFactor*2}{.7};

\foreach \y in {0,...,2} 
	 \foreach \x in {2,...,0}
		 \mycubeNoBorder{zeroColour}{.3}{\distanceFactor*\x}{\distanceFactor*\y}{\distanceFactor*2};

\mycube{plusOneColour}{\opacityColouredBlocksLayerTwo}{\distanceFactor*2}{\distanceFactor*2}{\distanceFactor*1}{.7};
   
\foreach \y in {0,...,2} 
	 \foreach \x in {2,...,0}
		 \mycubeNoBorder{zeroColour}{.4}{\distanceFactor*\x}{\distanceFactor*\y}{\distanceFactor*1};
\foreach \y in {0,...,2} 
	 \foreach \x in {2,...,0}
		 \mycubeNoBorder{zeroColour}{.7}{\distanceFactor*\x}{\distanceFactor*\y}{\distanceFactor*0};

\mycube{minusTwoColour}{\opacityColouredBlocksLayerOne}{\distanceFactor*0}{\distanceFactor*0}{\distanceFactor*0}{.7};
\mycube{plusOneColour}{\opacityColouredBlocksLayerOne}{\distanceFactor*0}{\distanceFactor*2}{\distanceFactor*0}{.7};
\mycube{plusThreeColour}{\opacityColouredBlocksLayerOne}{\distanceFactor*1}{\distanceFactor*0}{\distanceFactor*0}{.7};
\mycube{minusOneColour}{\opacityColouredBlocksLayerOne}{\distanceFactor*1}{\distanceFactor*2}{\distanceFactor*0}{.7};
\mycube{minusTwoColour}{\opacityColouredBlocksLayerOne}{\distanceFactor*2}{\distanceFactor*0}{\distanceFactor*0}{.7};
\mycube{plusTwoColour}{\opacityColouredBlocksLayerOne}{\distanceFactor*2}{\distanceFactor*1}{\distanceFactor*0}{.7};
\mycube{minusOneColour}{\opacityColouredBlocksLayerOne}{\distanceFactor*2}{\distanceFactor*2}{\distanceFactor*0}{.7};
\end{scope}

\end{tikzpicture}
}

\usepackage[noframe]{showframe}
\usepackage{framed}
\renewenvironment{shaded}{%
  \MakeFramed{\advance\hsize-\width \FrameRestore\FrameRestore}}%
 {\endMakeFramed}
\definecolor{shadecolor}{gray}{0.88}

\usepackage{caption}
\usepackage{subcaption}
\usepackage[export]{adjustbox}
\usepackage{hyperref} 
\hypersetup{colorlinks=true,citecolor=blue,linkcolor=blue,urlcolor=blue}

\newcommand{\YES}{\textsc{Yes}}
\newcommand{\NO}{\textsc{No}}
\newcommand{\Test}[2]{
\def\temp{#2}\ifx\temp\empty
  \operatorname{Test}_{#1}
\else
  \operatorname{Test}_{#1}^{#2}
\fi
}

\newcommand{\ignore}[1]{}
\newcommand{\A}{\mathbf{A}}
\newcommand{\B}{\mathbf{B}}

\newcommand{\GG}{\mathbf{G}}
\newcommand{\HH}{\mathbf{H}}
\newcommand{\K}{\mathbf{K}}

\newcommand{\X}{\mathbf{X}}

\newcommand{\SSS}{\mathbf{S}}

\newcommand{\Vset}{V}
\newcommand{\Eset}{E}

\newcommand{\N}{\mathbb{N}}

\newcommand{\Q}{\mathbb{Q}}
\newcommand{\Z}{\mathbb{Z}}

\newcommand{\cT}{\mathcal{T}}
\newcommand{\freeM}{\mathbb{F}_{\Mminion}}

\newcommand{\freeQ}{\mathbb{F}_{\Qconv}}
\newcommand{\freeZ}{\mathbb{F}_{\Zaff}}
\newcommand{\freeBA}{\mathbb{F}_{\BAminion}}
\newcommand{\Xk}{\X^\tensor{k}}
\newcommand{\Ak}{\A^\tensor{k}}
\newcommand{\Bk}{\B^\tensor{k}}

\renewcommand{\vec}[1]{\mathbf{#1}}
\newcommand{\ba}{\vec{a}}
\newcommand{\bb}{\vec{b}}
\newcommand{\bc}{\vec{c}}
\newcommand{\bd}{\vec{d}}

\newcommand{\bh}{\vec{h}}
\newcommand{\bi}{\vec{i}}
\newcommand{\bj}{\vec{j}}
\newcommand{\bbm}{\vec{m}}
\newcommand{\bn}{\vec{n}}

\newcommand{\bq}{\vec{q}}
\newcommand{\br}{\vec{r}}
\newcommand{\bs}{\vec{s}}
\newcommand{\bu}{\vec{u}}
\newcommand{\bx}{\vec{x}}
\newcommand{\bv}{\vec{v}}
\newcommand{\by}{\vec{y}}
\newcommand{\bw}{\vec{w}}
\newcommand{\be}{\vec{e}}
\newcommand{\bell}{{\ensuremath{\boldsymbol\ell}}}
\newcommand{\bk}{\vec{k}}
\newcommand{\bz}{\vec{z}}
\newcommand{\bepsilon}{{\bm{\epsilon}}}

\newcommand{\balpha}{{\bm{\alpha}}}
\newcommand{\bbeta}{{\bm{\beta}}}

\DeclareMathOperator{\tr}{tr}

\DeclareMathOperator{\Bmat}{B}
\DeclareMathOperator{\Amat}{A}

\DeclareMathOperator{\BLP}{BLP}
\DeclareMathOperator{\AIP}{AIP}
\DeclareMathOperator{\IP}{IP}
\DeclareMathOperator{\BA}{BA}
\DeclareMathOperator{\SDP}{SDP}
\DeclareMathOperator{\PCSP}{PCSP}
\DeclareMathOperator{\CSP}{CSP}
\DeclareMathOperator{\PCSPs}{PCSPs}
\DeclareMathOperator{\CSPs}{CSPs}
\DeclareMathOperator{\parCSPs}{(P)CSPs}

\DeclareMathOperator{\supp}{supp}
\DeclareMathOperator{\set}{set}

\DeclarePairedDelimiter{\floor}{\lfloor}{\rfloor}

\newcommand{\Mminion}{\ensuremath{{\mathscr{M}}}}

\newcommand{\Qconv}{\ensuremath{{\mathscr{Q}_{\operatorname{conv}}}}}
\newcommand{\Zaff}{\ensuremath{{\mathscr{Z}_{\operatorname{aff}}}}}

\newcommand{\BAminion}
{\ensuremath{{\mathscr{M}_{\operatorname{\BA}}}}}

\newcommand{\bone}{\mathbf{1}}  
\newcommand{\bzero}{\mathbf{0}} 
\newcommand{\tensor}[1]{\textsuperscript{\raisebox{-.5pt}{\normalfont\textcircled{\raisebox{-.1pt}{\tiny #1}}}}}
\newcommand\ang[1]{{\ensuremath\langle #1\rangle}}

\newcommand\cont[1]{\overset{\tiny#1}{\ast}}

\theoremstyle{plain}
\newtheorem{thm}{Theorem}
\newtheorem*{thm*}{Theorem}
\newtheorem{lem}[thm]{Lemma}
\newtheorem*{lem*}{Lemma}
\newtheorem{prop}[thm]{Proposition}
\newtheorem*{prop*}{Proposition}
\newtheorem{cor}[thm]{Corollary}
\newtheorem*{cor*}{Corollary}

\theoremstyle{definition}
\newtheorem{defn}[thm]{Definition}
\newtheorem*{defn*}{Definition}
\newtheorem{rem}[thm]{Remark}
\newtheorem{example}[thm]{Example}

\newcommand{\lemeq}[1]
{
\ensuremath{\stackrel{\operatorname{L}.\ref{#1}}{\;\;=\;\;}}
}
\newcommand{\propparteq}[2]
{
\ensuremath{\stackrel{\operatorname{P}.\ref{#1}\eqref{#2}}{\;\;=\;\;}}
}

\newcommand{\equationeq}[1]
{
\ensuremath{\stackrel{\eqref{#1}}{\;\;=\;\;}}
}
\newcommand{\spaceeq}
{
\ensuremath{\stackrel{}{\;\;=\;\;}}
}
\newcommand{\spaceneq}
{
\ensuremath{\stackrel{}{\;\;\neq\;\;}}
}

\begin{document}

\title{Approximate Graph Colouring\\ and the Crystal with a Hollow Shadow\thanks{Two extended abstracts of different parts of this work appeared in the Proceedings of the 2023 ACM-SIAM Symposium on Discrete Algorithms (SODA'23)~\cite{cz23soda:aip} and in the Proceedings of the 2023 ACM Symposium on Theory of Computing (STOC'23)~\cite{cz23stoc:ba}, respectively.
This research was funded in whole by UKRI EP/X024431/1. For the purpose of Open Access, the authors have applied a CC BY public copyright licence to any Author Accepted Manuscript version arising from this submission. All data is provided in full in the results section of this paper.}}

\author{Lorenzo Ciardo\\
University of Oxford\\
\texttt{lorenzo.ciardo@cs.ox.ac.uk} 
\and 
Stanislav {\v{Z}}ivn{\'y}\\
University of Oxford\\
\texttt{standa.zivny@cs.ox.ac.uk}
}

\date{\today}
\maketitle

\begin{abstract}
\noindent We show that approximate graph colouring is not solved by
the lift-and-project hierarchy for the combination of linear programming and  linear Diophantine equations.
 The proof is based on combinatorial tensor theory.
\end{abstract}

\section{Introduction}
\label{sec:intro}

\noindent The \emph{approximate graph colouring} problem (AGC) consists in finding a $d$-colouring of a given
$c$-colourable graph, where $3\leq c\leq d$. There is a huge gap in our
understanding of this problem.
For an $n$-vertex graph and $c=3$, the best known polynomial-time
algorithm of Kawarabayashi, Thorup, and Yoneda~\cite{KawarabayashiTY24} finds a $d$-colouring with $d=\tilde O(n^{0.19747})$, building on a long line of works started by Wigderson~\cite{Wigderson83:jacm}.
It was conjectured by Garey and Johnson~\cite{GJ76} that the problem is NP-hard for any fixed constants
$3\leq c\leq d$ even in the decision variant: Given a graph, output $\YES$ if
it is $c$-colourable and output $\NO$ if it is not $d$-colourable.

For $c=d$, the problem becomes the classic $c$-colouring problem, which appeared
on Karp's original list of $21$ NP-complete problems~\cite{Karp72}. The case
$c=3$, $d=4$ was only proved to be NP-hard in 2000 by Khanna, Linial, and Safra~\cite{KhannaLS00}
(and a simpler proof was given 
by Guruswami and Khanna
in~\cite{GK04}); more generally, \cite{KhannaLS00} showed hardness of the case
$d=c+2\floor{c/3}-1$. This was improved to $d=2c-2$ in 2016 by Brakensiek and Guruswami~\cite{BrakensiekG16}, and recently to $d=2c-1$ by Barto, Bul\'in, Krokhin, and Opr\v{s}al~\cite{BBKO21}. In particular, this last result implies
hardness of the case $c=3$, $d=5$; the complexity of the case $c=3$, $d=6$ is
still open. Building on the work of Khot~\cite{Khot01} and Huang~\cite{Huang13},
Krokhin, Opr\v{s}al, Wrochna, and \v{Z}ivn\'{y} established 
NP-hardness for  $d={c\choose\floor{c/2}}-1$
for $c\geq 4$ in~\cite{KOWZ22}.
NP-hardness of
AGC was established for all constants $3\leq c\leq d$ 
by Dinur, Mossel, and Regev
in~\cite{Dinur09:sicomp} under a non-standard variant of
the Unique Games Conjecture, 
by Guruswami and Sandeep
in~\cite{GS20:icalp}
under the $d$-to-1 conjecture~\cite{Khot02stoc} for any fixed $d$, and (an even stronger statement of distinguishing 3-colourability from not having an independent set of significant size)
by Braverman, Khot, Lifshitz, and Minzer
in~\cite{Braverman21:focs} under the rich $2$-to-$1$ conjecture of Braverman, Khot, and Minzer~\cite{Braverman21:itcs}. Conditional to suitable strengthened versions of the UGC, Dinur and Shinkar proved NP-hardness in a $4$ vs. superconstant regime in~\cite{dinur2010conditional}. 

AGC is a prominent example of so called \emph{Promise Constraint Satisfaction
Problems} ($\PCSPs$), which we define next.
A \emph{directed graph} (\emph{digraph}) $\A$ consists of a set $\Vset(\A)$ of elements
called \emph{vertices} and a set $\Eset(\A)\subseteq \Vset(\A)^{2}$ of pairs of vertices called \emph{edges}.
Given two digraphs $\A$ and $\B$, a map $f:\Vset(\A)\to\Vset(\B)$ is a
\emph{homomorphism} from $\A$ to $\B$ if $(f(u),f(v))\in \Eset(\B)$ for any
$(u,v)\in\Eset(\A)$. We shall indicate the existence of a homomorphism from $\A$ to $\B$ by writing $\A\to\B$.
Let $\A$ and $\B$ be two fixed finite digraphs with $\A\to\B$; we call the pair $(\A,\B)$ a \emph{template}.
The $\PCSP$ parameterised by the template $(\A,\B)$, denoted by $\PCSP(\A,\B)$,
is the following decision problem: Given a finite digraph $\X$ as input, answer
$\YES$ if $\X\to\A$ and $\NO$ if $\X\not\to\B$.\footnote{The requirement $\A\to\B$ implies that the two cases cannot happen simultaneously, as homomorphisms compose; the \emph{promise} is that one of the two cases always happens.}
A $p$-colouring of a digraph $\X$ is precisely a homomorphism from $\X$ to the \emph{clique} $\K_p$---i.e., the digraph on vertex set $\{1,\dots,p\}$ such that any pair of distinct vertices is a (directed) edge. Hence, AGC is $\PCSP(\K_c,\K_d)$.
It is customary to study $\parCSPs$ on more general objects than digraphs, known as \emph{relational structures}, which consist of a collection of  relations of arbitrary arities on a vertex set, cf.~\cite{BBKO21}.

By letting $\A=\B$ in the definition of a $\PCSP$, one obtains the standard (non-promise) \emph{Constraint Satisfaction Problem} ($\CSP$)~\cite{Feder98:monotone}.
$\PCSPs$ were introduced by Austrin, Guruswami, and H{\aa}stad~\cite{AGH17} and Brakensiek and Guruswami~\cite{BG21} 
as a general framework for studying approximability of perfectly satisfiable $\CSPs$ and have
emerged as a new exciting  direction in constraint satisfaction that requires different
techniques than $\CSPs$. Recent works on $\PCSPs$ include
those using analytical methods~\cite{Bhangale21:stoc,Braverman21:itcs,BGS23,Bhangale22:stoc}
and those building on algebraic methods~\cite{BG19,bgwz20,GS20:icalp,
AB21,Barto21:stacs,BWZ21,Butti21:mfcs,Barto22:soda, cz23sicomp:clap,NZ22}
developed in~\cite{BBKO21}. However, most basic questions are still wide open, including complexity classifications and applicability of different types of algorithms.

Two main algorithmic techniques have been utilised for solving CSPs and their variants: enforcing (some type of) \emph{local consistency}, and solving (generalisations of) \emph{linear equations}.
The first type of algorithms divides a given CSP into multiple small CSPs, each
of which requires meeting \emph{local} constraints on a portion of the instance
of bounded size, and then enforces \emph{consistency} between all solutions
(called partial homomorphisms); i.e., it requires that solutions should agree on the intersection of their domains. Instead, the second type of algorithms seeks a \emph{global} solution that satisfies 
a \emph{linearised} version of the constraints.
More precisely, it is always possible to formulate a CSP (and, in fact, any homomorphism problem) as a system of linear equations over $\{0,1\}$; then, the algorithms of the second type work by suitably modifying the system (in particular, extending the domain of its variables) in a way that it can be efficiently solved through variants of Gaussian elimination.

Remarkably, all algorithms hitherto proposed in the literature on (variants of) CSPs can be broadly classified as instances of one of the two aforementioned techniques, or a combination of both. 
A primary example of the first type is the \emph{bounded width} algorithm, which outputs $\YES$ if and only if a consistent collection of partial homomorphisms exists~\cite{Feder98:monotone}.
More powerful versions of the local consistency technique require that the
partial homomorphisms should be sampled according to a probability distribution
(which results in the \emph{Sherali--Adams LP} hierarchy~\cite{Sherali1990}), and
that the probabilities should be treated as vectors satisfying certain orthogonality requirements (which gives the \emph{sum-of-squares} or \emph{Lasserre SDP} hierarchy~\cite{shor1987class,parrilo2000structured,Lasserre02}).
As for the second type, the linear-system formulation of a CSP can be
efficiently solved in $\Z$ by computing the Hermite or the Smith canonical forms
of the corresponding coefficient matrix~\cite{MR874114}; this results in the
\emph{affine integer programming} ($\AIP$) relaxation (also known as the system
of \emph{linear Diophantine equations}), studied in the context of PCSPs in~\cite{BG21,BBKO21}. 

Neither of the two techniques, alone, is powerful enough to solve all tractable
CSPs, even in the non-promise variant and on Boolean domains.
In fact, the elusive interaction between consistency-checking methods and linear
equations for non-Boolean CSPs was the major obstacle to the proof of the Feder--Vardi dichotomy conjecture~\cite{Feder98:monotone}, 
finally settled independently by Bulatov~\cite{Bulatov17:focs} and by
Zhuk~\cite{Zhuk17_FOCS,Zhuk20:jacm}. 
Hence, efforts have been directed to \emph{blending the two techniques}, in order to design a stronger \emph{local-global} algorithm~\cite{BG19,bgwz20,BhangaleKM24,cnp24:stoc,CiardoZ24}. 
In~\cite{bgwz20}, Brakensiek, Guruswami, Wrochna, and \v{Z}ivn\'y proposed an algorithm that combines the first level of the Sherali--Adams LP hierarchy (known as the \emph{basic linear programming} ($\BLP$) relaxation) with the $\AIP$ relaxation. Remarkably, that algorithm, which we call $\BA$ in this paper, solves all tractable cases of Schaefer's dichotomy of Boolean CSPs~\cite{Schaefer78:stoc}, as proved in~\cite{bgwz20}.
While the $\BA$ algorithm admits a characterisation in terms of polymorphic identities and, thus, the class of (P)CSPs solved by it is well understood~\cite{bgwz20}, the power of the hierarchy\footnote{A hierarchy similar to the $\BA$ hierarchy from this paper was considered by Berkholz and Grohe~\cite{Berkholz17:soda} in the context of the graph isomorphism problem.} built on top of $\BA$ is still unknown, even for non-promise CSPs. Very recently, Lichter and Pago have constructed the first example of a tractable, finite-domain CSP that is not solved by any constant level of such hierarchy~\cite{lichter2024limitations}.

Since polynomial-time algorithms are not expected to solve NP-hard problems, a well-established line
of work has sought lower bounds on the efficacy of these 
algorithms; see~\cite{Arora06:toc,Braun15:stoc,Chan16:jacm-lp,Kothari22:sicomp,Ghosh18:toc}
for lower bounds on LPs arising from lift-and-project hierarchies such as that of Sherali--Adams,~\cite{Tulsiani09:stoc,Lee15:stoc,Chan15:jacm} for lower bounds on SDPs, and~\cite{Berkholz17:soda} for lower bounds on linear Diophantine equations. 
If, as conjectured by Garey and Johnson~\cite{GJ76}, AGC is NP-hard and P$\neq$NP, neither of the two algorithmic techniques discussed above (nor their blend) should be able to solve it. 
In a striking sequence of
works 
by Dinur, Khot, Kindler, Minzer, and Safra~\cite{Khot17:stoc-independent,Dinur18:stoc-non-optimality,Dinur18:stoc-towards,Khot18:focs-pseudorandom},
the 2-to-2 conjecture of Khot~\cite{Khot02stoc} (with imperfect completeness)
was resolved. As detailed in~\cite{Khot18:focs-pseudorandom}, this implies
(together with~\cite{GS20:icalp}) that AGC is not solved by the sum-of-squares hierarchy (and, as a consequence, by the weaker Sherali--Adams LP and bounded width hierarchies). That lower bound is obtained by transferring known sum-of-squares integrality gaps for linear equations mod $2$~\cite{Grigoriev01:tcs,Schoenebeck08} to AGC. Since linear equations are solved by $\AIP$, the reduction from~\cite{Khot17:stoc-independent,Dinur18:stoc-non-optimality,Dinur18:stoc-towards,Khot18:focs-pseudorandom} cannot be used to produce lower bounds against AIP-based algorithms.

\paragraph{Contributions}
We prove that AGC is not solved by the $\BA$ hierarchy. 
This substantially extends the state of the art on non-solvability of AGC. In particular, our result directly implies non-solvability of AGC by the $\AIP$ hierarchy and gives a new proof of non-solvability by the Sherali--Adams LP hierarchy, as both of these hierarchies are weaker than $\BA$.

Ruling out the first level of the $\BA$ hierarchy is trivial
using the characterisation from~\cite{bgwz20}, while the task is significantly more challenging for higher levels. 
The core of our proof is geometric. Using the framework recently developed by the authors in~\cite{cz23soda:minions} to study algorithmic hierarchies, we reduce the problem of finding a ``fooling instance'' for the $\BA$ hierarchy applied to AGC to the geometric problem of building a \emph{hollow-shadowed crystal}; i.e., a high-dimensional integral tensor whose projections onto hyperplanes of low dimension are equal up to reflection (i.e., up to permutations of the tensor modes; we call such a tensor a \emph{crystal}) and satisfy a sparsity condition dictating that certain entries should be set to zero (in this case, we say that the crystal has a \emph{hollow shadow}). The main technical result of this work is a constructive proof of the existence of tensors having these features.

Our construction consists of two phases. The first phase concerns the existence of crystals (regardless of the hollowness requirement). We perform this task by providing a complete combinatorial characterisation for \emph{realisable systems of shadows}; i.e., for those collections of low-dimensional tensors that can be realised as the projections of a single high-dimensional tensor. As detailed in the conference version~\cite{cz23soda:aip}, this construction is sufficient to prove non-solvability of AGC by the $\AIP$ hierarchy.
To prove the analogous result for the stronger $\BA$ hierarchy, we need to deal with the problem of enforcing hollowness of the shadow of a given crystal. This is accomplished in the second phase of our construction (extending the conference version~\cite{cz23stoc:ba}), which consists in applying local modifications to a tensor through certain crystals that we call \emph{quartzes}.

Two-dimensional variants of this problem have appeared in the literature in
combinatorial matrix theory. The problem of recovering a matrix (i.e., a
two-dimensional tensor) from its row- and column-sum vectors (i.e.,
one-dimensional projections) has been studied for different classes of matrices,
such as nonnegative integral matrices~\cite{brualdi1991combinatorial}, $0$--$1$ matrices~\cite{da20090,ryser1957combinatorial}, alternating-sign matrices~\cite{MR1392498}, and sign-restricted matrices~\cite{brualdi2021sign}, see also the survey~\cite{MR2890890}.
Moreover, an active research trend in combinatorial matrix theory investigates
the conditions for the existence of matrices over a certain domain having
prescribed row and column sums and a fixed \emph{pattern}, i.e., a fixed set of
entries allowed (or required) to be nonzero. Examples include $0$--$1$ matrices with zero trace (i.e., adjacency matrices of digraphs)~\cite{MR120166}, with at most one fixed zero in each column~\cite{anstee1982properties}, or with a fixed zero block~\cite{MR1997370}, 
real matrices with a fixed pattern~\cite{MR1758207}, and
integral matrices with fixed lower and upper bounds on each entry~\cite{MR1184987}; see also related work in~\cite{MR2444364,MR2446674,MR3551626}. 

To the best of our knowledge, the problem of reconstructing a tensor from low-dimensional projections has 
hitherto 
only been studied for matrices (but cf.~\cite{MR3759214}, where a related problem is investigated in three dimensions in the restricted setting of alternating-sign three-dimensional tensors). In order to rule out solvability of AGC for all numbers of colours, we need to build crystals of arbitrarily high dimension 
and hence approach the reconstruction problem for arbitrarily high-dimensional tensors. 
In addition to its application
to AGC, we believe that our result might be of independent interest to the linear algebra and tensor theory communities.
Furthermore, within complexity theory,
we expect that our method will be useful more broadly in bringing new insights into the power of algorithmic techniques that blend the consistency and the linear equation approaches---which are gaining much prominence in the wider context of CSPs and
PCSPs~\cite{Dalmau24:lics,bgwz20,OConghaileC22:mfcs,BG19,BhangaleKM24,cnp24:stoc,CiardoZ24}.
The geometric method we develop in the current work appears to be particularly well-suited for capturing the essence of such algorithms.  

\section{Overview of results and techniques}
\label{sec_overview}

\noindent 
Let $\X$ and $\A$ be two digraphs.
We can cast the question ``Is there a homomorphism from $\X$ to $\A$?'' as the question of checking whether a system of linear equations (over, say, $\Q$) has a solution in the set $\{0,1\}$. Indeed, introduce variables $\lambda_{x,a}$ for all vertices $x\in V(\X),a\in V(\A)$, and variables
$\mu_{\by,\bb}$ for all edges $\by\in E(\X),\bb\in E(\A)$, and consider the equations
\vspace*{-1.25em}
\begin{align}
\label{eqns_BLP}
\tag{IP}
\begin{array}{llllll}\\
\mbox{($\IP_1$)}\qquad &\displaystyle\sum_{a\in V(\A)} \lambda_{x,a}\ &=\ 1 & \hspace*{1cm} &\forall x\in V(\X)\\[15pt]
\mbox{($\IP_2$)}\qquad &\displaystyle\sum_{\substack{\bb\in E(\A)\\ b_i=a}} \mu_{\by,\bb}\ &=\ \lambda_{y_i,a} & \hspace*{1cm} &\forall \by\in E(\X),\; i\in \{1,2\},\; a\in V(\A).
\end{array}
\end{align}
\vspace*{-1em}

\noindent One readily checks that $\X\to\A$ if and only if \eqref{eqns_BLP} has a solution in $\{0,1\}$. Unless P=NP, this system is not solvable in polynomial time over $\{0,1\}$, in general. Relaxing it by allowing that the variables can be assigned rational nonnegative values results in the so-called \emph{basic linear programming} (BLP) relaxation. 
Similarly, allowing that the variables can be assigned integer values yields the~\emph{affine integer programming} (AIP) relaxation.
The $\BA$ relaxation described in~\cite{bgwz20} combines $\BLP$ and $\AIP$. More concretely, it outputs $\YES$ if and only if there exist a solution to $\BLP$ and a solution to $\AIP$ such that the following \emph{refinement condition} holds: Whenever a variable is zero in the $\BLP$ solution, it is zero in the $\AIP$ solution. It follows that $\BA$ is at least as strong as both $\BLP$ and $\AIP$; in fact, as shown in~\cite{bgwz20}, it is strictly stronger, in the sense that there exist templates
that are solved by $\BA$ but not by $\BLP$ or $\AIP$. 
Note that the three relaxations mentioned above result in algorithms that are complete but not necessarily sound, in the sense that they always output $\YES$ if $\X\to\A$, but may fail to output $\NO$ if $\X\not\to\A$.

The system~\eqref{eqns_BLP} can be refined by replacing the variables
$\lambda_{x,a}$ with variables $\lambda_{S,f}$, where $S$ is a set of vertices
of $\X$ of size at most $k$ and $f$ is a function from $S$ to $\Vset(\A)$.
Solving such refined system over the set of nonnegative rational numbers (integers) would then mean finding rational nonnegative (integer) distributions over the set of partial assignments from portions of the instance of size at most $k$ to $\A$. The former choice results in the Sherali--Adams LP hierarchy~
\cite{Sherali1990}, 
which we call the $\BLP$ hierarchy; the latter results in the affine integer programming hierarchy~\cite{cz23soda:aip}, which we call the $\AIP$ hierarchy.
Similarly, the $\BA$ hierarchy we consider in this work consists in applying the $\BA$ relaxation of~\cite{bgwz20} to progressively larger portions of the instance, in the same spirit as the $\BLP$ and $\AIP$ hierarchies. Equivalently, the $\BA$ hierarchy can be described as follows: Its $k$-th level, applied to two digraphs $\X$ and $\A$, outputs $\YES$ if and only if $(i)$ the $k$-th level of both $\BLP$ and $\AIP$ outputs $\YES$ when applied to $\X$ and $\A$, and $(ii)$ the two solutions they provide satisfy the refinement condition~\cite{cz23soda:minions}. In this case, we write $\BA^k(\X,\A)=\YES$. Given two digraphs $\A,\B$ such that $\A\to\B$, we say that the $k$-th level of $\BA$ \emph{solves} $\PCSP(\A,\B)$ if, for any instance $\X$, $\BA^k(\X,\A)=\YES$ implies $\X\to\B$. (The definition for the $\BLP$ and $\AIP$ hierarchies is analogous.) Note that, if $\PCSP(\A,\B)$ is solved by some level of the $\BLP$ or $\AIP$ hierarchies, then it is also solved by the same level of the $\BA$ hierarchy.

These three hierarchies are complete but not necessarily sound, and they become progressively stronger as the level $k$ increases.
Crucially, the $\BA$ hierarchy (and, in fact, already the weaker $\BLP$ hierarchy) ensures local consistency, in the sense that each assignment receiving nonzero weight corresponds to a partial homomorphism. Equivalently, the $\BA$ hierarchy is at least as strong as the bounded-width algorithm\footnote{More precisely, the $k$-th level of the $\BA$ (or $\BLP$) hierarchy is at least as strong as the $k$-th level of the bounded-width algorithm.}~\cite{Feder98:monotone,Barto14:jacm,Barto14:jloc} (and, in fact, strictly stronger, see~\cite{Atserias22:soda}). In particular, the $\BA$ hierarchy is \emph{sound in the limit}, in the sense that its $k$-th level correctly classifies all instances of size $k$ or less---which is clear from the fact that a partial homomorphism over the whole domain is a homomorphism.
The same is not true for the $\AIP$ hierarchy.

The main result of our work is that no constant level of the $\BA$ hierarchy solves the approximate graph colouring problem.

\begin{shaded}
\vspace{-.4cm}
\begin{thm}
\label{thm_BLPAIPk_no_solves_AGC}
For any fixed $3\leq c\leq d$, there is no $k\in\N$ such that 
$\BA^k$
solves $\PCSP(\K_c,\K_d)$.  
\end{thm}
\vspace{-.4cm}
\end{shaded}

A way to prove that approximate graph colouring is not solved by the $\BA$ hierarchy is to present \emph{fooling instances}---digraphs with a large chromatic number but yet whose structure meets all constraints of the hierarchy. More precisely, it suffices to build, for every $c$, $d$, and $k$, a digraph $\GG$ whose chromatic number is higher than $d$ and such that $\BA^k(\GG,\K_c)=\YES$. Thus, the high-level description of our strategy is:
\begin{center}
\emph{``Find a fooling instance for the $\BA$ hierarchy applied to AGC.''}
\end{center}

Instead of directly looking for instances that fool the hierarchy, our approach
shall be to consider the following questions: What does a \emph{certificate of
acceptance} for the $\BA$ hierarchy look like? Can we tell, from the shape of
such a certificate, what the \emph{limits} of the hierarchy applied to AGC are? The first step of our analysis is to translate the problem of whether the $\BA$ hierarchy accepts an input into a problem having a different, \emph{multilinear} nature.
Building on the framework developed in~\cite{cz23soda:minions}, we find that $\BA$ acceptance is implied by the existence of a family of \emph{tensors} having certain special characteristics. First of all, they need to satisfy $(i)$ a \emph{system of symmetries}.
At a high level, this requirement results from the marginality constraints that are enforced by all ``lift-and-project'' hierarchies such as the $\BLP$, $\AIP$, and Lasserre $\SDP$ hierarchies~\cite{Laurent03}, and is common to all algorithmic hierarchies studied in~\cite{cz23soda:minions} through the tensor approach. 
There is, however, a feature that is unique to the $\BA$ hierarchy.
Not only does $\BA$ require that both a linear program and a system of Diophantine equations have a solution; it also requires that any variable that is assigned zero weight by the former should be assigned zero weight by the latter. 
The translation of this refinement condition into the multilinear framework is $(ii)$ a \emph{hollowness} requirement: Each tensor certifying $\BA$ acceptance needs to be hollow; i.e., it needs to contain zeros in certain prescribed entries. In sum, the original problem has now become the following:
\begin{center}
\emph{``Produce a family of hollow tensors satisfying a system of symmetries.''}
\end{center}

There is a natural way to produce a family $\{T_i\}$ of tensors satisfying such symmetries: One starts with a high-dimensional tensor $C$ whose low-dimensional oriented projections (i.e., projections onto oriented hyperplanes) are equal. Then, the family of \emph{all} (not necessarily oriented) low-dimensional projections of $C$ satisfies the required symmetries. We call such a tensor $C$ a \emph{crystal}, while the \emph{shadow} of $C$ is any of its oriented projections. We then reformulate the problem 
to its final form; the solution of this problem is the main technical result of the paper.
\begin{center}
\emph{``Find a crystal whose shadow is hollow.''}
\end{center}

\paragraph{Organisation of the paper}
The rest of the article is conceptually organised in three parts, each corresponding to a different phase of the proof of Theorem~\ref{thm_BLPAIPk_no_solves_AGC}: 
    (1) a \emph{pre-processing} phase, where $\BA^k$ acceptance is turned into a multilinear problem;
    (2) a \emph{multilinear} phase, where the multilinear problem is solved (i.e., hollow-shadowed crystals are built);
    (3) a \emph{post-processing} phase, where the solution of the multilinear problem is translated back to the algorithmic framework, and it is used to recover a fooling instance.
Full details of the three phases are discussed in Sections~\ref{sec_BA_hierarchy_through_tensors},~\ref{sec_crystals}, and~\ref{sec_fooling_BA}, respectively, after providing some preliminaries in Section~\ref{sec_preliminaries}. Sections~\ref{subsec_BA_hierarchy_through_tensors},~\ref{subsec_crystals}, and~\ref{subsec_fooling_BA} below give a more intuitive overview of the contents of each of them.

\subsection{The BA hierarchy through tensors}
\label{subsec_BA_hierarchy_through_tensors}
All relaxation algorithms hitherto studied for (promise) CSPs, including the
$\BLP$, $\AIP$, and $\BA$ algorithms, have an algebraic counterpart described through the
notion of \emph{linear minion}---an algebraic structure consisting of a set of matrices
that is closed
under the application of elementary row operations (summing up or swapping two
rows, inserting an extra zero row). Given a linear minion $\Mminion$ and a digraph
$\A$ with $n$ vertices and $m$ edges, there exists a natural way of simulating
the structure of $\A$ in $\Mminion$, by defining a new (potentially infinite)
digraph $\freeM(\A)$ (the \emph{free structure} of $\Mminion$ generated by $\A$)
whose vertices are the matrices in $\Mminion$ having $n$ rows and whose edges
are pairs of matrices $(M,N)$ such that both $M$ and $N$ can be obtained from
some matrix $Q$ having $m$ rows through certain elementary row operations induced by the edges of $\A$. Then,
the relaxation induced by $\Mminion$ works as follows: Given an instance
$\X$, rather than directly checking whether $\X\to\A$, one checks whether
$\X\to\freeM(\A)$. The advantage is that, for certain linear minions, the latter can be tested in polynomial time, even when the former cannot. As an example, stochastic rational vectors form a linear minion (since they are preserved under elementary row operations) named $\Qconv$, whose corresponding relaxation is $\BLP$. Similarly, integer vectors whose entries sum up to $1$ form the linear minion $\Zaff$ corresponding to $\AIP$. By combining the two linear minions $\Qconv$ and $\Zaff$ in a suitable way, one obtains the linear minion $\BAminion$ corresponding to $\BA$.

The framework developed in~\cite{cz23soda:minions} allows to systematically strengthen the relaxation corresponding to any linear minion, by making use of the notion of \emph{tensor power} of a digraph: For $k\in\N$, the $k$-th tensor power of $\A$ is the \emph{hypergraph} 
$\Ak$ whose vertices are $k$-tuples of vertices of $\A$, and whose hyperedges are
$k$-dimensional tensors obtained by ``scattering'' the edges of $\A$ in $k$
dimensions. The $k$-th level of the hierarchy of the relaxation corresponding to
some linear minion $\Mminion$ essentially amounts to applying the relaxation to
the \emph{tensorised} digraphs rather than the original digraphs; in other
words, checking if there exists a homomorphism
$\Xk\to\freeM(\Ak)$.\footnote{We note that $\freeM(\Ak)$ is a hypergraph rather than a digraph; the definition is analogous.}
In addition, the homomorphism needs to 
preserve the tensor structure of the two hypergraphs (intuitively, it must ``behave well with respect to projections'')---in which case, we say that it is a \emph{$k$-tensorial} homomorphism.
The algorithm obtained in this way is progressively stronger as $k$ increases, and it still runs in polynomial time (for a fixed $k$) since the tensorised digraph can be constructed in polynomial time and its size is polynomial in the size of the original digraph. In particular, if the matrices in $\Mminion$ satisfy a certain positivity
requirement---in which case we say that the linear minion is \emph{conic}---the hierarchy is sound in the limit, as its $k$-th level correctly classifies all instances $\X$ on at most $k$ vertices. In fact, 
the hierarchies based on conic minions enforce local consistency. Crucially, the linear minions $\Qconv$ and $\BAminion$ are conic, while the linear minion $\Zaff$ is not~\cite{cz23soda:minions}.

It was shown in~\cite{cz23soda:minions} that the $\BA$ hierarchy---as well as the $\BLP$, $\AIP$, and other algorithmic hierarchies---fits within this framework: The fact that $\BA^k(\X,\A)=\YES$ is equivalent to the existence of a $k$-tensorial homomorphism from $\Xk$ to $\freeBA(\Ak)$. Moreover, it follows from the structure of $\BAminion$ that any such homomorphism can be \emph{decoupled} into a homomorphism $\xi$ to $\freeQ(\Ak)$ and a homomorphism $\zeta$ to $\freeZ(\Ak)$ (cf.~Theorem~\ref{thm_acceptance_BA_hierarchy_general}). If $\A$ is a clique---as it happens when the $\BA$ hierarchy is applied to AGC---one can design a simpler sufficient criterion, based on the fact that one may always assume $\xi$ to be the homomorphism mapping a tuple of vertices of $\X$ to a tensor in $\freeQ(\Ak)$ that is uniform on its support. After dealing with some combinatorial technicalities, this fact produces the following criterion of acceptance. (In the statement below, $E_\ba\ast\zeta(\bx)$ denotes the $\ba$-th entry of the tensor $\zeta(\bx)$, while $\ba\not\prec\bx$ means that there exist two indices $i,j$ for which $a_i=a_j$ but $x_i\neq x_j$.)

\begin{thm}
\label{thm_decoupling_BLPAIP_for_cliques}
Let $2\leq k\leq n\in\N$, let $\X$ be a loopless digraph, and let $\zeta:\Xk\to\freeZ(\K_n^\tensor{k})$ be a $k$-tensorial homomorphism such that $\ba\not\prec\bx$ implies $E_\ba\ast\zeta(\bx)=0$ for any $\bx\in \Vset(\X)^k$ and $\ba\in \{1,\dots,n\}^k$. 
Then $\BA^k(\X,\K_n)=\YES$.
\end{thm}

\subsection{Crystals}
\label{subsec_crystals}

\noindent The criterion of acceptance for $\BA^k$ stated in Theorem~\ref{thm_decoupling_BLPAIP_for_cliques} is multilinear. Indeed, $\freeZ(\K_n^\tensor{k})$ is a space of integer \textit{affine} tensors (where we call a tensor affine if its entries sum up to $1$), and the existence of a $k$-tensorial homomorphism from $\Xk$ to $\freeZ(\K_n^\tensor{k})$ corresponds to the existence of a family of tensors 
satisfying a specific system of symmetries, which are formally described in~Remark~\ref{rem_description_freestructure_tensorisation}, see also the discussion at the beginning of Section~\ref{sec_crystals}. 
Letting $q$ be the number of vertices in $\X$,
such a family can be realised as the family of $k$-dimensional projections of a
single affine $q$-dimensional \emph{crystal} tensor, which we next informally define. We let $\mathcal{T}^{{n}\cdot\bone_q}(\Z)$ denote the set of all integer \emph{cubical tensors of dimension $q$ and width ${n}$}---i.e., ${n}\times {n}\times\dots\times {n}$ arrays of integers, where ${n}$ appears $q$ times. The notion of projecting should intuitively be thought of as ``summing up all entries of a tensor along a certain set of directions''; the formal definition shall make use of the operation of \emph{tensor contraction}, which we define in Section~\ref{subsec_tensors}. 
By oriented projection we mean that the directions are considered to be ordered. This is because, for example, the $2$-dimensional oriented projection of a $3$-dimensional tensor onto the directions $1$ and $2$ is the transpose of  the $2$-dimensional oriented projection of the same tensor onto the directions $2$ and $1$.
\begin{defn}[Informal]
\label{defn_crystal_informal}
Let $q,{n}\in\N$ and $k\in\{0,\dots,q\}$. A cubical tensor $C\in\mathcal{T}^{{n}\cdot\bone_q}(\Z)$ is a \emph{$k$-crystal} if all its $k$-dimensional oriented projections are equal. In this case, the \emph{$k$-shadow} of $C$ is this common oriented projection.
\end{defn}
Equivalently, a $k$-crystal is required to have equal $k$-dimensional projections \emph{up to reflection}---where a reflection is a higher-dimensional analogue of the transpose operation for matrices.
Let $\zeta_C$ be the map---associated with an affine $k$-crystal $C$---that takes a $k$-tuple $\bx$ of vertices of $\X$ and maps it to the projection of $C$ onto the hyperplane generated by $\bx$. By construction, $\zeta_C$ behaves well with respect to projections, so it is automatically $k$-tensorial. In order to yield a certificate of acceptance for $\BA^k(\X,\K_n)$, according to Theorem~\ref{thm_decoupling_BLPAIP_for_cliques}, $\zeta_C$ also needs to be a homomorphism and satisfy the extra condition $\ba\not\prec\bx\;\Rightarrow\;E_\ba\ast\zeta_C(\bx)=0$. It turns out that
both these requirements translate as a condition on the $k$-\emph{shadow} $S$ of $C$: The only entries of $S$ allowed to be nonzero are those whose coordinates are all distinct. We say that a tensor having this property is \emph{hollow} (the formal definition is given in Section~\ref{sec_crystals}). As an example, if $k=2$, the condition means that the $n\times n$ matrix $S$ needs to have zero diagonal; if $k=3$, three diagonal planes of the $n\times n\times n$ tensor $S$ of the form $(a,a,b)$, $(a,b,a)$, $(b,a,a)$ should be set to zero, and so on.

To summarise the discussion above, an affine $k$-crystal of dimension $q$ and width $n$ whose $k$-shadow is hollow yields a certificate that $\BA^k(\X,\K_n)=\YES$ for \emph{any} loopless digraph $\X$ with $q$ vertices. The problem is now to verify whether hollow-shadowed crystals exist. It is not hard to check that such crystals cannot exist for all choices of $k$, $q$, and $n$; this parallels the fact that the $\BA$ hierarchy is sound in the limit, so it cannot be the case that \emph{any} $\X$ is accepted by \emph{any} level of $\BA$ applied to \emph{any} clique $\K_n$. This is in sharp contrast with the weaker $\AIP$ hierarchy, for which a similar acceptance result holds, cf.~\cite{cz23soda:aip}. It follows that, unlike for $\AIP$, one cannot simply take large cliques as fooling instances for $\BA$. As we shall see in Section~\ref{subsec_fooling_BA}, a more refined family of digraphs can be shown to provide fooling instances for the $\BA$ hierarchy as long as one can produce 
hollow-shadowed crystals whose width $n$ is \emph{sub-exponential} in the level $k$. The main technical contribution of this work is a method for mining hollow-shadowed crystals whose width is \emph{quadratic} in $k$, as stated next.

\begin{thm}
\label{thm_existence_crystals_with_hollow_shadow}
For any $k\leq q\in\N$ there exists an affine $k$-crystal $C\in\mathcal{T}^{\frac{k^2+k}{2}\cdot\bone_q}(\Z)$ with hollow $k$-shadow.
\end{thm}
\noindent 
The key to establishing Theorem~\ref{thm_existence_crystals_with_hollow_shadow} is proving the following.
\begin{thm}
\label{prop_hollow_shadows_exist}
For any $k\in\N$ there exists a hollow affine $(k-1)$-crystal $C\in\mathcal{T}^{\frac{k^2+k}{2}\cdot\bone_k}(\Z)$.
\end{thm}
\noindent 
We now discuss the main ideas of the proof of Theorem~\ref{prop_hollow_shadows_exist} for the case $k=3$.
Our goal is to find a hollow affine $2$-crystal $C\in\cT^{6\cdot\bone_3}(\Z)$. In other words, $C$ must be a three-dimensional cubical tensor of width $6$, such that $(i)$ $C$ is hollow, i.e., the only entries allowed to be nonzero are the ones whose three coordinates are all distinct; $(ii)$ $C$ is affine, i.e., its entries sum up to $1$; and $(iii)$ $C$ is a $2$-crystal, i.e., projecting it onto the $xy$-, $yz$-, and $xz$-planes results in the same $6\times 6$ ``shadow'' matrix.
By induction, we can assume that Theorem~\ref{prop_hollow_shadows_exist} holds for $k=2$. In fact, it is not hard to find by inspection that the matrix 
\begin{equation*}
  U=\left[\begin{array}{ccc}0&0&1\\1&0&-1\\0&0&0\end{array}\right]=\squareTensorA{2}{.4}\,
\end{equation*}
is a hollow affine $1$-crystal in $\cT^{3\cdot\bone_2}(\Z)$. (We indicate the numbers $-1$, ${0}$, ${1}$, and $2$ by the colours green, light grey, yellow, and orange, respectively.)

The next step is to build a (not necessarily hollow) $3$-dimensional $2$-crystal having shadow $U$. 
In order to perform this task, we investigate the following question: Given a collection $\mathcal{S}$ of low-dimensional tensors (which we call a \emph{system of shadows}), which property characterises the fact that $\mathcal{S}$ is \emph{realisable}---i.e., that $\mathcal{S}$ is the family of oriented projections of a single high-dimensional tensor $T$? % 
Now, if $\br$ and $\bc$ are the row- and column-sum vectors of a matrix, the sums of the entries of $\br$ and $\bc$ must coincide. We say that $\mathcal{S}$ is a \emph{realistic} system of shadows if its members meet an analogous compatibility requirement, which is trivially satisfied whenever $\mathcal{S}$ consists of the projections of a common tensor; i.e., if $\mathcal{S}$ is realisable, it must be realistic.
In Section~\ref{subsec_system_of_shadows_BODY} we prove that the two conditions are in fact equivalent: A system of shadows is realistic if and only if it is realisable. Concretely, our proof shows how to build a tensor $T$ realising a given realistic system of shadows, and it is based on a nested induction (first on the dimension of the shadows, second on the sum of the sizes of the modes of $T$). A key fact, essential to making the process work, is that the problem is invariant under reflections of the tensors involved, cf.~Lemma~\ref{lem_rotating_preserves_truth}.

\begin{figure}[!pt]
\centering
\begin{minipage}[][][b]{.5\textwidth}
  \centering
  \VCrystal{-3}{.6}
  \captionof{figure}{The crystal $V$.}
  \label{fig_shadow_V}
\end{minipage}%
\begin{minipage}[][][b]{.5\textwidth}
  \centering
  \cubeTensorW{-3}{.5}
  \vspace{.49cm}
  \captionof{figure}{The crystal $W$.}
  \label{fig_crystal_W}
\end{minipage}
\end{figure}

In particular, this results in a \emph{crystalisation} procedure: By letting each member of the system of shadows $\mathcal{S}$ be a single lower-dimensional crystal $S$, one constructs a higher-dimensional crystal whose shadow is $S$ (see Section~\ref{subsec_crystalisation}). Applying this procedure to $U$ results in the crystal
\begin{equation*}
  V=\left[\begin{array}{ccc|ccc|ccc}
-1&0&1 &0&0&0 &1&0&0\\
2&0&-1 &0&0&0 &-1&0&0\\
-1&1&0 &0&0&0 &1&-1&0
\end{array}\right]\hspace{-.1cm},
%=\mbox{\cubeTensorV{-3}{.4}}.
\end{equation*}
shown in Figure~\ref{fig_shadow_V} together with its shadow (recall the colour/number correspondence described above).
Clearly, $V$ is not hollow---for example, its $(1,1,1)$-th coordinate is $-1\neq 0$. In fact, it is not hard to check that a hollow affine $2$-crystal of dimension $3$ and width $3$ cannot exist (see~Example~\ref{ex_no_hollo_crystal_of_small_width}).
We need to increase the width to ``make more space''; we do so by padding $V$ with three layers of zeros along each of the three dimensions. The tensor $W$ we obtain in this way (Figure~\ref{fig_crystal_W}) is clearly still a $2$-crystal. 
We can view $W$ as a block tensor with eight $3\times 3\times 3$ blocks; note that all non-zero entries of $W$ are in one block. 

\begin{figure}[!pb]
\centering
\begin{minipage}{.5\textwidth}
  \centering
  \quartzExample{-3}{.5}
  \captionof{figure}{The quartz $Q_{\ba,\bb}$.}
  \label{fig_quartz}
\end{minipage}%
\begin{minipage}{.5\textwidth}
  \centering
  \WMinusOneQuartz{-3}{.5}
  \captionof{figure}{$W-w_{(1,1,1)}\cdot Q_{(1,1,1),\bb}$.}
  \label{fig_W_minus_one_quartz}
\end{minipage}
\end{figure}

The strategy is now to ``spread'' these entries in the other blocks, in a way that they migrate to positions whose indices have no repetitions. To this end, we make use of a particular class of ``transparent'' crystals that we call \emph{quartzes}. Such crystals are designed in a way that the shadow they project is identically zero, meaning that we can freely add them (or their integer multiples) to a given crystal without changing its shadow and maintaining it a crystal. 
A quartz can be built by choosing two cells $\ba$ and $\bb$ having disjoint coordinates, considering the parallelepiped generated by $\ba$ and $\bb$, assigning value $1$ or $-1$ to its vertices in a way that two adjacent vertices get values of opposite sign, and assigning value $0$ to all other cells. We refer to such a tensor as to $Q_{\ba,\bb}$, see Figure~\ref{fig_quartz}; this construction is easily generalised to arbitrary dimension. 
Quartzes yield a method to relocate some nonzero entry of $W$, while leaving the rest of $W$ \emph{almost} untouched. More precisely, if the $\ba$-th entry of $W$ has value $w_\ba\neq 0$, the $\ba$-th entry of $W-w_\ba\cdot Q_{\ba,\bb}$ is zero, and this operation modifies the value of only $8$ cells of $W$. 

The idea is then to modify $W$ with suitable quartzes, so as to transfer all nonzero entries to positions where they do not violate the hollowness requirement. 
To this end, we take as $\bb$ a fixed cell that generates the smallest number of ties and that lies in the block of $W$ opposite to the one containing the nonzero entries---for example, the cell $\bb=(4,5,6)$, as in Figure~\ref{fig_quartz}. Even with such a choice, it can happen that adding a multiple of a quartz introduces new nonzero entries in positions that violate hollowness.
For example, Figure~\ref{fig_W_minus_one_quartz} shows the tensor $W-w_{(1,1,1)}\cdot Q_{(1,1,1),\bb}$. 
The value of the cell $(1,1,1)$ has become zero, as wanted, but three new forbidden cells ($(1,1,6)$, $(1,5,1)$, and $(4,1,1)$) now have nonzero values. However, the nonzero values in these forbidden cells cancel out once this procedure is applied to \emph{all} entries in the nonzero block of $W$. In other words, the affine $2$-crystal 
\begin{align*}
C=W-\sum_{\ba\in \{1,2,3\}^3}w_\ba\cdot Q_{\ba,\bb},
\end{align*}
shown in Figure~\ref{fig_hollow_crystal},
is hollow. 

\begin{figure}
\begin{center}
\HollowCrystal{-3}{.72}
\end{center}
\caption{The hollow crystal $C$.}
\label{fig_hollow_crystal}
\end{figure}

\subsection{Fooling the hierarchy}
\label{subsec_fooling_BA}
Let $C$ be an affine $k$-crystal of dimension $q$ and width $\frac{k^2+k}{2}$
whose $k$-shadow is hollow, as per
Theorem~\ref{thm_existence_crystals_with_hollow_shadow}. Let $\X$ be a loopless
digraph on vertex set $\Vset(\X)=\{1,\dots,q\}$. Consider the map $\zeta_C$ taking as
input a tuple $\bx$ of $k$ vertices of $\X$ (i.e., a tuple of $k$ numbers in
$\{1,\dots,q\}$) and returning the $k$-dimensional projection of $C$ onto the
hyperplane corresponding to $\bx$. As discussed earlier, $\zeta_C$ yields a
$k$-tensorial homomorphism from $\Xk$ to $\freeZ(\K_{(k^2+k)/2}^\tensor{k})$, and the fact that the shadow of $C$ is hollow translates as $\zeta_C$ satisfying the extra requirement of Theorem~\ref{thm_decoupling_BLPAIP_for_cliques}. Hence, we obtain the following.

\begin{thm}
\label{acceptance_prop_BA_big_enough}
Let $2\leq k\in\N$ and let $\X$ be a loopless digraph. Then $\BA^k(\X,\K_{(k^2+k)/2})=\YES$.
\end{thm}
To prove Theorem~\ref{thm_BLPAIPk_no_solves_AGC}, we need to show that $\BA^k$
does not solve $\PCSP(\K_c,\K_d)$ for all choices of $k\in\N$ and $3\leq c\leq
d\in\N$. If $c=\frac{k^2+k}{2}$, any graph with chromatic number bigger than $d$
(for example, the clique $\K_{d+1}$) would then yield a fooling instance. Since
increasing $c$ can only make AGC harder, this argument shows that $\BA^k$ does not solve $\PCSP(\K_c,\K_d)$ as long as $c\geq\frac{k^2+k}{2}$, and the fooling instances are simply cliques. 
In order to establish Theorem~\ref{thm_BLPAIPk_no_solves_AGC} in full generality, however, we shall pick the fooling instances from a more refined class of digraphs: the so-called \emph{shift digraphs} (see Figure~\ref{fig_shift_digraphs}).

\begin{defn}
The \emph{line digraph} of a digraph $\X$ is the digraph $\delta\X$ defined by $\Vset(\delta\X)=\Eset(\X)$ and $\Eset(\delta\X)=\{((x,y),(y,z)): (x,y),(y,z)\in \Eset(\X)\}$.
\end{defn}

\begin{defn}
\label{defn_shift_digraphs}
Let $q\in\N$ and $i\in\N\cup\{0\}$. The \emph{shift digraph} $\SSS_{q,i}$ is recursively defined by setting $\SSS_{q,0}=\K_q$, $\SSS_{q,i}=\delta\SSS_{q,i-1}$ for each $i\geq 1$.
\end{defn}
It is not hard to verify that the following non-recursive description of shift
digraphs is equivalent to Definition~\ref{defn_shift_digraphs} for $i\geq 1$: $\SSS_{q,i}$ is
the digraph whose vertex set consists of all strings of length $i+1$ over the
alphabet $\{1,\ldots,q\}$ such that consecutive letters are distinct, and whose edge set contains all pairs $(a_1\dots a_{i+1},\;b_1\dots b_{i+1})$ of strings such that $b_\ell=a_{\ell+1}$ for $\ell=1,\dots,i$.\footnote{In~\cite[\S2.5]{hell2004graphs}, a slightly different definition of shift digraphs is given, where the case $i=0$ is a transitive tournament rather than a clique; there, the vertex set of $\SSS_{q,i}$ only includes \emph{monotonically increasing} strings.}
The line digraph has been utilised in~\cite{GS20:icalp,KOWZ22} as a polynomial-time (and in fact log-space) reduction between $\PCSP$s. This construction changes the chromatic number in a controlled way, as we now describe.
Consider the integer functions $a$ and $b$ defined by $a(p)=2^p$ and
$b(p)={p\choose \floor{p/2}}$ for $p\in\N$, and notice that $a(p)\geq b(p)$ for
each $p$. Let $a^{(i)}$ and $b^{(i)}$ be the functions obtained by iterating $a$
and $b$, respectively, $i$-many times, for $i\in \N$. The following result from~\cite[Theorems~8--9]{HarnerE72}  bounds the chromatic number of the line digraph in terms of that of the original digraph. 
\begin{thm}[\cite{HarnerE72}]
\label{prop_chromatic_number_line_digraph}
Let $\X$ be a digraph and let $p\in\N$.
If $\delta\X\to\K_p$, then $\X\to\K_{a(p)}$;
if $\X\to\K_{b(p)}$, then $\delta\X\to \K_p$.
\end{thm}

An interesting feature of the line digraph operator is that it preserves acceptance by hierarchies of relaxations corresponding to conic minions, at the only cost of halving the level (see Proposition~\ref{prop_reduction_line_digraph}). As stated next, this in particular holds for the $\BA$ hierarchy, whose corresponding minion $\BAminion$ is conic.

\begin{prop}
\label{prop_reduction_line_digraph_BA}
Let $2\leq k\in\N$, let $\X,\A$ be digraphs, and suppose that $\BA^{2k}(\X,\A)=\YES$.
Then $\BA^{k}(\delta\X,\delta\A)=\YES$.
\end{prop}

The key point is that, under the application of the line digraph operator, the chromatic number of a digraph decreases \emph{exponentially} fast, while the $\BA$ acceptance level decreases only \emph{polynomially} fast.
\begin{figure}
\centering
\begin{subfigure}{.22\textwidth}
  \adjincludegraphics[width=.8\linewidth,trim={{.2\width} {.45\width} {.15\width} {.45\width}},clip]{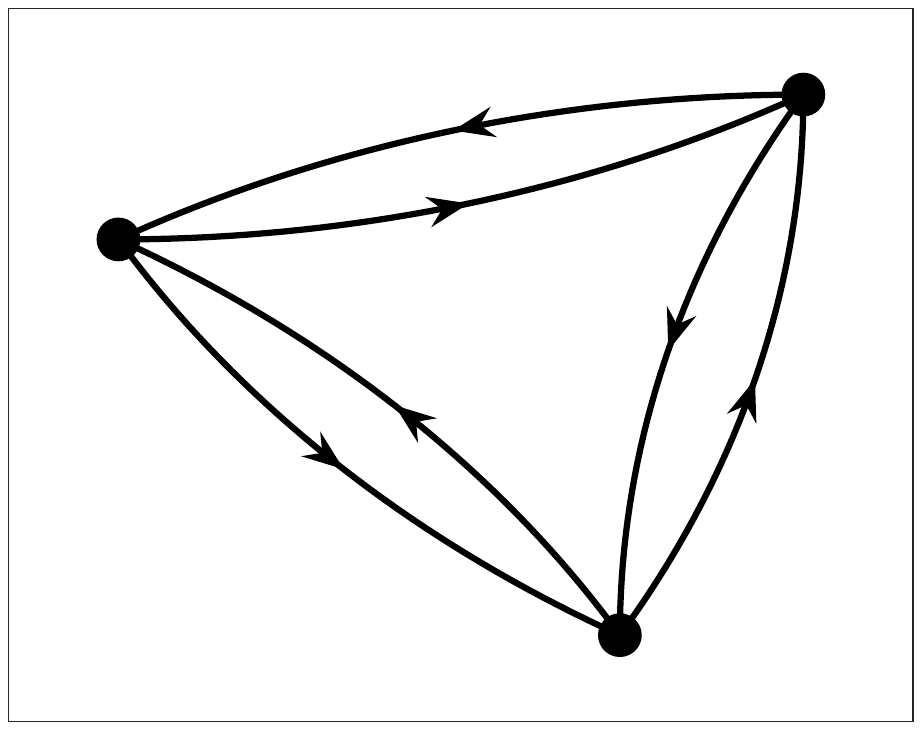}
  \caption{$\SSS_{3,0}$}
  \label{fig:sub1a}
\end{subfigure}%
\begin{subfigure}{.24\textwidth}
  \adjincludegraphics[width=.85\linewidth,trim={{.2\width} {.45\width} {.15\width} {.45\width}},clip]{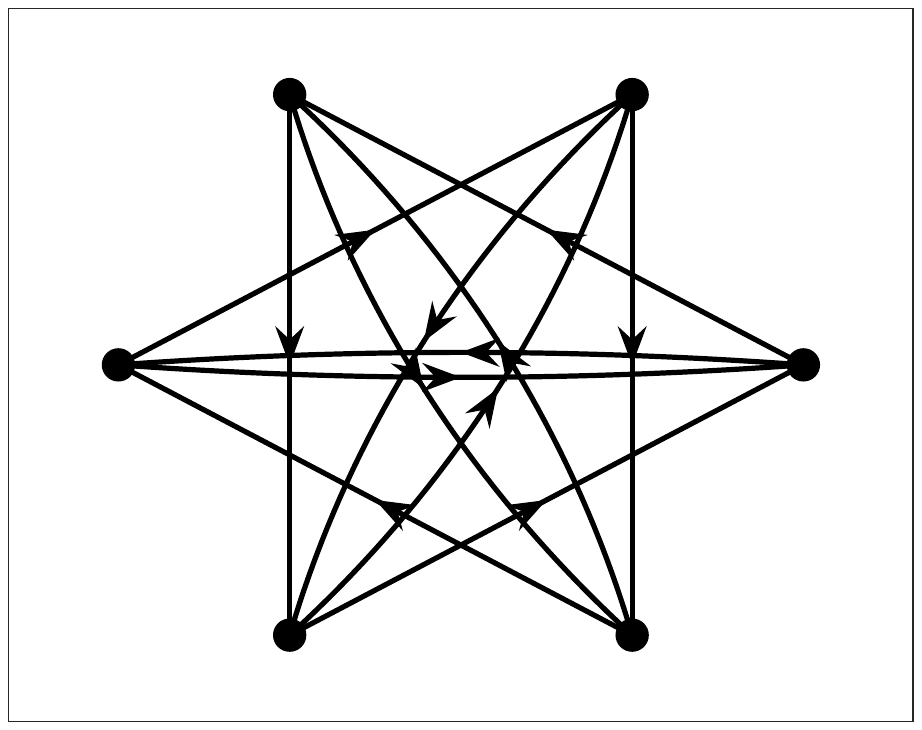}
  \caption{$\SSS_{3,1}$}
  \label{fig:sub1b}
\end{subfigure}%
\begin{subfigure}{.27\textwidth}
  \adjincludegraphics[width=.95\linewidth,trim={{.2\width} {.45\width} {.15\width} {.45\width}},clip]{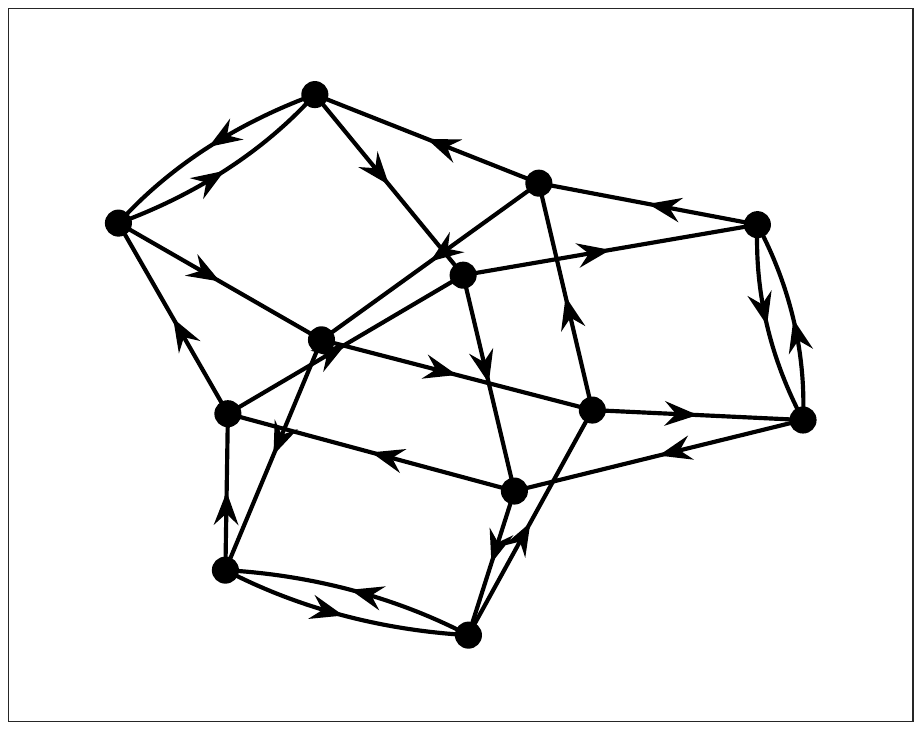}
  \caption{$\SSS_{3,2}$}
  \label{fig:sub1c}
\end{subfigure}%
\begin{subfigure}{.27\textwidth}
  \adjincludegraphics[width=1\linewidth,trim={{.2\width} {.45\width} {.15\width} {.45\width}},clip]{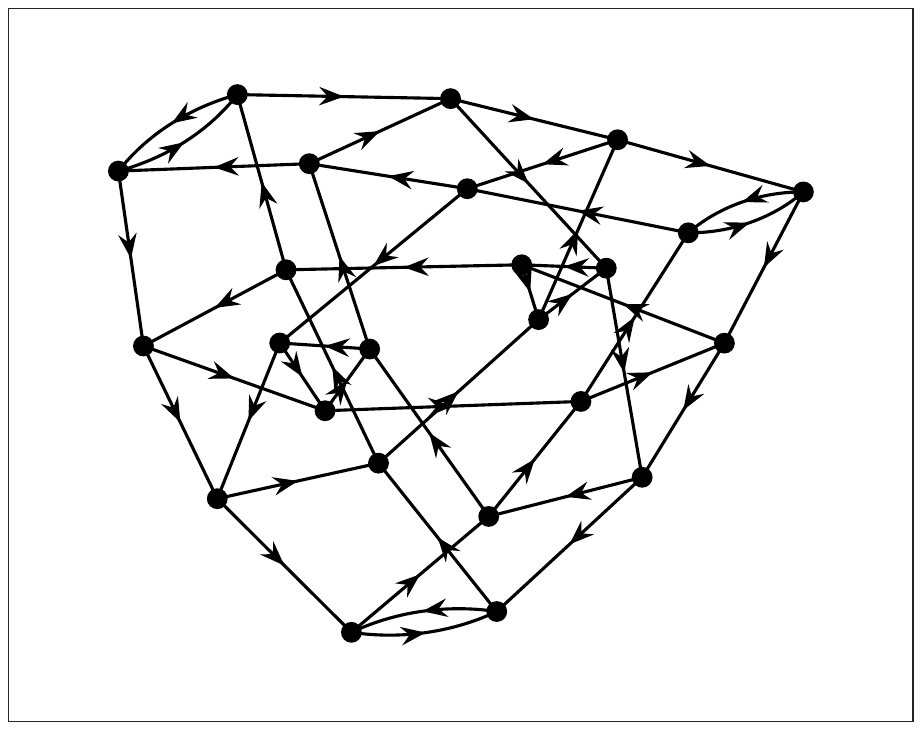}
  \caption{$\SSS_{3,3}$}
  \label{fig:sub1d}
\end{subfigure}%
\caption{Shift digraphs.}
\label{fig_shift_digraphs}
\end{figure}
Intuitively, our strategy to fool $\BA^k$ as an algorithm to solve $\PCSP(\K_c,\K_d)$ will be to take as the fooling instance
a shift digraph $\SSS_{q,i}$ where $q\sim\exp^{(i)}(d+1)$, rather than the clique $\K_{d+1}$. Here, by $\exp^{(i)}(\cdot)$ ($\operatorname{pol}^{(i)}(\cdot)$), we mean a function obtained by iterating $i$-many times an exponential (polynomial) function. Chromatically, this digraph is similar to $\K_{d+1}$ by Theorem~\ref{prop_chromatic_number_line_digraph}, so it is not $d$-colourable. On the other hand, for large enough $i$, the difference in speed decrease guarantees that $\BA^{\operatorname{pol}^{(i)}(k)}(\K_q,\K_{\operatorname{exp}^{(i)}(c)})=\YES$ by Theorem~\ref{acceptance_prop_BA_big_enough}---which, after applying Proposition~\ref{prop_reduction_line_digraph_BA} for a suitable number of times, eventually implies $\BA^k(\SSS_{q,i},\K_c)=\YES$.
We note that this argument crucially depends on the fact that the size $\frac{k^2+k}2$ of the clique in Theorem~\ref{acceptance_prop_BA_big_enough}---i.e., the width of the hollow-shadowed crystals constructed in Section~\ref{subsec_crystals}---is sub-exponential in $k$. Before proving Theorem~\ref{thm_BLPAIPk_no_solves_AGC} in full detail, we present a result---which holds for hierarchies corresponding to all linear minions---stating that acceptance of some instance $\X$ by some level of the $\BA$ hierarchy is preserved under homomorphisms of the template.

\begin{prop}
\label{prop_BA_preserves_homo}
Let $2\leq k\in\N$, let $\X,\A,\B$ be digraphs such that $\A\to\B$, and suppose that $\BA^k(\X,\A)=\YES$. Then $\BA^k(\X,\B)=\YES$.
\end{prop}

\begin{proof}[Proof of Theorem~\ref{thm_BLPAIPk_no_solves_AGC}]
Since $\BA^2$ is at least as powerful as $\BA^1$, we can assume that $k\geq 2$.
Suppose first that $c\geq 4$. In this case, we can find $i\in\N$ such that $b^{(i)}(c)\geq k^2 4^i$.
Take $q>a^{(i)}(d)$.
We claim that the shift digraph $\SSS_{q,i}$ is a fooling instance for the $k$-th level of $\BA$ applied to $\PCSP(\K_c,\K_d)$; in other words, we claim that $(1)$ $\BA^k(\SSS_{q,i},\K_c)=\YES$ and $(2)$ $\SSS_{q,i}\not\to\K_d$.

For $(1)$, we start by applying Theorem~\ref{acceptance_prop_BA_big_enough} to find that $\BA^{k2^i}(\K_q,\K_{(k^24^i+k2^i)/2})=\YES$.
Observe that $$\frac{k^2 4^i+k 2^i}{2}\leq k^24^i\leq b^{(i)}(c),$$ so $$\K_{(k^2
  4^i+k 2^i)/2}\to \K_{k^24^i} \to \K_{b^{(i)}(c)}.$$ By Proposition~\ref{prop_BA_preserves_homo}, we deduce that $\BA^{k2^i}(\K_q,\K_{b^{(i)}(c)})=\YES$. Applying Proposition~\ref{prop_reduction_line_digraph_BA} repeatedly, we obtain $\BA^k(\SSS_{q,i},\SSS_{b^{(i)}(c),i})=\YES$. Noticing that $\K_{b^{(i)}(c)}\to\K_{b^{(i)}(c)}$ and applying the second part of Theorem~\ref{prop_chromatic_number_line_digraph} repeatedly, we find $\SSS_{b^{(i)}(c),i}\to\K_c$. Again by Proposition~\ref{prop_BA_preserves_homo}, we conclude that $\BA^k(\SSS_{q,i},\K_c)=\YES$, as required.
For $(2)$, we first note that $\K_q\not\to\K_{a^{(i)}(d)}$ as $q>a^{(i)}(d)$. Applying the (contrapositive of the) first part of Theorem~\ref{prop_chromatic_number_line_digraph} repeatedly, we deduce that $\SSS_{q,i}\not\to\K_d$, as required.

Suppose now that $c=3$. Assume, for the sake of contradiction, that the $k$-th
  level of $\BA$ solves $\PCSP(\K_3,\K_d)$. Let $\X$ be a digraph such that
  $\BA^{4k}(\X,\K_4)=\YES$. Applying
  Proposition~\ref{prop_reduction_line_digraph_BA} twice, we find that
  $\BA^k(\delta(\delta\X),\SSS_{4,2})=\YES$. It was observed in~\cite{zhu1998survey} (see also~\cite{poljak1991coloring}) that $\SSS_{4,2}\to\K_3$. Combining this with
  Proposition~\ref{prop_BA_preserves_homo} yields
  $\BA^k(\delta(\delta\X),\K_3)=\YES$. Since we are assuming that $\BA^k$ solves
  $\PCSP(\K_3,\K_d)$, we must have $\delta(\delta\X)\to\K_d$, whence it follows,
  through a double application of the first part of
  Theorem~\ref{prop_chromatic_number_line_digraph}, that $\X\to\K_{a^{(2)}(d)}$.
  Note now that $d\geq c=3$ implies $a^{(2)}(d)=2^{2^d}\geq 2^{2^3}\geq 4$, so $\K_4\to\K_{a^{(2)}(d)}$, which means that $\PCSP(\K_4,\K_{a^{(2)}(d)})$ is well defined. Hence, we have shown that the $(4k)$-th level of $\BA$ solves $\PCSP(\K_4,\K_{a^{(2)}(d)})$, thus contradicting the first part of the proof.
\end{proof}

\section{Preliminaries}
\label{sec_preliminaries}
Throughout this work,
the expression ``$x \stackrel{\operatorname{L}.\bullet}{\;=\;} y$'' shall mean ``$x=y$ by Lemma~$\bullet$''. Similarly, 
``$x \stackrel{\operatorname{P}.\bullet}{\;=\;} y$'' and 
``$x \stackrel{(\bullet)}{\;=\;} y$'' shall mean 
``$x=y$ by Proposition~$\bullet$'' and 
``$x=y$ by equation~$(\bullet)$'', respectively.

\subsection{Hypergraphs}
For ${k}\in\N$, a \emph{${k}$-uniform hypergraph} $\HH$ consists of a set $\Vset(\HH)$ of elements called \emph{vertices} and a set $\Eset(\HH)\subseteq \Vset(\HH)^{k}$ of tuples of ${k}$ vertices called \emph{hyperedges}.
A $2$-uniform hypergraph is a \emph{digraph}, as defined in Section~\ref{sec:intro}.
The notion of homomorphism, defined in Section~\ref{sec:intro} for digraphs, naturally extends to hypergraphs: Given two ${k}$-uniform hypergraphs $\HH$ and $\tilde\HH$, a map $f:\Vset(\HH)\to\Vset(\tilde\HH)$ is a \emph{homomorphism} from $\HH$ to $\tilde\HH$ if $f(\bh)\in \Eset(\tilde\HH)$ for any $\bh\in \Eset(\HH)$, where $f$ is applied component-wise to the vertices in $\bh$. We indicate the existence of a homomorphism from $\HH$ to $\tilde\HH$ by writing $\HH\to\tilde\HH$.

\subsection{Tuples}
\label{subsec_tuples}
We let $\N$ be the set of positive integers, and we let $\N_0=\N\cup\{0\}$. Given $n\in\N$, $[n]$ denotes the set $\{1,\dots,n\}$. We additionally set $[0]=\emptyset$. Given a tuple $\bn=(n_1,\dots,n_q)\in\N^q$ for some $q\in\N$, we denote by $[\bn]$ the set $[n_1]\times\dots\times[n_q]$.
Given a tuple $\ba=(a_1,\dots,a_q)\in [\bn]$ and a tuple $\bi=(i_1,\dots,i_p)\in
[q]^p$ for $p\in\N$,
the \emph{projection} of $\ba$ onto $\bi$ is the tuple $\ba_\bi$ obtained by selecting from $\ba$ the entries indexed by $\bi$; i.e., $\ba_\bi= (a_{i_1},\dots,a_{i_p})$. Notice that $\ba_\bi\in [\bn_\bi]$. 
Tuple projection is associative, in the sense that, for $\bj\in [p]^m$, $(\ba_\bi)_\bj=\ba_{(\bi_\bj)}$. Hence, we will omit parantheses when dealing with iterated projections.
For $\tilde\bn\in\N^p$ and $\bb=(b_1,\dots,b_p)\in [\tilde{\bn}]$, the \emph{concatenation} of $\ba$ and $\bb$ is the tuple $(\ba,\bb)=(a_1,\dots,a_q,b_1,\dots,b_p)$. Notice that $(\ba,\bb)\in [(\bn,\tilde{\bn})]$. 
It will be handy to extend the notation above to include tuples of length zero.
For any set $S$, we define $S^0= \{\bepsilon\}$, where $\bepsilon$ denotes the empty tuple. For any tuple $\bx$, we let $\bx_\bepsilon=\bepsilon$ and $(\bx,\bepsilon)=(\bepsilon,\bx)=\bx$. We also define $[\bepsilon]= \{\bepsilon\}$. 
For $n\in\N$, we define the tuple $\ang{n}=(1,\dots,n)$. Also, we let $\ang{0}=\bepsilon$. 
The cardinality of a set $S$ is denoted by $|S|$. Given a tuple $\bs\in S^k$ for
some $k\in\N_0$, $\set(\bs)=\{s_i:i\in [k]\}$ is the set of elements appearing
in $\bs$, while $|\bs|=|\set(\bs)|$ is the number of distinct entries in
$\bs$.
Given two sets $S,\tilde S$ and two tuples $\bs=(s_1,\dots,s_k)\in S^k$, $\tilde{\bs}=(\tilde{s}_1,\dots,\tilde{s}_k)\in \tilde{S}^k$, we write $\bs\prec\tilde{\bs}$ if, for any $i,j\in [k]$, $s_i=s_j$ implies $\tilde{s}_i=\tilde{s}_j$. 
We write $\bs\sim\tilde\bs$ if $\bs\prec\tilde\bs$ and $\tilde\bs\prec\bs$. The symbols ``$\not\prec$'' and ``$\not\sim$'' denote the negations of ``$\prec$'' and ``$\sim$'', respectively.
Observe that, for every $k$-tuple $\bs$, it holds that $\ang{k}\prec\bs\prec\mathbf{c}$, where $\mathbf{c}$ is a constant $k$-tuple.
We denote by $\bzero_k$ and $\bone_k$ the all-zero tuple and the all-one tuple of length $k$, respectively.

\subsection{Hierarchies of relaxations}
\label{subsec_hierarhies_of_relaxations}
Given two digraphs $\X$ and $\A$ and an integer $k\in\N$, introduce a variable $\lambda_{\bx,\ba}$ for each $\bx\in\Vset(\X)^k$ and $\ba\in\Vset(\A)^k$, and a variable $\mu_{\by,\bb}$ for each $\by\in\Eset(\X)$ and $\bb\in\Eset(\A)$. Consider the following system of equations:
\begin{align}
\label{eqn_hierarchy}
\tag{$\IP^k$}
\begin{array}{llllllll}
\mbox{($\IP^k_1$)}\qquad & \displaystyle\sum_{\ba\in \Vset(\A)^k}\lambda_{\bx,\ba}&=&1 & \bx\in\Vset(\X)^k\\[10pt]
\mbox{($\IP^k_2$)} & \displaystyle \sum_{\substack{\hat\ba\in \Vset(\A)^k\\\hat\ba_\bi=\ba}}\lambda_{\bx,\hat\ba}&=&\lambda_{\bx_\bi,\ba} & 
\bx\in\Vset(\X)^k,\;\bi\in [k]^k,\;\ba\in\Vset(\A)^k\\[10pt]
\mbox{($\IP^k_3$)} & \displaystyle \sum_{\substack{\bb\in \Eset(\A)\\\bb_\bi=\ba}}\mu_{\by,\bb}&=&\lambda_{\by_\bi,\ba} & 
\by\in\Eset(\X),\;\bi\in [2]^k,\;\ba\in\Vset(\A)^k\\[10pt]
\mbox{($\IP^k_4$)} & \displaystyle \lambda_{\bx,\ba}&=&0 & 
\bx\in\Vset(\X)^k,\;\ba\in\Vset(\A)^k,\;\bx\not\prec\ba\\[10pt]
\mbox{($\IP^k_5$)} & \displaystyle \mu_{\by,\bb}&=&0 & 
\by\in\Eset(\X),\;\bb\in\Eset(\A),\;\by\not\prec\bb.
\end{array}
\end{align}
The equations~($\IP^k_1$) enforce that the variables be properly scaled\footnote{($\IP^k_1$) requires that only the $\lambda$ variables should sum up to $1$, but combining~($\IP^k_1$)~and~($\IP^k_3$) yields the same requirement for the $\mu$-variables as well.}---which is particularly desirable if we wish to interpret them as probability distributions over the set of assignments of vertices of $\A$ (``colours'') to sets of vertices of $\X$. Given a joint probability distribution over some random variables, the corresponding probability distribution over a subset of the variables is obtained by \emph{marginalising}; i.e., by summing up over all variables that are ignored. The equations~($\IP^k_2$)~and~($\IP^k_3$) simulate this marginality requirement for the distributions $\lambda$ and $\mu$, respectively. Finally, the equations~($\IP^k_4$)~and~($\IP^k_5$) simply make sure that a vertex of $\X$ appearing multiple times in the same tuple receives the same colour.
Note that~($\IP^k_5$) is superfluous when $\X$ is loopless since, in that case, no edge $\by$ satisfies $\by\not\prec\bb$.\footnote{A different formulation of the system~\eqref{eqn_hierarchy} would consider $\lambda$-variables corresponding to \emph{sets} rather than \emph{tuples} of vertices; by virtue of~($\IP^k_4$), the two formulations are equivalent.} 

Let $k\geq 2$. We write $\BLP^k(\X,\A)=\YES$ if the system~\eqref{eqn_hierarchy} admits a solution such that all variables take rational nonnegative values. 
We write $\AIP^k(\X,\A)=\YES$ if the system~\eqref{eqn_hierarchy} admits a solution such that all variables take integer (possibly negative) values. 
We write $\BA^k(\X,\A)=\YES$ if the system~\eqref{eqn_hierarchy} admits both a solution such that all variables take rational nonnegative values and a solution such that all variables take integer values, and the following \emph{refinement condition} holds: Denoting by the superscript $(\Bmat)$ the variables in the $\BLP^k$ solution and by the superscript $(\Amat)$ those in the $\AIP^k$ solution, we require that
\begin{subequations}
\begin{align}
     \lambda^{(\Bmat)}_{\bx,\ba}=0 \;\;\;\Rightarrow\;\;\; \lambda^{(\Amat)}_{\bx,\ba}=0
    &&&
    \mbox{for each }
    \bx\in\Vset(\X)^k,\;
    \ba\in\Vset(\A)^k
    \label{eqn_refinement_condition_lambda}\\[5pt]
    \mu^{(\Bmat)}_{\by,\bb}=0 \;\;\;\Rightarrow\;\;\; \mu^{(\Amat)}_{\by,\bb}=0
    &&&
    \mbox{for each }
    \by\in\Eset(\X),\;
    \bb\in\Eset(\A).\label{eqn_refinement_condition_mu}
\end{align}
\end{subequations}
\begin{rem}
The following is a procedure to check whether $\BA^k(\X,\A)=\YES$ in polynomial time in the size of $\X$ (cf.~\cite{bgwz20}): 
\begin{enumerate}
    \item  Check whether~\eqref{eqn_hierarchy} has a rational nonnegative solution. If it does not, output $\NO$; otherwise:
    \item Select a solution $(\lambda^{\operatorname{ri}},\mu^{\operatorname{ri}})$ lying in the relative interior of the polytope of solutions.
    \item
    Check whether there exists an integer solution to the system~\eqref{eqn_hierarchy}, \emph{refined} with the requirement that all variables whose value in $(\lambda^{\operatorname{ri}},\mu^{\operatorname{ri}})$ is zero should be set to zero. If there is one, output $\YES$; otherwise, output $\NO$.
\end{enumerate}
The procedure above can be implemented in a way that it runs in polynomial time in the size of $\X$: Step 1 corresponds to checking whether an LP on polynomially many variables is feasible; step 2 has polynomial run-time by virtue of a result in~\cite{MR1261419} (cf.~\cite{BG19}); step 3 corresponds to checking feasibility of a system of linear Diophantine equations on polynomially many variables, which can be done in polynomial time by computing the Hermite or the Smith normal forms of the matrix of coefficients, see~\cite{MR874114}.

Clearly, if such procedure outputs $\YES$, then $\BA^k(\X,\A)=\YES$. For the
  converse implication, suppose that $\BA^k(\X,\A)=\YES$ and let
  $(\lambda^{(B)},\mu^{(B)})$ and $(\lambda^{(A)},\mu^{(A)})$ be solutions
  to~\eqref{eqn_hierarchy} witnessing it. Notice that, in this case, the
  procedure does produce a solution
  $(\lambda^{\operatorname{ri}},\mu^{\operatorname{ri}})$, but this may differ
  from $(\lambda^{(B)},\mu^{(B)})$. Nevertheless, any variable that is zero in
  $(\lambda^{\operatorname{ri}},\mu^{\operatorname{ri}})$ is also zero in
  $(\lambda^{(B)},\mu^{(B)})$ (by the definition of relative interior,
  cf.~\cite{schrijver1998theory}), so $(\lambda^{(A)},\mu^{(A)})$ does witness that the refined system of step 3 has an integer solution and, thus, that the procedure outputs $\YES$.
\end{rem}

We also define $\BLP^1$, $\AIP^1$, and $\BA^1$ as $\BLP$, $\AIP$, and $\BA$, respectively, as described in Section~\ref{sec_overview}. Notice that this almost entirely corresponds to taking $k=1$ in the definition above, except for the fact that the equations ($\IP^1_5$) are dropped. Indeed, looking at~\eqref{eqns_BLP}, we observe that ($\IP^1_1$) is equivalent to ($\IP_1$), ($\IP^1_3$) is equivalent to ($\IP_2$), while ($\IP^1_2$) and ($\IP^1_4$) are vacuous; however, ($\IP^1_5$) is not implied by the system~\eqref{eqns_BLP}.

\begin{rem}
\label{rem_redundancy_mu_refinement}
For $k\geq 2$,
the equations~\eqref{eqn_refinement_condition_mu} are implied by the equations~\eqref{eqn_refinement_condition_lambda}.
Indeed, suppose that $\mu^{(\Bmat)}_{\by,\bb}=0$ for some $\by\in\Eset(\X)$, $\bb\in\Eset(\A)$. Observe that, for the tuple $\bi=(1,2,1,\dots,1)\in [2]^k$, we have $\{\bc\in\Eset(\A):\bc_\bi=\bb_\bi\}=\{\bb\}$. Hence,~($\IP^k_3$) yields
\begin{align*}
    \mu^{(\Bmat)}_{\by,\bb}
    \spaceeq
    \sum_{\substack{\bc\in\Eset(\A)\\\bc_\bi=\bb_\bi}}\mu^{(\Bmat)}_{\by,\bc}
    \spaceeq
    \lambda^{(\Bmat)}_{\by_\bi,\bb_\bi}
\end{align*}
and, similarly, $\mu^{(\Amat)}_{\by,\bb}=\lambda^{(\Amat)}_{\by_\bi,\bb_\bi}$. Therefore, $\mu^{(\Bmat)}_{\by,\bb}=0$ implies $\lambda^{(\Bmat)}_{\by_\bi,\bb_\bi}=0$, whence it follows through~\eqref{eqn_refinement_condition_lambda} that $\lambda^{(\Amat)}_{\by_\bi,\bb_\bi}=0$, thus forcing $\mu^{(\Amat)}_{\by,\bb}=0$. In fact, the same holds if the hierarchy is applied to arbitrary relational structures rather than digraphs---in which case, we require that $k$ be at least the maximum arity of the relations in the structures. 
\end{rem}

Given two digraphs $\A$ and $\B$ such that $\A\to\B$, we say that $\BA^k$ ($\BLP^k$, $\AIP^k$) \emph{solves} $\PCSP(\A,\B)$ if $\X\to\B$ whenever $\BA^k(\X,\A)=\YES$ ($\BLP^k(\X,\A)=\YES$, $\AIP^k(\X,\A)=\YES$). Note that the algorithms are complete: If $\X\to\A$ then $\BA^k(\X,\A)=\BLP^k(\X,\A)=\AIP^k(\X,\A)=\YES$.
Indeed, given a homomorphism $f: \X\to\A$, the distributions assigning weight $1$ to $(\bx,f(\bx))$ for each $\bx\in \Vset(\X)^k\cup E(\X)$ and weight $0$ to any other pair $(\bx,\ba)$ are easily seen to yield both a $\BLP^k$ and an $\AIP^k$ solution, and to satisfy the extra refinement condition.
Hence, the algorithms do not produce false negatives (but may produce false positives).

\subsection{Tensors}
\label{subsec_tensors}

Take a set $S$, an integer $q\in\N_0$, and a tuple $\bn\in\N^q$. We denote by $\cT^\bn(S)$ the set of functions from $[\bn]$ to $S$. We call a function $T$ in $\cT^\bn(S)$ a \emph{tensor} on $q$ \emph{modes} of sizes $n_1,\dots,n_q$, and we visualise $T$ as a $q$-dimensional array or hypermatrix, each of whose cells contains an element of $S$. We sometimes use the notation $T=(t_\bi)_{\bi\in [\bn]}$ where, for $\bi\in [\bn]$, $t_\bi\in S$ is the \emph{$\bi$-th entry of $T$}; i.e., the image of $\bi$ under $T$.
For example, $\cT^{n}(S)$ and $\cT^{(m,n)}(S)$ are the sets of $n$-vectors and $m\times n$ matrices, respectively, having entries in $S$. Notice that $\cT^\bepsilon(S)$ is the set of functions from $[\bepsilon]=\{\bepsilon\}$ to $S$, which we identify with $S$.
We will often consider \emph{cubical} tensors, all of whose modes have equal length; i.e., tensors in the set $\cT^{n\cdot\bone_q}(S)$ for some $n\in \N$.

Many tensors appearing throughout this work have entries in the field $\Q$ of rational numbers. Such tensors can be multiplied with each other via an operation that generalises several linear-algebraic products. Take three integers $a,b,c\in\N_0$ and three tuples $\ba\in\N^a$, $\bb\in\N^b$, $\bc\in\N^c$. The \emph{contraction} of two tensors $T=(t_\bi)_{\bi\in [(\ba,\bb)]}\in\cT^{(\ba,\bb)}(\Q)$ and $U=(u_\bi)_{\bi\in [(\bb,\bc)]}\in\cT^{(\bb,\bc)}(\Q)$, denoted by $T\cont{b}U$, is the tensor in $\cT^{(\ba,\bc)}(\Q)$ whose $(\bj,\bell)$-th entry is \[\sum_{\bk\in[\bb]}t_{(\bj,\bk)}u_{(\bk,\bell)}\] 
for $\bj\in[\ba]$ and $\bell\in[\bc]$.
If at least one of $a$ and $c$ equals zero---i.e., if we are contracting over all modes of $T$ or $U$---we write $T\ast U$ for $T\cont{b}U$ to increase readability. 
\begin{example}
For $m,n,p\in\N$, consider the tensors $z\in\cT^{\bepsilon}(\Q)=\Q$; $\bu,\bv\in\cT^{m}(\Q)$; $\bw\in\cT^{n}(\Q)$; $M,N\in\cT^{(m,n)}(\Q)$; and $Q\in\cT^{(n,p)}(\Q)$. Following~\cite[Example~22]{cz23soda:minions}, we can list several classic linear-algebraic products as instances of tensor contraction:
\begin{align*}
    \begin{array}{llllllll}
        z\cont{0}\bu&=&z\ast\bu&=&z\bu & \mbox{ (scalar times vector)} \\
        z\cont{0}M&=&z\ast M&=&zM & \mbox{ (scalar times matrix)} \\
        \bu\cont{1}\bv&=&\bu\ast\bv&=&\bu^T\bv & \mbox{ (inner product of vectors)} \\
        \bu\cont{0}\bw&=&\bu\bw^T && &\mbox{ (outer product of vectors)} \\
        M\cont{1}Q&=&MQ && & \mbox{ (standard matrix product)} \\
        M\cont{2}N&=&M\ast N&=&\tr(M^TN) &\mbox{ (Frobenius inner product of matrices)}.
    \end{array}
\end{align*}
\end{example}
\noindent Let $a\in\N_0$ and $\ba\in\N^a$.
Given $\bi\in[\ba]$, we denote by $E_\bi$ the \emph{$\bi$-th standard unit tensor}; i.e., the tensor in
$\cT^\ba(\Q)$ all of whose entries are $0$, except the
$\bi$-th entry that is $1$. 
(While this tensor is defined in terms of both $\bi$ and $\ba$, the latter tuple shall always be clear from the context, and we do not indicate it explicitly in the notation to improve readability.)
Observe that, for any $T\in \cT^\ba(\Q)$, we may express the $\bi$-th entry of $T$ as $E_\bi\ast T$. 
In other words, if $T=(t_\bi)_{\bi\in [\ba]}$, then $E_\bi\ast T=t_\bi$.
We let the \emph{support} of $T$ be the set of indices of all nonzero entries of $T$; i.e., the set $\supp(T)=\{\bi\in [\ba]:E_\bi\ast T\neq 0\}$.
\begin{rem}
\label{rem_E_epsilon_equals_one}
    Since $\N^0=\{\bepsilon\}$ and $[\bepsilon]=\{\bepsilon\}$, the tensor $E_\bepsilon$ is well defined and lives in $\cT^\bepsilon(\Q)=\Q$. Observe that $E_\bepsilon=1$, as
    its unique entry---i.e., its $\bepsilon$-th entry---is $1$ by definition.
\end{rem}

\noindent 
As noted in~\cite{cz23soda:minions},
tensor contraction satisfies a specific form of associativity. We include a simple proof of this fact for completeness.
\begin{lem}
\label{lem_associativity_contraction}
Take five integers $a,b,c,d,e\in\N_0$, five tuples $\ba\in\N^{a},\bb\in\N^{b},\bc\in\N^{c},\bd\in\N^{d},\be\in\N^{e}$, and three tensors $T\in\cT^{(\ba,\bb)}(\Q)$, $U\in\cT^{(\bb,\bc,\bd)}(\Q)$, $V\in\cT^{(\bd,\be)}(\Q)$. Then 
\begin{align*}
(T\cont{b}U)\cont{d}V
\spaceeq
T\cont{b}(U\cont{d}V).    
\end{align*}
\end{lem}
\begin{proof}
Let $W=(T\cont{b}U)\cont{d}V$ and $Z=T\cont{b}(U\cont{d}V)$, and observe that both $W$ and $Z$ are tensors in $\cT^{(\ba,\bc,\be)}(\Q)$. Take $\bi\in [\ba]$, $\bj\in [\bc]$, and $\bk\in [\be]$, and observe that the $(\bi,\bj,\bk)$-th entry of $W$ is
\begin{align*}
    E_{(\bi,\bj,\bk)}\ast W
    &\spaceeq
    \sum_{\bell\in [\bd]}\left[E_{(\bi,\bj,\bell)}\ast(T\cont{b}U)\right]\cdot \left[E_{(\bell,\bk)}\ast V\right]\\
    &\spaceeq
    \sum_{\bell\in [\bd]}\sum_{\bbm\in [\bb]}\left[E_{(\bi,\bbm)}\ast T\right]\cdot\left[E_{(\bbm,\bj,\bell)}\ast U\right]\cdot \left[E_{(\bell,\bk)}\ast V\right]
    \intertext{while the $(\bi,\bj,\bk)$-th entry of $Z$ is}
    E_{(\bi,\bj,\bk)}\ast Z
    &\spaceeq
    \sum_{\bbm\in[\bb]}\left[E_{(\bi,\bbm)}\ast T\right]\cdot\left[E_{(\bbm,\bj,\bk)}\ast (U\cont{d}V)\right]\\
    &\spaceeq
    \sum_{\bbm\in[\bb]}\left[E_{(\bi,\bbm)}\ast T\right]\cdot\sum_{\bell\in [\bd]}\left[E_{(\bbm,\bj,\bell)}\ast U\right]\cdot\left[E_{(\bell,\bk)}\ast V\right].
\end{align*}
The value of the two expressions is the same, so $W=Z$, as required.
\end{proof}
\begin{rem}
Lemma~\ref{lem_associativity_contraction} establishes that tensor contraction is associative if it is taken over disjoint sets of modes. It is easy to check that, if this hypothesis is dropped, associativity may not hold (see~\cite[\S4.1]{cz23soda:minions}). For example, consider three tensors $T\in\cT^{(\ba,\bb)}(\Q)$, $U\in\cT^{(\bb,\bc)}(\Q)$, and $V\in\cT^{(\ba,\bc)}(\Q)$, where $\ba,\bb,\bc$ are as in Lemma~\ref{lem_associativity_contraction}. Then, the expression $(T\cont{b}U)\cont{a+c}V$ is well defined, while the expression obtained by switching the order of the contractions is not well defined in general. For this reason, we define the contraction operation to be left-associative by default, in the sense that the expression $T_1\cont{k_1}T_2\cont{k_2}T_3$ shall mean $(T_1\cont{k_1}T_2)\cont{k_2}T_3$. Whenever this is possible (i.e., whenever we are contracting over disjoint sets of modes), we shall tacitly make use of the associativity property of Lemma~\ref{lem_associativity_contraction}.
In particular, in this way, we can express the entry of index $(\bi,\bj)$ of a tensor $T\in\cT^{(\ba,\bb)}(\Q)$ (where $\bi\in[\ba]$ and $\bj\in [\bb]$) by the notation $E_\bi\ast T\ast E_\bj$; note that this is the same as $E_{(\bi,\bj)}\ast T$.
\end{rem}
\noindent 
\subsection{The projection tensor}
\label{subsec_projection_tensor}
\noindent Take $a,b\in\N_0$, $\ba\in\N^a$, and $\bell\in [a]^b$, and consider the \emph{projection tensor} $\Pi^\ba_\bell\in\cT^{(\ba_\bell,\ba)}(\Q)$
 defined, for each $\bi\in [\ba_\bell]$ and each $\bj\in [\ba]$,
 by 
\begin{align}
\label{eqn_projection_tensor}
E_\bi\ast \Pi^\ba_\bell\ast E_{\bj}\spaceeq
\left\{
\begin{array}{ll}
1 & \mbox{if }\bj_\bell=\bi\\
0 & \mbox{otherwise.}
\end{array}
\right.
\end{align}
In particular, observe that setting $b=0$ yields $[a]^b=\{\bepsilon\}$, so $\Pi^\ba_\bepsilon$ is well defined and lives in $\cT^{(\ba_\bepsilon,\ba)}(\Q)=\cT^\ba(\Q)$.

We now present some basic results on this special tensor, which justify its name and which shall be used throughout this work.
\begin{lem}
\label{Pi_epsilon_all_one}
Given $a\in\N_0$ and $\ba\in\N^a$, $\Pi^\ba_\bepsilon$ is the all-one tensor in $\cT^\ba(\Q)$.
\end{lem}
\begin{proof}
  Using that $E_\bepsilon=1$, as seen in Remark~\ref{rem_E_epsilon_equals_one}, and applying the definition~\eqref{eqn_projection_tensor}, we find that, for any $\bj\in [\ba]$,
\begin{align*}
\Pi^\ba_\bepsilon\ast E_\bj
&\spaceeq
E_\bepsilon\ast\Pi^\ba_\bepsilon\ast E_\bj
\spaceeq 1,
\end{align*}
as required.
\end{proof}
\noindent
The following description of the entries of the projection tensor is essentially a reformulation of~\cite[Lemma~34]{cz23soda:minions} in the notation of the current paper. We include the straightforward proof for completeness.
\begin{lem}
\label{lem_Pi_basic}
Given $a,b\in \N_0$, $\ba\in\N^a$, $\bell\in [a]^b$, and $\bi\in [\ba_\bell]$,
$
E_\bi\ast \Pi^\ba_\bell
=
\sum_{\bj\in [\ba],\;\bj_\bell=\bi}E_{\bj}$.
\end{lem}
\begin{proof}
If $b=0$, we have $\bell=\bi=\bepsilon$. Using Remark~\ref{rem_E_epsilon_equals_one} and Lemma~\ref{Pi_epsilon_all_one}, we find
\begin{align*}
E_\bepsilon\ast \Pi^\ba_\bepsilon
&\spaceeq
\Pi^\ba_\bepsilon
\spaceeq
\sum_{\bj\in [\ba]}E_\bj
\spaceeq
\sum_{\substack{\bj\in [\ba]\\\bj_\bepsilon=\bepsilon}}E_{\bj},
\end{align*}
as required.
Suppose now that $b\in\N$. In this case, we can assume that $a\in\N$ as $[0]^b=\emptyset^b=\emptyset$.
For any $\bj'\in [\ba]$, we have
\begin{align*}
\sum_{\substack{\bj\in [\ba]\\\bj_\bell=\bi}}E_{\bj}\ast E_{\bj'}
&\spaceeq
\sum_{\substack{\bj\in [\ba]\\\bj_\bell=\bi\\\bj=\bj'}}1
\spaceeq
\left\{
\begin{array}{lll}
1&\mbox{ if }\bj'_\bell=\bi\\
0&\mbox{ otherwise}
\end{array}
\right.
\spaceeq
E_\bi\ast\Pi^\ba_\bell\ast E_{\bj'},
\end{align*}
thus proving the result.
\end{proof}

Given a tensor $T\in\cT^{\ba}(\Q)$, we have from Lemma~\ref{lem_Pi_basic} and from the associativity rule of Lemma~\ref{lem_associativity_contraction} that, for $\bi\in [\ba_\bell]$, the $\bi$-th entry of $\Pi^\ba_\bell\ast T$ is 
\[
E_\bi\ast \Pi^\ba_\bell\ast T=\sum_{\bj\in [\ba],\;\bj_\bell=\bi}E_{\bj}\ast T;
\]
i.e., the sum of all entries of $T$ whose index $\bj$ projected onto $\bell$ gives $\bi$. In other words, contracting $T$ by $\Pi^\ba_\bell$ amounts to selecting a set of modes of $T$ (given by the tuple $\bell$) and projecting $T$ onto the hyperplane corresponding to those modes---whence the name ``projection tensor''.
In particular, if one lets $a=b=|\bell|$ in the definition of the projection tensor $\Pi^\ba_\bell$, contracting $T$ by $\Pi^\ba_\bell$  has the effect of permuting the modes of $T$. We call the resulting tensor $\Pi^\ba_\bell\ast T$ a \emph{reflection} of $T$.
For instance, for $\ba=(a_1,a_2)\in\N^2$, contracting by $\Pi^\ba_{(1,2)}$ results in the identity operator (cf.~Lemma~\ref{lem_identity_Pi} below), while contracting by $\Pi^\ba_{(2,1)}$ gives the transpose operator. Indeed, for any $a_1\times a_2$ matrix $M$, $\Pi^\ba_{(1,2)}\ast M=M$ and $\Pi^\ba_{(2,1)}\ast M=M^T$.

The assignment $\bell\mapsto\Pi^\ba_\bell$
creates a correspondence between 
tuples and projection tensors. Under this assignment, Lemma~\ref{lem_proj_contr} below shows that
the operation of tuple projection 
is translated 
into the operation of tensor contraction, while Lemma~\ref{lem_identity_Pi} shows that the tuple $\ang{a}$, that acts by projection as the identity on the set of tuples of appropriate length, corresponds to the projection tensor that acts by contraction as the identity on the space of tensors of appropriate size.
\begin{lem}
\label{lem_proj_contr}
Let $a,b,c\in\N_0$, and consider two tuples $\bell\in [a]^b$ and $\bbm\in [b]^c$. Then, for any $\ba\in\N^a$,
$
\Pi^\ba_{\bell_\bbm}=\Pi^{\ba_\bell}_\bbm \cont{b} \Pi^\ba_\bell.
$
\end{lem}
\begin{proof}
Take $\bi\in [\ba_{\bell_\bbm}]$ and $\bj'\in [\ba]$, and observe that
\begin{align*}
E_\bi\ast (\Pi^{\ba_\bell}_\bbm \cont{b} \Pi^\ba_\bell)\ast E_{\bj'}
&\lemeq{lem_associativity_contraction}
E_\bi\ast \Pi^{\ba_\bell}_\bbm \ast \Pi^\ba_\bell\ast E_{\bj'}
\lemeq{lem_Pi_basic}
\sum_{\substack{\bj\in [\ba_\bell]\\\bj_\bbm=\bi}}E_\bj\ast\Pi^\ba_\bell\ast E_{\bj'}
\spaceeq
\sum_{\substack{\bj\in [\ba_\bell]\\\bj_\bbm=\bi\\\bj'_\bell=\bj}}1\\
&\spaceeq
\left\{
\begin{array}{cl}
1&\mbox{ if }\bj'_{\bell_\bbm}=\bi\\
0&\mbox{ otherwise}
\end{array}
\right.
\spaceeq
E_\bi\ast\Pi^\ba_{\bell_\bbm}\ast E_{\bj'},
\end{align*}
whence the result follows.
\end{proof}
\begin{lem}
\label{lem_identity_Pi}
Let $a,b\in\N_0$, $\ba\in\N^a$, $\bb\in\N^{b}$, and $T\in\cT^{(\ba,\bb)}(\Q)$. Then
$
\Pi^\ba_{\ang{a}}\cont{a} T=T.
$
\end{lem}
\begin{proof}
For any $\bi\in [\ba]$, we find
\begin{align*}
E_\bi\ast (\Pi^\ba_{\ang{a}}\cont{a} T)
&\lemeq{lem_associativity_contraction}
E_\bi\ast \Pi^\ba_{\ang{a}}\ast T
\lemeq{lem_Pi_basic}
\sum_{\substack{\bj\in [\ba]\\\bj_{\ang{a}}=\bi}}E_\bj\ast T
\spaceeq
\sum_{\substack{\bj\in [\ba]\\\bj=\bi}}E_\bj\ast T
\spaceeq
E_\bi\ast T,
\end{align*}
as required.
\end{proof}

\section{The BA hierarchy through tensors}
\label{sec_BA_hierarchy_through_tensors}

When does $\BA^k(\X,\A)=\YES$?
In this section, we shall see that the acceptance problem for the $\BA$ hierarchy can be conveniently translated and studied in an algebraic---in fact, linear-algebraic---framework, through the notions of linear minions and tensorisation. The final result of this process, Theorem~\ref{thm_decoupling_BLPAIP_for_cliques}, will allow us to see $\BA^k$ acceptance (when the hierarchy is applied to AGC) as the problem of checking for the existence of some integer tensors satisfying certain geometric properties. This ``ultra-processed'' acceptance criterion will allow turning the quest for a fooling instance for $\BA^k$ (the goal of this paper) into the problem of building certain special hollow-shadowed crystal tensors---which will be accomplished in later sections.

\subsection{Relaxations and linear minions}
\label{subsec_relaxations_and_linear_minions}

All relaxation algorithms studied in the literature on CSPs and their promise variant are captured algebraically through the notion of linear minion, which we describe in this section. 

Given two integers $\ell,m\in\N$ and a function $\pi:[\ell]\to [m]$, let $P_\pi$ be the $m\times \ell$ matrix such that, for $i\in [m]$ and $j\in[\ell]$, the $(i,j)$-th entry of $P_\pi$ is $1$ if $\pi(j)=i$, and $0$ otherwise. 

\begin{defn}[\cite{cz23soda:minions}]
\label{defn_linear_minion}
A \emph{linear minion} $\Mminion$ of \emph{depth} 
${d}\in\N$ consists of the union of sets $\Mminion^{(\ell)}$ of $\ell\times d$ rational matrices
for $\ell\in\N$, that satisfy the following condition: $P_\pi M\in\Mminion^{(m)}$ whenever $\ell,m\in\N$, $\pi:[\ell]\to [m]$, and $M\in\Mminion^{(\ell)}$.\footnote{The definition of linear minions we give here is less general than the one in~\cite[Definition~16]{cz23soda:minions}, which includes linear minions of infinite depth and whose matrices have entries in arbitrary semirings rather than $\Q$.}
\end{defn}
Observe that pre-multiplying a matrix $M$ by $P_\pi$ amounts to performing a combination of the following three elementary operations to the rows of $M$: swapping two rows, replacing two rows with their sum, and inserting a zero row. Hence, a linear minion is simply a set of matrices having a fixed number of columns that is closed under such elementary operations. 

\begin{example}
\label{example_Qconv_Zaff}
For each $\ell\in\N$, let 
\begin{itemize}
    \item 
    $\Qconv^{(\ell)}$ be the set of rational vectors of length $\ell$ whose entries are nonnegative and sum up to $1$,
    \item
    $\Zaff^{(\ell)}$ be the set of integer vectors of length $\ell$ whose (possibly negative) entries sum up to $1$, and
    \item
    $\BAminion^{(\ell)}$ be the set of $\ell\times 2$ matrices whose left column $\bv$ belongs to $\Qconv^{(\ell)}$, whose right column $\bw$ belongs to $\Zaff^{(\ell)}$, and such that, for each $i\in [\ell]$, $v_i=0$ implies $w_i=0$.
\end{itemize}
Using that $\bone_m^T P_\pi=\bone_\ell^T$ for each $\pi:[\ell]\to [m]$, we easily check that $\Qconv=\bigcup_{\ell\in\N}\Qconv^{(\ell)}$ and $\Zaff=\bigcup_{\ell\in\N}\Zaff^{(\ell)}$ are both linear minions of depth $1$, while $\BAminion=\bigcup_{\ell\in\N}\BAminion^{(\ell)}$ is a linear minion of depth $2$.
\end{example}
In order to be consistent with the notation of~\cite[Definition~5]{bgwz20},
given a linear minion $\Mminion$, a function $\pi:[\ell]\to [m]$, and a matrix $M\in\Mminion^{(\ell)}$, 
we shall often denote the product $P_\pi M$ by the notation $M_{/\pi}$.

\begin{rem}
For two maps $\pi:[\ell]\to [m]$ and $\sigma:[m]\to [p]$, we easily check that $P_{\sigma\circ\pi}=P_\sigma P_\pi$. As a consequence,
\begin{align}
\label{eqn_composition_minors}
    M_{/\sigma\circ\pi}
    \spaceeq
    (M_{/\pi})_{/\sigma}.
\end{align}
Also, if $\operatorname{id}$ is the identity function on $[\ell]$,
  $P_{\operatorname{id}}$ is the identity matrix of size $\ell\times\ell$, so $M_{/\operatorname{id}}=M$. This shows that linear minions form a subclass of the so-called \emph{abstract minions} (or simply \emph{minions}) introduced in~\cite{bgwz20} (see also~\cite{BBKO21}).
\end{rem}

Each linear minion corresponds to a relaxation for (P)CSPs through the notion of free structure. Intuitively, the free structure of a linear minion $\Mminion$ generated by a hypergraph $\HH$ simulates the structure of $\HH$ inside $\Mminion$: The vertices become matrices of $\Mminion$, while the hyperedges are tuples of matrices that can all be obtained from a single other matrix through elementary row operations.
The formal definition
is given below. We define the free structure for uniform hypergraphs rather than digraphs, because we will later use it in that more general case. In fact, the same construction can be applied to arbitrary relational structures, see~\cite[Definition~4.1]{BBKO21}.
\begin{defn}[\cite{BBKO21}]
\label{defn_minion_free_structure}
Let $\HH$ be a $p$-uniform hypergraph having $n$ vertices and $m$ hyperedges. Without loss of generality, let the domain of $\HH$  be $[n]$.
The \emph{free structure} $\freeM(\HH)$ of a linear minion $\Mminion$ generated by $\HH$ is the (potentially infinite) $p$-uniform hypergraph on the vertex set $\Vset(\freeM(\HH))=\Mminion^{(n)}$ whose hyperedges are defined as follows: Given $M_1,\dots,M_p\in \Mminion^{(n)}$, the tuple $(M_1,\dots,M_p)$ belongs to $\Eset(\freeM(\HH))$ if and only if there exists some $Q\in\Mminion^{(m)}$ such that $M_i=Q_{/\pi_i}$ for each $i\in[p]$, where $\pi_i:\Eset(\HH)\to \Vset(\HH)$ maps a hyperedge $\bh$ to its $i$-th entry $h_i$.
\end{defn}

Take a linear minion $\Mminion$ and two digraphs $\X$ (the instance) and $\A$ (the template). The relaxation corresponding to $\Mminion$ outputs $\YES$ if $\X\to\freeM(\A)$ and $\NO$ otherwise.\footnote{In~\cite{cz23soda:minions}, this relaxation was described as the ``minion test'' associated with $\Mminion$.}
For certain linear minions, the problem of deciding whether $\X\to\freeM(\A)$ can be solved in polynomial time (in the size of the input $\X$) for any $\A$. In particular, this is the case for the linear minions $\Qconv$, $\Zaff$, and $\BAminion$ from Example~\ref{example_Qconv_Zaff}. It was shown in~\cite{BBKO21} that $\Qconv$ and $\Zaff$ correspond to the polynomial-time relaxations $\BLP$ and $\AIP$, respectively, while it was shown in~\cite{bgwz20} that $\BAminion$ corresponds to the polynomial-time relaxation $\BA$.

In~\cite{cz23soda:minions}, a class of linear minions enjoying particularly desirable features was identified.

\begin{defn}[\cite{cz23soda:minions}]
\label{defn_conic_minion}
A \emph{conic minion} $\Mminion$ is a linear minion of depth $d$ such that $(i)$ $\Mminion$ does not contain any all-zero matrix, and $(ii)$ for every $\ell\in\N$, every $M\in\Mminion^{(\ell)}$,
and every $V\subseteq [\ell]$, the following implication is true:
\begin{align*}
\begin{array}{ll}
\sum_{i\in V}E_i\ast M=\bzero_d \quad\Rightarrow\quad E_i\ast M=\bzero_d\;\;\forall i\in V.
\end{array}
\end{align*}
\end{defn}
\noindent In other words, a linear minion $\Mminion$ is conic if it does not contain all-zero matrices and if summing up nonzero rows of a matrix in $\Mminion$ does not yield the all-zero vector.

\begin{example}
\label{Qconv_Zaff_conic_not}
It is not hard to check that $\Qconv$ and $\BAminion$ are conic, while $\Zaff$ is not (cf.~\cite{cz23soda:minions}).
\end{example}

\noindent The following property of the entries of $P_\pi$ is a reformulation of~\cite[Lemma~30]{cz23soda:minions} and shall prove useful on multiple occasions. We include the simple proof for completeness.
\begin{lem}
\label{lem_calculation_rule_P}
Let $\ell,m\in\N$, let $\pi:[\ell]\to [m]$, and let $i\in [m]$. Then $E_i\ast P_\pi=\sum_{j\in\pi^{-1}(i)}E_j$.
\end{lem}
\begin{proof}
For any $z\in [\ell]$, we have 
\begin{align*}
    \sum_{j\in\pi^{-1}(i)}E_j\ast E_z
    \spaceeq
    \left\{
    \begin{array}{cl}
         1&\mbox{ if }z\in\pi^{-1}(i)  \\
         0&\mbox{ otherwise} 
    \end{array}
    \right.
    \spaceeq
    \left\{
    \begin{array}{cl}
         1&\mbox{ if }\pi(z)=i  \\
         0&\mbox{ otherwise} 
    \end{array}
    \right.
    \spaceeq
    E_i\ast P_\pi\ast E_z,
\end{align*}
which means that $\sum_{j\in\pi^{-1}(i)}E_j=E_i\ast P_\pi$, as required.
\end{proof}

\subsection{Hierarchies and tensors}

The framework developed in~\cite{cz23soda:minions} allows to progressively strengthen the relaxation corresponding to any linear minion (called ``minion test'' therein) through the notion of tensor power of a digraph (given in~\cite[Definition~10]{cz23soda:minions} for the more general case of relational structures).

\begin{defn}[\cite{cz23soda:minions}]
\label{defn_tensorisation}
Given $k\in\N$, the \emph{$k$-th tensor power} of a digraph $\A$ is the $2^k$-uniform hypergraph $\Ak$ having vertex set $\Vset(\Ak)=\Vset(\A)^k$ and hyperedge set $\Eset(\Ak)=\{\ba^\tensor{k}:\ba\in \Eset(\A)\}$ where, for $\ba\in \Eset(\A)$, $\ba^\tensor{k}$ is the tensor
in $\cT^{2\cdot\bone_k}(\Vset(\A)^k)$ whose $\bi$-th entry is $\ba_\bi$ for every $\bi\in [2]^k$.
\end{defn}

\noindent 
Let us see what happens when we take the free structure generated by the tensor power of a digraph.

\begin{rem}
\label{rem_description_freestructure_tensorisation}
Let $\Mminion$ be a linear minion of depth $d$ and let $\A$ be a digraph with $n$ vertices\footnote{Here and throughout the rest of the paper, we shall often assume that the vertex set of the digraph $\A$ is $[n]$.} and $m$ edges. Just like $\Ak$, $\freeM(\Ak)$ is a $2^k$-uniform hypergraph. Its vertex set is $\Vset(\freeM(\Ak))=\Mminion^{(n^k)}$. Hence, the vertices of $\freeM(\Ak)$ are $n^k\times d$ rational matrices; it will be convenient to identify them with tensors in $\cT^{(n\cdot\bone_k,d)}(\Q)$. A family $\{M^{(\bi)}\}_{\bi\in [2]^k}$ of vertices (i.e., of tensors in $\Vset(\freeM(\Ak))$) forms a hyperedge if and only if there exists some matrix $Q\in\Mminion^{(m)}$ such that $M^{(\bi)}=Q_{/\pi_\bi}$ for each $\bi\in [2]^k$, where $\pi_\bi:\Eset(\A)\to\Vset(\A)^k$ maps $\ba\in\Eset(\A)$ to $\ba_\bi$. Note that $Q_{/\pi_\bi}$ can be expressed as a contraction by the multilinear version of the matrix $P_{\pi_i}$ associated with the map $\pi_i$ from Definition~\ref{defn_minion_free_structure};
i.e., $Q_{/\pi_\bi}=P_{\pi_\bi}\cont{1}Q$, where $P_{\pi_\bi}\in\cT^{(n\cdot\bone_k,m)}(\Q)$ is the tensor whose $(\ba,\bb)$-th entry is $1$ if $\bb_\bi=\ba$ and $0$ otherwise, for $\ba\in\Vset(\A)^k$ and $\bb\in\Eset(\A)$.
\end{rem}

The strategy introduced in~\cite{cz23soda:minions} for strengthening a minion test consists in applying the test to the tensor powers of both the instance and the template---with one extra technicality: The homomorphism certifying acceptance of the relaxation thus obtained should be compatible with the tensorised structures, in the sense of Definition~\ref{defn_tensorial_homo}.

\begin{defn}
\label{defn_tensorial_homo}
Let $\Mminion$ be a linear minion, let $k\in\N$, and let $\X,\A$ be two digraphs. We say that a homomorphism $\xi:\Xk\to\freeM(\Ak)$ is \emph{$k$-tensorial} if $\xi(\bx_\bi)=\Pi^{n\cdot\bone_k}_\bi\cont{k}\xi(\bx)$ for any $\bx\in \Vset(\X)^k$, $\bi\in [k]^k$.
\end{defn}

In other words, a $k$-tensorial homomorphism
translates the operation of tuple projection into the operation of tensor projection---where the latter is expressed as contraction by the projection tensor $\Pi^{n\cdot\bone_k}_\bi$ introduced in Section~\ref{subsec_projection_tensor}.

Given a linear minion $\Mminion$ and an integer $k\in\N$, the $k$-th level of the relaxation induced by $\Mminion$ is defined as follows: For any pair of digraphs $\X$ (the instance) and $\A$ (the template), it outputs $\YES$ if there exists a $k$-tensorial homomorphism $\Xk\to\freeM(\Ak)$ and $\NO$ otherwise.
It was shown in~\cite{cz23soda:minions} that both the $\BLP$ and the $\AIP$ hierarchies fit into this framework, in the sense that, for any two digraphs $\X,\A$ and any integer $k\in\N$, $\BLP^k(\X,\A)=\YES$ ($\AIP^k(\X,\A)=\YES$) if and only if there exists a $k$-tensorial homomorphism $\Xk\to\freeQ(\Ak)$ ($\Xk\to\freeZ(\Ak)$).
A similar characterisation was also established for the $\BA$ hierarchy we consider in this work (see~\cite[Theorem~15]{cz23soda:minions}). Moreover, using that the minion $\BAminion$ capturing the $\BA$ hierarchy is the \emph{semi-direct product} of the two minions $\Qconv$ and $\Zaff$, it was shown in~\cite[Proposition~44]{cz23soda:minions} that any $k$-tensorial homomorphism $\Xk\to\freeBA(\Ak)$ can be split into homomorphisms to the free structures of $\Qconv$ and $\Zaff$, separately. These results are summarised in the next theorem.
\begin{thm}[\cite{cz23soda:minions}]
\label{thm_acceptance_BA_hierarchy_general}
    Let $\X$ and $\A$ be digraphs and let $2\leq k\in\N$. The following are equivalent:
    \begin{itemize}
        \item $\BA^k(\X,\A)=\YES$;
        \item there exists a $k$-tensorial homomorphism from $\Xk$ to $\freeBA(\Ak)$;
        \item
        there exist $k$-tensorial homomorphisms $\xi:\Xk\to\freeQ(\Ak)$ and $\zeta:\Xk\to\freeZ(\Ak)$ such that $\supp(\zeta(\bx))\subseteq\supp(\xi(\bx))$ for any $\bx\in \Vset(\X)^k$.
    \end{itemize}
\end{thm}

\begin{rem}
    It was shown in~\cite[Proposition~36]{cz23soda:minions} that the existence of a $k$-tensorial homomorphism from $\Xk$ to $\freeM(\Ak)$ is equivalent to the existence of a homomorphism from $\tilde{\X}^\tensor{k}$ to $\freeM(\tilde{\A}^\tensor{k})$, where $\tilde{\X}$ and $\tilde{\A}$ are obtained from $\X$ and $\A$ by \emph{$k$-enhancing} them, i.e., by adding to their signatures an extra relation that includes all tuples of length $k$. We prefer to adopt the description in terms of $k$-tensorial homomorphisms, as $k$-enhancing a digraph results in a structure having two different relations, while in this work we only consider structures with one relation (digraphs or hypergraphs). We also remark that the term ``$k$-tensorial'' does not appear in~\cite{cz23soda:minions}.
\end{rem}

\subsection{BA$^k$ acceptance for AGC}

The goal of this work is to show that no level of the $\BA$ hierarchy solves the approximate graph colouring problem $\PCSP(\K_c,\K_d)$. To that end, we need to find instances $\X$ that are able to fool the hierarchy, i.e., such that $\BA^k(\X,\K_c)=\YES$ but $\X$ is not $d$-colourable. 
It turns out that, for the particular case that the $\BA$ hierarchy is applied to the colouring problem (i.e., when $\A$ is a clique), the acceptance criterion of Theorem~\ref{thm_acceptance_BA_hierarchy_general} can be simplified: As stated in Theorem~\ref{thm_decoupling_BLPAIP_for_cliques}, it is enough to check for the existence of a $k$-tensorial homomorphism $\zeta$ from $\Xk$ to $\freeZ(\Ak)$ that satisfies a simple combinatorial condition. The reason why one does not have to explicitly verify the existence of a homomorphism $\xi$ to $\freeQ(\Ak)$, too, is that, when the size of the clique $\A$ is at least $k$, there exists a standard $k$-tensorial homomorphism $\xi_0$ from $\Xk$ to $\freeQ(\Ak)$ that gives equal weight to \emph{all} admitted assignments---equivalently, the tensors that are images of elements of $\Xk$ under $\xi_0$ are uniform within their admitted support. This homomorphism is ``as good as possible'' for our purposes, in the sense that it makes the support of $\xi_0(\bx)$ as large as it can be, thus leaving more room for the existence of some $\zeta$ satisfying the refinement condition $\supp(\zeta(\bx))\subseteq\supp(\xi(\bx))$. In other words, whenever a pair of $k$-tensorial homomorphisms $(\xi,\zeta)$ certifying $\BA^k$ acceptance exists, the pair $(\xi_0,\zeta)$ also works. As it will later become more clear, thanks to the criterion given in Theorem~\ref{thm_decoupling_BLPAIP_for_cliques}, we can view $\BA^k$ acceptance in terms of the existence of a family of integer tensors satisfying a system of symmetries (dictated by the fact that $\zeta$ needs to be a $k$-tensorial homomorphism) together with a ``hollowness requirement'' expressed through the extra combinatorial condition. The hollow-shadowed crystals we shall seek in the next section will generate a family of such tensors.

The proof of Theorem~\ref{thm_decoupling_BLPAIP_for_cliques} makes use of two technical lemmas that we present next. The first is a special case of~\cite[Lemma~32]{cz23soda:minions}.
Recall the definition of the symbols ``$\prec$'' and ``$\sim$'' given in Section~\ref{subsec_tuples}.

\begin{lem}[\cite{cz23soda:minions}]
\label{cor_vanishing_contractions}
Let $\Mminion$ be a linear minion of depth $d$, let $k\in\N$, let $\X,\A$ be two digraphs, and let $\xi:\X^{\tensor{k}}\to\freeM(\Ak)$ be a $k$-tensorial homomorphism. Then $E_\ba\ast \xi(\bx)=\bzero_d$ for any $\bx\in \Vset(\X)^k$ and $\ba\in \Vset(\A)^k$ for which $\bx\not\prec\ba$.  
\end{lem}
\noindent Crucially, Lemma~\ref{cor_vanishing_contractions} does not require that the linear minion be conic. In the proof of Theorem~\ref{thm_decoupling_BLPAIP_for_cliques}, we shall apply this lemma to the (non-conic) minion $\Zaff$.

\begin{lem}
\label{lem_involved_1511_1718}
Let $k\leq n\in\N$, let $X$ be a set, and consider the tuples $\bx\in X^k$, $\bi\in [k]^k$, and $\ba\in [n]^k$. Then
\begin{align*}
|\{\bb\in [n]^k: \bb_\bi=\ba\mbox{ and }\bb\sim\bx\}|=
\left\{
\begin{array}{cl}
\frac{(n-|\bx_\bi|)!}{(n-|\bx|)!} & \mbox{if } \ba\sim\bx_\bi\\
0 & \mbox{otherwise.}
\end{array}
\right.
\end{align*}
\end{lem}
\begin{proof}
Assume first that $\ba\sim \bx_\bi$. Note that there exists a bijection $\vartheta$ between the set $\{\bb\in[n]^k: \bb\sim\bx\}$
  and the set of injective functions from $\set(\bx)$ to $[n]$.
  (Indeed, $\bb\sim\bx$ means that $b_p=b_q$ if and only if $x_p=x_q$ for every
  $p,q\in[k]$.) Now, if $\bb_\bi=\ba\sim \bx_\bi$, the restriction of $\vartheta(\bb)$ to $\set(\bx_\bi)$ is entirely determined by $\ba$. The remaining values of
  $\vartheta(\bb)$ can be chosen in 
\begin{align*}
(n-|\bx_\bi|)\cdot(n-|\bx_\bi|-1)\cdots(n-|\bx|+1) \spaceeq \frac{(n-|\bx_\bi|)!}{(n-|\bx|)!}
\end{align*}
distinct ways, thus proving the first case in the statement of the lemma.

Assume now that $\ba\not\sim \bx_\bi$. By definition, if $\bb\sim \bx$, then $\bb_\bi\sim \bx_\bi$. Thus, if $\bb_\bi=\ba$ and $\bb\sim\bx$, then $\ba\sim \bx_\bi$, a contradiction. This proves the second case in the statement of the lemma.
\end{proof}

\begin{thm*}[Theorem~\ref{thm_decoupling_BLPAIP_for_cliques} restated]
Let $2\leq k\leq n\in\N$, let $\X$ be a loopless digraph, and let $\zeta:\Xk\to\freeZ(\K_n^\tensor{k})$ be a $k$-tensorial homomorphism such that $E_\ba\ast\zeta(\bx)=0$ for any $\bx\in \Vset(\X)^k$ and $\ba\in [n]^k$ for which $\ba\not\prec\bx$. 
Then $\BA^k(\X,\K_n)=\YES$.
\end{thm*}
\begin{proof}
For $\bx\in \Vset(\X)^k$, consider the tensor $T_\bx\in\cT^{n\cdot\bone_k}(\Q)$ defined by 
\begin{align*}
E_\ba\ast T_\bx=
\left\{
\begin{array}{cl}
1 & \mbox{if }\ba\sim\bx\\
0 & \mbox{otherwise}
\end{array}
\right.
\hspace{1cm}
\forall \ba\in [n]^k.
\end{align*}
We shall prove that the function
\begin{align*}
\xi:\Vset(\X)^k&\to \mathcal{T}^{n\cdot\bone_k}(\Q)\\
\bx&\mapsto \frac{1}{\Pi^{n\cdot\bone_k}_\bepsilon\ast T_\bx}T_\bx
\end{align*}
yields a $k$-tensorial homomorphism from $\Xk$ to $\freeQ(\K_n^\tensor{k})$. First, observe that $\xi$ is well defined as, using that $k\leq n$,
\begin{align}
\label{eqn_1255_28sep}
    \Pi^{n\cdot\bone_k}_\bepsilon\ast T_\bx
    \lemeq{Pi_epsilon_all_one}
    \sum_{\ba\in [n]^k}E_\ba\ast T_\bx
    \spaceeq
    |\{\ba\in [n]^k:\ba\sim\bx\}|
    \spaceeq
    \frac{n!}{(n-|\bx|)!}
\end{align}
which is not zero.
Moreover, we have that $\xi(\bx)\in\Qconv^{(n^k)}$ since
\begin{align*}
    \Pi^{n\cdot\bone_k}_{\bepsilon}\ast \xi(\bx)
    \spaceeq
    \frac{\Pi^{n\cdot\bone_k}_\bepsilon\ast T_\bx}{\Pi^{n\cdot\bone_k}_\bepsilon\ast T_\bx}
    \spaceeq
    1.
\end{align*}
We now prove that $\xi$ sends hyperedges of $\Xk$ to hyperedges of $\freeQ(\K_n^\tensor{k})$. Take $(x,y)\in \Eset(\X)$, so $(x,y)^\tensor{k}\in\Eset(\Xk)$; since $\X$ is loopless, $x\neq y$. Observe that $|\Eset(\K_n)|=n^2-n$. Take $Q=\frac{1}{n^2-n}\cdot\bone_{n^2-n}\in\Qconv^{(n^2-n)}$; we claim that $\xi((x,y)_\bi)=Q_{/\pi_\bi}$ for each $\bi\in [2]^k$, which then implies that $\xi((x,y)^\tensor{k})\in\Eset(\freeQ(\K_n^\tensor{k}))$, as needed. For $\bi\in [2]^k$ and $\ba\in [n]^k$, we have
\begin{align}
\label{eqn_1246_28sep}
\notag
    E_\ba\ast Q_{/\pi_\bi}
    &\spaceeq
    E_\ba\ast P_{\pi_\bi}\ast Q
    \spaceeq
    \frac{1}{n^2-n}E_\ba\ast P_{\pi_\bi}\ast\bone_{n^2-n}
    \spaceeq
    \frac{1}{n^2-n}\sum_{(a',b')\in\Eset(\K_n)}E_\ba\ast P_{\pi_\bi}\ast E_{(a',b')}\\
    &\spaceeq
    \frac{1}{n^2-n}|\{(a',b')\in\Eset(\K_n):(a',b')_\bi=\ba\}|.
\end{align}
Suppose that $\bi=\bone_k$. In this case,~\eqref{eqn_1246_28sep} yields
\begin{align*}
    E_\ba\ast Q_{/\pi_\bi}
    &\spaceeq
    \frac{1}{n^2-n}|\{(a',b')\in\Eset(\K_n):(a',\dots,a')=\ba\}|
    \spaceeq
    \left\{
    \begin{array}{cl}
        \frac{1}{n} & \mbox{ if }\ba \mbox{ is constant} \\
        0 & \mbox{ otherwise.} 
    \end{array}
    \right.
\end{align*}
On the other hand,
\begin{align*}
    E_\ba\ast\xi((x,y)_\bi)
    &\spaceeq
    E_\ba\ast\xi((x,\dots,x))
    \spaceeq
    \frac{1}{\Pi^{n\cdot\bone_k}_\bepsilon\ast T_{(x,\dots,x)}}E_\ba\ast T_{(x,\dots,x)}
    \equationeq{eqn_1255_28sep}
    \frac{(n-1)!}{n!}E_\ba\ast T_{(x,\dots,x)}\\
    &\spaceeq
    \left\{
    \begin{array}{cl}
        \frac{1}{n} & \mbox{ if }\ba \mbox{ is constant} \\
        0 & \mbox{ otherwise.} 
    \end{array}
    \right.
\end{align*}
Hence, the claim holds in this case. The case $\bi=2\cdot\bone_k$ follows analogously. Suppose now that $|\bi|=2$. In this case,~\eqref{eqn_1246_28sep} yields
\begin{align*}
    E_\ba\ast Q_{/\pi_\bi}
    &\spaceeq
    \left\{
    \begin{array}{cl}
        \frac{1}{n^2-n} & \mbox{ if }\ba\sim\bi  \\
        0 & \mbox{ otherwise.} 
    \end{array}
    \right.
\end{align*}
On the other hand,
\begin{align*}
    E_\ba\ast\xi((x,y)_\bi)
    &\spaceeq
    \frac{1}{\Pi^{n\cdot\bone_k}_\bepsilon\ast T_{(x,y)_\bi}}E_\ba\ast T_{(x,y)_\bi}
    \equationeq{eqn_1255_28sep}
    \frac{(n-2)!}{n!}E_\ba\ast T_{(x,y)_\bi}
    \spaceeq
    \left\{
    \begin{array}{cl}
        \frac{1}{n^2-n} & \mbox{ if }\ba\sim (x,y)_\bi \\
         0 & \mbox{ otherwise.} 
    \end{array}
    \right.
\end{align*}
Using that $(x,y)_\bi\sim\bi$ and that ``$\sim$'' is transitive, we conclude that the claim holds in this case, too.
It follows that $\xi$ is a homomorphism from $\Xk$ to $\freeQ(\K_n^\tensor{k})$.
To show that $\xi$ is $k$-tensorial, consider three tuples $\bx\in \Vset(\X)^k$, $\bi\in [k]^k$, and $\ba\in [n]^k$, and observe that
\begin{align*}
    E_\ba\ast \Pi^{n\cdot\bone_k}_\bi\ast \xi(\bx)
    &\lemeq{lem_Pi_basic}
    \sum_{\substack{\bb\in [n]^k\\\bb_\bi=\ba}}E_\bb\ast\xi(\bx)
    \spaceeq
    \frac{1}{\Pi^{n\cdot\bone_k}_\bepsilon\ast T_\bx}\sum_{\substack{\bb\in [n]^k\\\bb_\bi=\ba}}E_\bb\ast T_\bx
    \equationeq{eqn_1255_28sep}
    \frac{(n-|\bx|)!}{n!}\sum_{\substack{\bb\in [n]^k\\\bb_\bi=\ba}}E_\bb\ast T_\bx\\
    &\spaceeq
    \frac{(n-|\bx|)!}{n!}|\{\bb\in [n]^k:\bb_\bi=\ba \mbox{ and } \bb\sim\bx\}|\\
    &\lemeq{lem_involved_1511_1718}
    \left\{
    \begin{array}{cl}
        \frac{(n-|\bx|)!}{n!}\cdot\frac{(n-|\bx_\bi|)!}{(n-|\bx|)!} & \mbox{ if } \ba\sim\bx_\bi \\
        0 & \mbox{ otherwise.} 
    \end{array}
    \right.
    \spaceeq
    \left\{
    \begin{array}{cl}
        \frac{(n-|\bx_\bi|)!}{n!} & \mbox{ if } \ba\sim\bx_\bi \\
        0 & \mbox{ otherwise.} 
    \end{array}
    \right.
\intertext{On the other hand,}
    E_\ba\ast\xi(\bx_\bi)
    &\spaceeq
    \frac{1}{\Pi^{n\cdot\bone_k}_\bepsilon\ast T_{\bx_\bi}}E_\ba\ast T_{\bx_\bi}
    \equationeq{eqn_1255_28sep}
    \frac{(n-|\bx_\bi|)!}{n!}E_\ba\ast T_{\bx_\bi}
    \spaceeq
    \left\{
    \begin{array}{cl}
        \frac{(n-|\bx_\bi|)!}{n!} & \mbox{ if }\ba\sim\bx_\bi\\
        0 & \mbox{ otherwise.} 
    \end{array}
    \right.
\end{align*}
It follows that $\xi(\bx_\bi)=\Pi^{n\cdot\bone_k}_\bi\ast \xi(\bx)$, which means that $\xi$ is $k$-tensorial.

Take $\bx\in \Vset(\X)^k$ and $\ba\in [n]^k$, and suppose that $E_\ba\ast\xi(\bx)=0$. This implies $E_\ba\ast T_\bx=0$, which means that $\ba\not\sim\bx$; i.e., either $\ba\not\prec\bx$ or $\bx\not\prec\ba$. Using the hypothesis of the theorem (in the former case) or Lemma~\ref{cor_vanishing_contractions} applied to $\zeta$ (in the latter case), we find that $E_\ba\ast\zeta(\bx)=0$.
It follows that $\supp(\zeta(\bx))\subseteq\supp(\xi(\bx))$ for any $\bx\in \Vset(\X)^k$. By virtue of Theorem~\ref{thm_acceptance_BA_hierarchy_general}, this 
implies that $\BA^k(\X,\K_n)=\YES$. 
\end{proof}

\section{Crystals}
\label{sec_crystals}
In Section~\ref{sec_BA_hierarchy_through_tensors}, we obtained a multilinear criterion for the acceptance of the $\BA$ hierarchy applied to AGC: According to Theorem~\ref{thm_decoupling_BLPAIP_for_cliques}, to have $\BA^k(\X,\K_n)=\YES$ it suffices to find a $k$-tensorial homomorphism $\zeta$ from $\Xk$ to $\freeZ(\K_n^\tensor{k})$ satisfying the extra condition 
\begin{align}
\label{eqn_1536_1nov}
\ba\not\prec\bx\;\;\;\;\Rightarrow\;\;\;\;E_\ba\ast\zeta(\bx)=0.    
\end{align}
(Note that, by virtue of Lemma~\ref{cor_vanishing_contractions}, the condition ``$\ba\not\prec\bx$'' might be replaced with ``$\ba\not\sim\bx$''.)
It follows from Remark~\ref{rem_description_freestructure_tensorisation} that $\freeZ(\K_n^\tensor{k})$ is a $2^k$-uniform infinite hypergraph whose vertices are elements of $\cT^{n\cdot\bone_k}(\Z)$, i.e., $k$-dimensional integer cubical tensors of width $n$, whose entries sum up to $1$. As for the hyperedges, a family $\{T^{(\bi)}\}_{\bi\in [2]^k}$ of $2^k$ such tensors  forms a hyperedge if and only if there exists an integer vector $\bq$ of length $n^2-n=|\Eset(\K_n)|$ (i.e., an integer distribution over the edges of $\K_n$) whose entries sum up to $1$ and such that all tensors in the family can be obtained from $\bq$ by specific contractions; more precisely, we require that $T^{(\bi)}=\bq_{/\pi_\bi}=P_{\pi_\bi}\ast\bq$ for each $\bi\in [2]^k$. 
\begin{defn}
\label{defn_affine_tensors}
Let $q\in\N_0$, let $\bn\in\N^q$, and let $T\in \mathcal{T}^\bn(\Z)$. We say that $T$ is \emph{affine} if $\Pi^\bn_\bepsilon\ast T=1$.
\end{defn}
Hence, finding a homomorphism $\zeta$ from $\Xk$ to $\freeZ(\K_n^\tensor{k})$ means selecting some $k$-dimensional integer affine cubical tensors of width $n$ (one for each tuple $\bx\in\Vset(\X)^k$) in such a way that the hyperedge relation is preserved. In order for $\zeta$ to be $k$-tensorial, this family of tensors needs to behave well with respect to projections: The tensor associated with the (combinatorial) projection of a tuple $\bx$ of vertices onto a tuple $\bi\in [k]^k$ should be the (geometric) projection of the tensor associated with $\bx$ onto the hyperplane generated by $\bi$; in symbols, $\zeta(\bx_\bi)=\Pi^{n\cdot\bone_k}_\bi\ast\zeta(\bx)$. 
One way to build a family of tensors having this property is to consider the $k$-dimensional projections of a single higher-dimensional affine cubical tensor $C$ of width $n$, whose dimension $q$ is the number of vertices of $\X$. Specifically,
we build a map $\zeta_C$ associated with the tensor $C$ as follows: The image of a tuple $\bx\in\Vset(\X)^k$ under $\zeta_C$ is the projection of $C$ onto the hyperplane generated by $\bx$; i.e., the tensor $\Pi^{n\cdot\bone_q}_\bx\ast C$. In this way, $\zeta_C$ is automatically $k$-tensorial. Indeed, Lemma~\ref{lem_proj_contr} and Lemma~\ref{lem_associativity_contraction} imply that $$\zeta_C(\bx_\bi)=\Pi^{n\cdot\bone_q}_{\bx_\bi}\ast C=\Pi^{n\cdot\bone_k}_{\bi}\ast\Pi^{n\cdot\bone_q}_\bx\ast C=\Pi^{n\cdot\bone_k}_{\bi}\ast\zeta_C(\bx),$$ as needed. 

For the map $\zeta_C$ to yield a homomorphism from $\Xk$ to $\freeZ(\K_n^\tensor{k})$, it is enough to require that the $2$-dimensional projections of $C$ be equal up to taking the transpose and have zero diagonal (cf.~the proof of Theorem~\ref{acceptance_prop_BA_big_enough} in Section~\ref{sec_fooling_BA}). Since a cubical tensor $C$ of width $n$ and dimension $q$ with this property exists for all choices of $n\geq 3$ and $q$, \emph{any} loopless digraph $\X$ is accepted by \emph{any} level of the $\AIP$ hierarchy applied to the template $\K_n$ for \emph{any} $n\geq 3$---whence it follows that to fool any level of the $\AIP$ hierarchy applied to $\PCSP(\K_c,\K_d)$ one can simply take the clique $\K_{d+1}$ (cf.~\cite{cz23soda:aip}).

This clearly cannot be true for the stronger $\BA$ hierarchy that, unlike $\AIP$, is sound in the limit. The obstruction is the condition~\eqref{eqn_1536_1nov}. The goal is then to identify a class of more refined tensors $C$ such that the associated homomorphism $\zeta_C$ satisfies the above condition. To this end, we start by enforcing a stronger requirement on the projections on $C$: \emph{The $k$-dimensional} (as opposed to $2$-dimensional) \emph{projections of $C$ should coincide}. 
Note that we cannot require that \emph{all} such projections be equal.
Indeed, already for $k=2$, if a matrix $M$ is the projection of $C$ onto some $2$-dimensional plane $xy$, then the projection of $C$ onto the reflected plane $yx$ is $M^T$. If these two projections need to be equal, it follows that $M$ must be symmetric.
In addition, $M$ is required to be affine and have zero diagonal, which clearly leads to a contradiction.
We then relax the hypothesis, by requiring that 
only the \emph{oriented} $k$-dimensional projections be equal. We say that a tensor having this property is a \emph{crystal}, as we next define.

Given $q,k\in\N$, we let $[q]^k_\rightarrow$ denote the set of increasing tuples in $[q]^k$; i.e., $[q]^k_\rightarrow=\{(i_1,\dots,i_k)\in [q]^k \mbox{ s.t. } i_1<i_2<\dots<i_k\}$. We also set $[q]^0_\to=\{\bepsilon\}$. Observe that $[q]^k_\to\neq\emptyset$ if and only if $k\leq q$.

\begin{defn*}[Formal version of Definition~\ref{defn_crystal_informal}]
Let $q,n\in\N$ and $k\in\{0,\dots,q\}$. A cubical tensor $C\in\mathcal{T}^{n\cdot\bone_q}(\Z)$ is a \emph{$k$-crystal} if $\Pi^{n\cdot\bone_q}_\bi\ast C=\Pi^{n\cdot\bone_q}_\bj\ast C$
for each $\bi,\bj\in [q]^k_\to$. In this case, the \emph{$k$-shadow} of $C$ is the tensor $\Pi^{n\cdot\bone_q}_\bi\ast C$ (for some $\bi\in [q]^k_\to$).
\end{defn*}

\begin{rem}
Given a not necessarily increasing tuple $\bj\in[q]^k$, we can always find two tuples $\bi\in [q]^k_\to$ and $\bell\in [k]^k$ for which $\bj=\bi_\bell$. Then, if $S$ is the $k$-shadow of a $k$-crystal $C$, we obtain \begin{align*}
    \Pi^{n\cdot\bone_q}_\bj\ast C
    \spaceeq
    \Pi^{n\cdot\bone_q}_{\bi_\bell}\ast C
    \lemeq{lem_proj_contr}
    \left(\Pi^{n\cdot\bone_k}_{\bell}\cont{k}\Pi^{n\cdot\bone_q}_\bi\right)\ast C
    \lemeq{lem_associativity_contraction}
    \Pi^{n\cdot\bone_k}_{\bell}\ast\left(\Pi^{n\cdot\bone_q}_\bi\ast C\right)
    \spaceeq
    \Pi^{n\cdot\bone_k}_{\bell}\ast S.
\end{align*}
If $|\bell|=k$ (equivalently, $|\bj|=k$), the tensor $\Pi^{n\cdot\bone_k}_{\bell}\ast S$ is a reflection of $S$; i.e., it is obtained from $S$ by simply permuting its modes (cf.~Section~\ref{subsec_projection_tensor}).
As a consequence, the definition above may be rephrased by asking that the projections of a $k$-crystal onto hyperplanes generated by $k$ distinct modes should be equal \emph{up to the reflection} associated with the orderings of the modes.  
\end{rem}
Let now $C$ be a $k$-crystal, and let $S$ be its $k$-shadow.
The condition~\eqref{eqn_1536_1nov} for the map $\zeta_C$ associated with $C$ becomes now a condition on the shadow $S$:
The only entries of $S$ that are allowed to be nonzero are the ones whose coordinates are all distinct. We say that a tensor satisfying this requirement is \emph{hollow}.

\begin{figure}
\begin{center}
\SCrystal{-3}{.8}
\end{center}
\caption{The tensor $S$ from Example~\ref{ex_no_hollo_crystal_of_small_width}.}
\label{fig_shadow_S}
\end{figure}

\begin{defn}
\label{defn_hollow_tensors}
Let $k\in\N$, let $\bn\in\N^k$, and let $T\in \mathcal{T}^\bn(\Z)$. 
A tuple $\ba\in [\bn]$ is a
\emph{tie} for $T$ if $|\ba|<k$ and $\ba\in\supp(T)$. We say that $T$ is \emph{hollow} if $T$ does not have any ties.
\end{defn}

In summary, we have (informally) shown that an affine $q$-dimensional $k$-crystal $C$ of width $n$ whose $k$-shadow is hollow yields a $k$-tensorial homomorphism $\zeta_C$ satisfying~\eqref{eqn_1536_1nov} and thus, through Theorem~\ref{thm_decoupling_BLPAIP_for_cliques}, certifies that $\BA^k(\X,\K_n)=\YES$ if $\X$ has $q$ vertices. (How to explicitly construct $\zeta_C$ from a hollow-shadowed crystal $C$ is discussed in more detail in the proof of Theorem~\ref{acceptance_prop_BA_big_enough} in Section~\ref{sec_fooling_BA}.) The problem is now to verify if such crystals actually exist. The next example shows that it is not possible to build a hollow-shadowed crystal whose width is too small.

\begin{example}
\label{ex_no_hollo_crystal_of_small_width}
We now show by contradiction that, for any $q\geq 4$, it is not possible to build an affine $q$-dimensional $3$-crystal $C$ of width $3$ whose $3$-shadow $S$ is hollow.

Suppose that such $C$ exists.
First, observe that $S$ belongs to $\cT^{3\cdot\bone_3}(\Z)$; i.e., it is a $3\times 3\times 3$ integer tensor. Figure~\ref{fig_shadow_S} shows $S$ together with its three $2$-dimensional oriented projections; in grey are the cells that need to be zero to satisfy the hollowness requirement, while each of the other six cells is assigned a different colour.\footnote{The colours in Figure~\ref{fig_shadow_S} are not related to the colours used in Section~\ref{subsec_crystals} and in Example~\ref{example_four_dimensional_crystal}.}
We shall see in Proposition~\ref{prop_crystals_smaller_dimension} that, if $C$ is a $3$-crystal, it also needs to be a $2$ crystal; let $\tilde S$ be the $2$-shadow of $C$. Then, for any $\bi\in [3]^2_\to$, we have
\begin{align*}
    \Pi^{3\cdot\bone_3}_\bi\ast S
    \spaceeq
    \Pi^{3\cdot\bone_3}_\bi\ast \left(\Pi^{3\cdot\bone_q}_{\ang{3}}\ast C\right)
    \lemeq{lem_associativity_contraction}
    \Pi^{3\cdot\bone_3}_\bi\cont{3}\Pi^{3\cdot\bone_q}_{\ang{3}}\ast C
    \lemeq{lem_proj_contr}
    \Pi^{3\cdot\bone_q}_{\ang{3}_\bi}\ast C
    \spaceeq
    \Pi^{3\cdot\bone_q}_{\bi}\ast C
    \spaceeq
    \tilde S.
\end{align*}
In other words, $S$ is a $2$-crystal itself. It follows that the three oriented $2$-dimensional projections of $S$ depicted in Figure~\ref{fig_shadow_S} need to coincide:
\begin{align*}
\SProjectionXY{-1.5}{.5}
\spaceeq
\SProjectionXZ{-1.5}{.5}
\spaceeq
\SProjectionYZ{-1.5}{.5}\;.
\end{align*}
This forces all six non-grey entries of $S$ to be equal. On the other hand, $C$ is affine, and we will see in Lemma~\ref{lem_conservation_sum_crystals} that $S$ is affine, too. Since the entries of $S$ are integers, this yields a contradiction.
\end{example}
As a consequence, taking an arbitrary digraph with high chromatic number is not enough for fooling the $\BA$ hierarchy applied to AGC; in particular, unlike for the $\AIP$ hierarchy,
one cannot simply use cliques as fooling instances.
This motivates the strategy, discussed in Section~\ref{subsec_fooling_BA} (see also Section~\ref{sec_fooling_BA}), of using \emph{shift digraphs} instead of cliques as fooling instances. To guarantee $\BA^k$ acceptance for this more refined class of digraphs, it shall be enough to have hollow-shadowed crystals whose width is \emph{sub-exponential} in $k$.
The result stated next is the
main technical contribution of this work, and it shows the existence of hollow-shadowed crystals whose width is \emph{quadratic} in $k$.

\begin{thm*}[Theorem~\ref{thm_existence_crystals_with_hollow_shadow} restated]
For any $k\leq q\in\N$ there exists an affine $k$-crystal $C\in\mathcal{T}^{\frac{k^2+k}{2}\cdot\bone_q}(\Z)$ with hollow $k$-shadow.
\end{thm*}

\noindent The core of this section is dedicated to the proof of the next result, from which Theorem~\ref{thm_existence_crystals_with_hollow_shadow} will follow via a simple \textit{crystalisation} argument described in Section~\ref{subsec_crystalisation}.

\begin{thm*}[Theorem~\ref{prop_hollow_shadows_exist} restated]
For any $k\in\N$ there exists a hollow affine $(k-1)$-crystal $C\in\mathcal{T}^{\frac{k^2+k}{2}\cdot\bone_k}(\Z)$.
\end{thm*}

\noindent Our strategy to prove Theorem~\ref{prop_hollow_shadows_exist} shall be the following: 
\begin{enumerate}
    \item[($\spadesuit$\,1)]
    We start with a hollow affine $(k-1)$-dimensional $(k-2)$-crystal $U$
    of width $\frac{k^2-k}{2}$,    
    whose existence we assume by induction.
    \item[($\spadesuit$\,2)]
    We build a (not necessarily hollow) $k$-dimensional $(k-1)$-crystal $V$ whose shadow is $U$. This is done by using a general construction---described in Section~\ref{subsec_system_of_shadows_BODY}---that, given a ``realistic system of shadows'' $\mathcal{S}$, produces a ``realisation'' of $\mathcal{S}$, i.e., a
    tensor whose projections are precisely the members of $\mathcal{S}$. In particular, the construction yields the crystalisation procedure of Section~\ref{subsec_crystalisation}.
    \item[($\spadesuit$\,3)]
    We pad $V$ with $k$ layers of zeros in each dimension,
    thus obtaining a wider tensor $W$ that is still a $k$-dimensional $(k-1)$-crystal.
    \item[($\spadesuit$\,4)]
    We modify $W$ by adding to it certain transparent crystals, which we call \emph{quartzes}, discussed in Section~\ref{subsec_quartzes_BODY}. These crystals have the property of projecting an all-zero shadow, which implies in particular that the tensor $C$ obtained after this process is still a crystal. 
    \item[($\spadesuit$\,5)]
    By carefully choosing the quartzes, we end up with $C$ being hollow (as shown in Section~\ref{subsec_hollow_shadowed_crystals_BODY}).
\end{enumerate} 
\begin{rem}
The step ($\spadesuit$\,3) has the consequence that the hollow crystals resulting from this process are progressively wider as $k$ increases.
In fact, we are not able to build an affine hollow $(k-1)$-crystal $C\in\mathcal{T}^{n\cdot\bone_k}(\Z)$ for \emph{all} choices of $k$ and $n$.
This is not a deficit of our methods: For instance, it follows from Example~\ref{ex_no_hollo_crystal_of_small_width} that an affine hollow $2$-crystal in $\cT^{3\cdot\bone_3}(\Z)$ cannot exist.
\end{rem}

\begin{rem}
All of the steps ($\spadesuit$\,1)--($\spadesuit$\,5) in the proof of Theorem~\ref{prop_hollow_shadows_exist} are \emph{constructive}, in that they directly translate into an algorithm to find the required crystal. As a consequence, the proof of Theorem~\ref{thm_existence_crystals_with_hollow_shadow} on the existence of hollow-shadowed crystals of quadratic width is constructive, too.
\end{rem}

\subsection{Monotonicity of crystals}
As a warm-up, we start by proving the next monotonicity property of crystals.
\begin{prop}
\label{prop_crystals_smaller_dimension}
Let ${q},n\in\N$, let ${h},{k}\in\N_0$, and suppose that ${h}\leq{k}<{q}$. Then any ${k}$-crystal in $\cT^{n\cdot\bone_{q}}(\Z)$ is also an ${h}$-crystal.
\end{prop}
\noindent Before proving the proposition, we illustrate it with an example.
\begin{example}
Suppose for concreteness that $h=2$, $k=3$, and $q=6$, and let $C$ be a $6$-dimensional $3$-crystal. To simplify the notation in this example, let $C_{ij\dots}$ denote the projection of $C$ onto the modes $(i,j,\dots)$.
To see why Proposition~\ref{prop_crystals_smaller_dimension} is true, observe first that some of the oriented $2$-dimensional projections of $C$ must be equal as an immediate consequence of the definition of a $3$-crystal.
For example, the fact that, say, $C_{12}=C_{23}$ immediately follows from the fact that $C_{123}=C_{234}$---which, in turn, is implied by $C$ being a $3$-crystal. However, in order to show that, say, $C_{12}=C_{56}$, one step is not sufficient: \emph{Two} of the equalities enforced by $C$ being a $3$-crystal need to be considered. For example, we may derive from $C_{123}=C_{456}$ that $C_{12}=C_{45}$, and from $C_{345}=C_{456}$ that $C_{45}=C_{56}$. 
\end{example}
Thus, in some sense, Proposition~\ref{prop_crystals_smaller_dimension} relies on the connectedness of the graph encoding the projections of the given $k$-crystal onto lower dimensional spaces.
The proof below formalises this idea in arbitrary dimensions via a simple minimality argument. 
We point out that the assumption $k<q$ is crucial for this argument to work---and for the result to hold. Indeed, it is easily verified from Definition~\ref{defn_crystal_informal} that \emph{any} tensor $C\in\mathcal{T}^{n\cdot\bone_q}(\Z)$ is a $k$-crystal for $k=q$.

\begin{proof}[Proof of Proposition~\ref{prop_crystals_smaller_dimension}]
We can assume that ${h}={k}-1$ without loss of generality.
Given a tuple $\bi\in [{q}]^{{h}}_\to$ and $p\in [{q}]\setminus\set(\bi)$, we define $\bi\boxplus p$ as the tuple in $[{q}]^{k}_\to$ obtained by inserting $p$ into $\bi$ in the unique position that makes the resulting tuple monotonically increasing; in other words, $\bi\boxplus p=(\bi_\ang{\alpha},p,\bi_{(\alpha+1,\dots,{h})})$, where $\alpha=|\{\beta\in [{h}]:i_\beta<p\}|$. Similarly, given $\bj\in [{q}]^{k}_\to$ and $r\in\set(\bj)$, we define $\bj\boxminus r$ as the tuple in $[{q}]^{{h}}_\to$ obtained by removing $r$ from $\bj$.

Let ${C}$ be a ${k}$-crystal in $\cT^{n\cdot\bone_{q}}(\Z)$, and consider the tensor $S=\Pi^{n\cdot\bone_{q}}_\ang{{h}}\ast {C}$. 
We now show that $\Pi^{n\cdot\bone_{q}}_\bi\ast {C}=S$ for each $\bi\in [{q}]^{{h}}_\to$, which implies the result. For the sake of contradiction,
let $\bi\in [{q}]^{{h}}_\to$ be a tuple such that $\Pi^{n\cdot\bone_{q}}_\bi\ast {C}\neq S$ and such that the quantity $\bi^T\bone_{{h}}$ is minimum among the set of tuples $\bi'\in [{q}]^{{h}}_\to$ for which $\Pi^{n\cdot\bone_{q}}_{\bi'}\ast {C}\neq S$. Notice that the set $[{q}]\setminus\set(\bi)$ has at least two elements as ${h}={k}-1\leq {q}-2$. Therefore, the numbers $\mu=\min([{q}]\setminus\set(\bi))$ and $\nu=\min([{q}]\setminus(\set(\bi)\cup \{\mu\}))$ are well defined.
Consider the tuples $\ba=\bi\boxplus\nu$ and $\bb=\ba\boxminus a_{{k}}\boxplus\mu$ (where the operations are meant to be executed from the left to the right). By construction, we have $2\leq \nu\leq {k}+1$,
so $\nu-1\in[k]=\set(\ang{k})$. Hence,
we can define the tuple $\bc=\ang{{k}}\boxminus (\nu-1)$. By the definition of $\mu$ and $\nu$, we have that $a_{\nu-1}=\nu$. This implies that $\ba_\bc=\bi$, so
\begin{align}
\label{eqn_1736_7oct}
\notag
    S&\spaceneq
    \Pi^{n\cdot\bone_{q}}_\bi\ast {C}
    \spaceeq
    \Pi^{n\cdot\bone_{q}}_{\ba_\bc}\ast {C}
    \lemeq{lem_proj_contr}
    \left(\Pi^{n\cdot\bone_{k}}_\bc\cont{k}\Pi^{n\cdot\bone_{q}}_\ba\right)\ast {C}
    \lemeq{lem_associativity_contraction}
    \Pi^{n\cdot\bone_{k}}_\bc\ast\left(\Pi^{n\cdot\bone_{q}}_\ba\ast {C}\right)\\
    &\spaceeq
    \Pi^{n\cdot\bone_{k}}_\bc\ast\left(\Pi^{n\cdot\bone_{q}}_\bb\ast {C}\right)
    \lemeq{lem_associativity_contraction}
    \left(\Pi^{n\cdot\bone_{k}}_\bc\cont{k}\Pi^{n\cdot\bone_{q}}_\bb\right)\ast {C}
    \lemeq{lem_proj_contr}
    \Pi^{n\cdot\bone_{q}}_{\bb_\bc}\ast {C},
\end{align}
where the fourth equality uses that ${C}$ is a ${k}$-crystal and that $\ba,\bb\in [{q}]^{k}_\to$.
Observe that $a_{k}\geq \nu>\mu$, so $\bb\leq\ba$ entrywise. It follows that $\bb_\bc\leq\ba_\bc=\bi$ entrywise. 

Assume first that $\nu\leq {k}$. In this case, we have $c_{{h}}={k}$. We claim that $i_{h}>\nu$; otherwise, we would have $i_{h}\leq \nu\leq {k}$, which would yield $\set(\bi)\cup\{\mu\}=[{k}]$ since $\bi$ is monotonically increasing. This would force $\nu={k}+1$, a contradiction. In turn, $i_{h}>\nu$ implies that $a_{k}=i_{h}$. In particular, this means that $a_{k}>\mu$, so $b_{k}<a_{k}$. We conclude that $b_{c_{{h}}}=b_{k}<a_{k}=i_{h}$.
Since, as noted above, $\bb_\bc\leq\bi$ entrywise, it follows that
$\bb_\bc^T\bone_{{h}}<\bi^T\bone_{{h}}$. Putting all together, we have derived that $\Pi^{n\cdot\bone_{q}}_{\bb_\bc}\ast {C}\neq S$ and $\bb_\bc^T\bone_{{h}}<\bi^T\bone_{{h}}$, thus contradicting our minimality assumption.

On the other hand, if $\nu={k}+1$, we deduce that $\bi=\ang{{k}}\boxminus\mu$, so $\ba=\ang{{k}}\boxminus\mu\boxplus({k}+1)$, thus yielding $a_{k}={k}+1$. Therefore, 
\begin{align*}
\bb
&\spaceeq
\ang{{k}}\boxminus\mu\boxplus({k}+1)\boxminus a_{k}\boxplus\mu
\spaceeq
\ang{{k}}\boxminus\mu\boxplus({k}+1)\boxminus ({k}+1)\boxplus\mu
\spaceeq
\ang{{k}}\boxminus\mu\boxplus\mu
\spaceeq
\ang{{k}},
\end{align*}
while $\bc=\ang{{k}}\boxminus(\nu-1)=
\ang{{k}}\boxminus{k}=\ang{{h}}$ and, thus, $\bb_\bc=\ang{{k}}_{\ang{{h}}}=\ang{{h}}$. Then,~\eqref{eqn_1736_7oct} yields $\Pi^{n\cdot\bone_{q}}_\ang{{h}}\ast {C}\neq S$, which again contradicts our assumptions.

Hence, we have shown that $\Pi^{n\cdot\bone_{q}}_\bi\ast {C}=S$ for each $\bi\in [{q}]^{{h}}_\to$, and the proof is concluded.
\end{proof}

\subsection{Systems of shadows}
\label{subsec_system_of_shadows_BODY}
A crystal tensor has the property of projecting the same shadow onto each oriented hyperplane of appropriate dimension, cf.~Definition~\ref{defn_crystal_informal}. The step ($\spadesuit$\,2) of the strategy to prove Theorem~\ref{prop_hollow_shadows_exist} requires reconstructing a crystal from its shadow. We now show how to accomplish this task. In fact, our approach shall be more general: In Theorem~\ref{thm_realistic_equals_realisable}, we characterise those sets of (lower-dimensional) tensors that can be realised as the oriented projections of a single (higher-dimensional) tensor. Then, we shall see in Section~\ref{subsec_crystalisation} (Corollary~\ref{cor_exogenesis_of_crystals}) that this characterisation easily implies the existence of the crystal required in ($\spadesuit$\,2).

\begin{defn}
\label{defn_system_of_shadows}
For $p,q\in\N$ and $\bn\in\N^q$, a \emph{$(p,\bn)$-system of shadows} is a set $\mathcal{S}=\{S_\bi\}_{\bi\in [q]^p_\rightarrow}$ such that $S_\bi\in\cT^{\bn_\bi}(\Z)$ for each $\bi\in [q]^p_\rightarrow$.
\begin{itemize}
    \item 
    $\mathcal{S}$ is a \emph{realistic} system of shadows if 
    \begin{align}
    \label{eqn_condition_balanced_sums}
    \Pi^{\bn_\bi}_\br\ast S_\bi=\Pi_\bs^{\bn_\bj}\ast S_\bj && \mbox{ for any }&&\bi,\bj\in [q]^p_\rightarrow,\;\br,\bs\in [p]^{p-1}_\rightarrow\mbox{ such that }\bi_\br=\bj_\bs.
    \end{align}
    \item
    $\mathcal{S}$ is a \emph{realisable} system of shadows if there exists a tensor $C\in\cT^{\bn}(\Z)$ such that $\Pi^\bn_\bi\ast C=S_\bi$ for each $\bi\in [q]^p_\rightarrow$.
\end{itemize}
\end{defn}
In other words, a system of $p$-dimensional ``shadow'' tensors is realistic if the shadows are locally compatible with each other in the sense of the requirement~\eqref{eqn_condition_balanced_sums}, while it is realisable if it can actually be realised as the set of $p$-dimensional oriented projections of a single $q$-dimensional tensor. Notice that, for the set $\mathcal{S}$ to be nonempty, we must have $p\leq q$. Observe also that the tensors $S_\bi$ and $C$ are not required to be cubical.

Using Lemma~\ref{lem_proj_contr}, it is not hard to check that a realisable system of shadows is always realistic. As stated in the next theorem, it turns out that the two conditions are in fact equivalent.
\begin{thm}
\label{thm_realistic_equals_realisable}
Let $p,q\in\N$ and $\bn\in\N^q$. A $(p,\bn)$-system of shadows is realistic if and only if it is realisable.
\end{thm}

Theorem~\ref{thm_realistic_equals_realisable} is proved through a nested
induction---first on the dimension of the shadows $S_\bi$ (i.e., $p$), and second on the
sum of the sizes of the modes of the tensor $C$ that realises the shadows (i.e.,
$\bn^T\bone_q$). Lemmas~\ref{lem_1_dim_albums_are_truthful} and~\ref{lem_realistic_realisable_case_p_1} contain the base cases for the second and the first inductions, respectively. We note that the proof of Theorem~\ref{thm_realistic_equals_realisable}---as well as the proofs of Lemmas~\ref{lem_1_dim_albums_are_truthful} and~\ref{lem_realistic_realisable_case_p_1}---is constructive, as it directly provides a procedure to recover the tensor $C$ realising a given realistic system of shadows $\mathcal{S}$. See also Example~\ref{example_four_dimensional_crystal}, which illustrates this procedure applied to the problem of building a $4$-dimensional $2$-crystal having a given shadow.

In order to establish that a realistic system of shadows is always realisable---the non-trivial direction in Theorem~\ref{thm_realistic_equals_realisable}---we start by showing that the problem is invariant under 
reflections
of the tensors involved.
\begin{lem}
\label{lem_rotating_preserves_truth}
Let $p,q\in\N$, let $\bell\in [q]^q$ be such that $|\bell|=q$, and let
  $\bn\in\N^q$. If every realistic $(p,\bn_\bell)$-system of shadows is realisable then every realistic $(p,\bn)$-system of shadows is realisable.
\end{lem}
\begin{proof}
For purely typographical reasons, in this proof we will adopt an in-line notation for the operation of tuple projection: Given tuples $\ba,\bb,\bc,\dots$ of suitable lengths, we will denote the iterated projection $\ba_{\bb_{\bc_{\dots}}}$ by $\ba\bb\bc\cdots$.

Since every permutation can be expressed as the composition of transpositions, it is
  enough to consider the case that $\bell$ is a transposition; in particular, $\bell\bell=\ang{q}$.
Let $\cS=\{S_\bi\}_{\bi\in [q]^p_\rightarrow}$ be a realistic $(p,\bn)$-system of shadows. For any $\bi\in [q]^p_\rightarrow$, let $\bi^+$ be the (unique) tuple in $[p]^p$ such that $\bell{\bi{\bi^+}}\in [q]^p_\rightarrow$. Let also $\bi^-$ be the (unique) tuple in $[p]^p$ such that $\bi^+{\bi^-}=\bi^-{\bi^+}=\ang{p}$.
For each $\bi\in [q]^p_\rightarrow$, define the tensor
\begin{align}
\label{eqn_1045_0607}
\tilde{S}_\bi= \Pi^{\bn{\bell{\bi{\bi^+}}}}_{\bi^-}\ast S_{\bell{\bi{\bi^+}}}.
\end{align}
Observe that $\tilde{S}_\bi\in\cT^{\bn{\bell\bi}}(\Z)$, so $\tilde{\cS}=\{\tilde{S}_\bi\}_{\bi\in [q]^p_\rightarrow}$ is a $(p,\bn\bell)$-system of shadows. 
We claim that $\tilde{\cS}$ is a realistic system.
To prove the claim, take $\bi,\bj\in [q]^p_\rightarrow$ and $\br,\bs\in [p]^{p-1}_\rightarrow$ such that $\bi\br=\bj\bs$. We need to show that 
\begin{align}
\label{eqn_2041_08062022}
\Pi^{\bn{\bell\bi}}_{\br}\ast\tilde{S}_\bi
=
\Pi^{\bn{\bell\bj}}_{\bs}\ast\tilde{S}_\bj.
\end{align}
Let $\balpha,\bbeta\in [p-1]^{p-1}$ be the (unique) tuples such that $\bi^-{\br\balpha}\in [p]^{p-1}_\rightarrow$ and $\balpha\bbeta=\bbeta\balpha=\ang{p-1}$. We claim that $\bj^-{\bs\balpha}\in [p]^{p-1}_\rightarrow$. Indeed, for any $x,y\in [p-1]$ such that $x<y$ we have 
\begin{align*}
i^-{r{\alpha x}}<i^-{r{\alpha y}}
\quad
&\Rightarrow
\quad
\ell{i{i^+{i^-{r{\alpha x}}}}}
<
\ell{i{i^+{i^-{r{\alpha y}}}}}
\quad
\Rightarrow
\quad
\ell{i{{{r{\alpha x}}}}}
<
\ell{i{{{r{\alpha y}}}}}
\quad
\Rightarrow
\quad
\ell{j{{{s{\alpha x}}}}}
<
\ell{j{{{s{\alpha y}}}}}\\
\quad
&\Rightarrow
\quad
\ell{j{j^+{j^-{s{\alpha x}}}}}
<
\ell{j{j^+{j^-{s{\alpha y}}}}}
\quad
\Rightarrow
\quad
j^-{s{\alpha x}}<j^-{s{\alpha y}},
\end{align*}
thus proving the claim.
Therefore,
\begin{align}
\label{eqn_2117_0806_2022_A}
\notag
\Pi^{\bn{\bell\bi}}_{\br}\ast\tilde{S}_\bi
&\equationeq{eqn_1045_0607}
\Pi^{\bn{\bell\bi}}_{\br}\ast\left(\Pi^{\bn{\bell{\bi{\bi^+}}}}_{\bi^-}\ast S_{\bell{\bi{\bi^+}}}\right)
\lemeq{lem_associativity_contraction}
\Pi^{\bn{\bell\bi}}_{\br}\cont{p}\Pi^{\bn{\bell{\bi{\bi^+}}}}_{\bi^-}\ast S_{\bell{\bi{\bi^+}}}
\lemeq{lem_proj_contr}
\Pi^{\bn{\bell{\bi{\bi^+}}}}_{\bi^-\br}\ast S_{\bell{\bi{\bi^+}}}\\
&\spaceeq
\Pi^{\bn{\bell{\bi{\bi^+}}}}_{\bi^-{\br{\balpha\bbeta}}}\ast S_{\bell{\bi{\bi^+}}}
\lemeq{lem_proj_contr}
\Pi^{\bn{\bell{\bi{\br\balpha}}}}_\bbeta\cont{p-1}\Pi^{\bn{\bell{\bi{\bi^+}}}}_{\bi^-{\br{\balpha}}}\ast S_{\bell{\bi{\bi^+}}}
\lemeq{lem_associativity_contraction}
\Pi^{\bn{\bell{\bi{\br\balpha}}}}_\bbeta\ast\left(\Pi^{\bn{\bell{\bi{\bi^+}}}}_{\bi^-{\br{\balpha}}}\ast S_{\bell{\bi{\bi^+}}}\right)
\intertext{and, similarly,}
\label{eqn_2117_0806_2022_B}
\Pi^{\bn{\bell\bj}}_{\bs}\ast\tilde{S}_\bj
&\spaceeq
%=
\Pi^{\bn{\bell{\bj{\bs\balpha}}}}_\bbeta\ast\left(\Pi^{\bn{\bell{\bj{\bj^+}}}}_{\bj^-{\bs{\balpha}}}\ast S_{\bell{\bj{\bj^+}}}\right).
\end{align}
Let us now focus on the tuples $\bell{\bi{\bi^+}},\bell{\bj{\bj^+}}\in [q]^p_\rightarrow$ and $\bi^-{\br\balpha},\bj^-{\bs\balpha}\in [p]^{p-1}_\rightarrow$. Observe that 
\begin{align*}
\bell{\bi{\bi^+{\bi^-{\br\balpha}}}}
&\spaceeq
\bell{\bi{{{\br\balpha}}}}
\spaceeq
\bell{\bj{{{\bs\balpha}}}}
\spaceeq
\bell{\bj{\bj^+{\bj^-{\bs\balpha}}}}.
\end{align*}
Using that $\cS$ is a realistic system, we deduce
\begin{align}
\label{eqn_2117_0806_2022_C}
\Pi^{\bn{\bell{\bi{\bi^+}}}}_{\bi^-{\br\balpha}}\ast S_{\bell{\bi{\bi^+}}}
&\spaceeq
\Pi^{\bn{\bell{\bj{\bj+}}}}_{\bj^-{\bs\balpha}}\ast S_{\bell{\bj{\bj+}}}.
\end{align}
Combining~\eqref{eqn_2117_0806_2022_A},~\eqref{eqn_2117_0806_2022_B}, and~\eqref{eqn_2117_0806_2022_C}, and recalling that $\bi\br=\bj\bs$, yields~\eqref{eqn_2041_08062022}, thus proving that $\tilde{\cS}$ is a realistic $(p,\bn\bell)$-system of shadows, as claimed. From the hypothesis of the lemma, we deduce that $\tilde{\cS}$ is realisable, so there exists a tensor $\tilde{C}\in\cT^{\bn\bell}(\Z)$ such that $\Pi^{\bn\bell}_\bi\ast\tilde{C}=\tilde{S}_\bi$ for each $\bi\in [q]^p_\rightarrow$. Define $C= \Pi^{\bn\bell}_{\bell}\ast\tilde{C}\in\cT^{\bn}(\Z)$ (where we are using that $\bell\bell=\ang{q}$). Given $\bi\in [q]^p_\rightarrow$, we find
\begin{align}
\label{eqn_2158_08062022_A}
\notag
\displaystyle \Pi^\bn_\bi\ast C
&\spaceeq
\displaystyle \Pi^\bn_\bi\ast (\Pi^{\bn\bell}_{\bell}\ast\tilde{C})
\spaceeq
\displaystyle \Pi^\bn_{\bi{\bi^+{\bi^-}}}\ast (\Pi^{\bn\bell}_{\bell}\ast\tilde{C})
\lemeq{lem_associativity_contraction}
\displaystyle \Pi^\bn_{\bi{\bi^+{\bi^-}}}\cont{q} \Pi^{\bn\bell}_{\bell}\ast\tilde{C}\\
\notag
&\lemeq{lem_proj_contr}
\displaystyle \Pi^{\bn{\bi{\bi^+}}}_{\bi^-}\cont{p}\Pi^{\bn{\bi}}_{\bi^+}\cont{p}\Pi^\bn_\bi\cont{q} \Pi^{\bn\bell}_{\bell}\ast\tilde{C}
\lemeq{lem_proj_contr}
\displaystyle \Pi^{\bn{\bi{\bi^+}}}_{\bi^-}\cont{p}
\Pi^{\bn\bell}_{\bell{\bi{\bi^+}}}\ast\tilde{C}
\lemeq{lem_associativity_contraction}
\displaystyle \Pi^{\bn{\bi{\bi^+}}}_{\bi^-}\ast(
\Pi^{\bn\bell}_{\bell{\bi{\bi^+}}}\ast\tilde{C})\\
&\spaceeq
\displaystyle \Pi^{\bn{\bi{\bi^+}}}_{\bi^-}\ast\tilde{S}_{\bell{\bi{\bi^+}}}.
\end{align}
Notice that $\bell{\bell{\bi{\bi^+{\bi^-}}}}=\bi$, which is an increasing tuple. Hence, $\left(\bell{\bi{\bi^+}}\right)^+=\bi^-$ and, consequently, $\left(\bell{\bi{\bi^+}}\right)^-=\bi^+$. It follows from~\eqref{eqn_1045_0607} that
\begin{align}
\label{eqn_2158_08062022_B}
\tilde{S}_{\bell{\bi{\bi^+}}}
&\spaceeq
\Pi^{\bn{\bell{\bell{\bi{\bi^+{\bi^-}}}}}}_{\bi^+}\ast S_{\bell{\bell{\bi{\bi^+{\bi^-}}}}}
\spaceeq
\Pi^{\bn{\bi}}_{\bi^+}\ast S_{\bi}.
\end{align}
Combining~\eqref{eqn_2158_08062022_A} and~\eqref{eqn_2158_08062022_B} yields
\begin{align*}
\Pi^\bn_\bi\ast C
&\spaceeq
\Pi^{\bn{\bi{\bi^+}}}_{\bi^-}\ast(\Pi^{\bn{\bi}}_{\bi^+}\ast S_{\bi})
\lemeq{lem_associativity_contraction}
\Pi^{\bn{\bi{\bi^+}}}_{\bi^-}\cont{p}\Pi^{\bn{\bi}}_{\bi^+}\ast S_{\bi}
\lemeq{lem_proj_contr}
\Pi^{\bn\bi}_{\ang{p}}\ast S_\bi
\lemeq{lem_identity_Pi}
S_\bi,
\end{align*}
which concludes the proof that $\cS$ is a realisable system of shadows.
\end{proof}

The next result establishes the base case for the second
induction in the proof of Theorem~\ref{thm_realistic_equals_realisable}. Its proof is a simple connectivity argument for the shadows' modes, analogous to the one used to prove Proposition~\ref{prop_crystals_smaller_dimension}.

\begin{lem}
\label{lem_1_dim_albums_are_truthful}
A realistic $(p,\bone_q)$-system of shadows is realisable for any $p,q\in\N$.
\end{lem}
\begin{proof}
Let $\cS=\{S_\bi\}_{\bi\in [q]^p_\rightarrow}$ be a realistic $(p,\bone_q)$-system of shadows. For any $\bi\in [q]^p_\rightarrow$, $S_\bi\in\cT^{(\bone_q)_\bi}(\Z)=\cT^{\bone_p}(\Z)$. We claim that $S_\bi=S_\bj$ for any $\bi,\bj\in [q]^p_\rightarrow$.  Define, for each pair $\bi,\bj\in [q]^p_\rightarrow$, their \emph{distance} $\operatorname{d}(\bi,\bj)$ as the cardinality of the set $\{t\in [p]:i_t\neq j_t\}$. Suppose, for the sake of contradiction, that the claim is false, and let  $\bi,\bj\in [q]^p_\rightarrow$ attain the minimum distance among all pairs $\bi',\bj'$ for which $S_{\bi'}\neq S_{\bj'}$. Let $\alpha=\max\{t\in [p]:i_t\neq j_t\}$. Assume, without loss of generality, that $i_\alpha<j_\alpha$, and define a new tuple $\bell\in [q]^p$ obtained from $\bi$ by replacing $i_\alpha$ with $j_\alpha$. Observe that $i_1<i_2<\dots<i_{\alpha-1}<i_\alpha<j_\alpha<j_{\alpha+1}=i_{\alpha+1}<i_{\alpha+2}<\dots<i_p$, so $\bell\in [q]^p_\rightarrow
$. Letting $\br\in [p]^{p-1}_\rightarrow$ be obtained from $\ang{p}$ by deleting its $\alpha$-th entry, observe that $\bi_\br=\bell_\br$. Using that $\cS$ is a realistic system, we obtain $\Pi^{\bone_p}_\br\ast S_\bi=\Pi_\br^{\bone_p}\ast S_\bell$. Therefore,
\begin{align*}
E_{\bone_p}\ast S_\bi
&\lemeq{lem_Pi_basic}
E_{\bone_{p-1}}\ast \Pi^{\bone_p}_\br\ast S_\bi
\spaceeq
E_{\bone_{p-1}}\ast \Pi^{\bone_p}_\br\ast S_\bell
\lemeq{lem_Pi_basic}
E_{\bone_p}\ast S_\bell,
\end{align*} 
so $S_\bell=S_\bi\neq S_\bj$. But this contradicts the choice of the pair $(\bi,\bj)$, as $\operatorname{d}(\bell,\bj)=\operatorname{d}(\bi,\bj)-1$. Hence, the claim is true.
We can then define a tensor $C\in\cT^{\bone_q}(\Z)$ by setting $E_{\bone_q}\ast C=E_{\bone_p}\ast S_\bi$ for any $\bi\in [q]^p_\rightarrow$. In this way, we get
\begin{align*}
E_{\bone_p}\ast\Pi^{\bone_q}_\bi\ast C
\lemeq{lem_Pi_basic}
E_{\bone_q}\ast C
\spaceeq
E_{\bone_p}\ast S_\bi.
\end{align*}
We conclude that $\Pi^{\bone_q}_\bi\ast C=S_\bi$ for any $\bi\in [q]^p_\rightarrow$, which means that $\cS$ is a realisable system.
\end{proof}

The next result establishes the base case for the first
induction in the proof of Theorem~\ref{thm_realistic_equals_realisable}.

\begin{lem}
\label{lem_realistic_realisable_case_p_1}
A realistic $(1,\bn)$-system of shadows is realisable for any $q\in\N$ and $\bn\in\N^q$.
\end{lem}
\begin{example}
 For $q=2$, the statement above expresses the fact that, given two integer vectors $S_1$ of length $n_1$ and $S_2$ of length $n_2$ such that the sums of the entries of $S_1$ and $S_2$ coincide, there exists an $n_1\times n_2$ integer matrix $C$ whose row-sum and column-sum vectors are $S_1$ and $S_2$, respectively.    
\end{example}
\begin{proof}[Proof of Lemma~\ref{lem_realistic_realisable_case_p_1}]
If $q=1$, the result is trivially true; indeed, in this case, the vector $C=S_1$ witnesses that the given system of shadows $\mathcal{S}=\{S_1\}$ is realisable. Hence, assume $q\geq 2$.
Notice that $ [q]^1_\rightarrow=[q]$, so each element of $[q]^1_\rightarrow$ is a single number.
We prove the statement by induction on $\bn^T\bone_q$.
If $\bn^T\bone_q=q$, then $\bn=\bone_q$, and the result follows from Lemma~\ref{lem_1_dim_albums_are_truthful}. Suppose that $\bn^T\bone_q\geq q+1$. 
Using Lemma~\ref{lem_rotating_preserves_truth}, we can assume $n_q\geq 2$ without loss of generality. Let $\cS=\{S_i\}_{i\in [q]}$ be a realistic $(1,\bn)$-system of shadows; observe that $S_i$ is a vector in $\cT^{n_i}(\Z)$ for each $i\in [q]$.
Set $\ell=E_{n_q}\ast S_q$ (i.e., $\ell$ is the last entry of $S_q$), and consider a new family of tensors $\tilde\cS=\{\tilde S_i\}_{i\in [q]}$ defined by 
\begin{align*}
\tilde{S}_i
\spaceeq
\left\{
\begin{array}{lll}
S_i-\ell E_{n_i} & \mbox{if }i\in [q-1]\\
(E_1\ast S_q,\dots,E_{n_q-1}\ast S_q)  & \mbox{if }i=q.
\end{array}
\right.
\end{align*}
Let $\tilde{\bn}=\bn-E_q$ and notice that $\tilde{\bn}\in\N^q$ since $n_q\geq 2$. We have that $S_i\in\cT^{\tilde{n}_i}(\Z)$ for each $i\in [q]$, so $\tilde\cS$ is a $(1,\tilde{\bn})$-system of shadows. 

We now show that $\tilde\cS$ is realistic.
By definition, $[1]^{0}_\rightarrow=\{\bepsilon\}$, so we only need to show that
\begin{align}
\label{eqn_1739_06062022_A}
\Pi^{\tilde n_i}_\bepsilon\ast \tilde S_i
\spaceeq
\Pi^{\tilde n_j}_\bepsilon\ast \tilde S_j
&& \forall\, i,j\in [q].
\end{align} 
We claim that 
\begin{align}
\label{eqn_1739_06062022_B}
\Pi^{\tilde n_i}_\bepsilon\ast \tilde S_i
\spaceeq
\Pi^{n_i}_\bepsilon\ast S_i-\ell && &\forall\, i\in [q].
\end{align}
Then,~\eqref{eqn_1739_06062022_A} will follow from the fact that $\cS$ is a realistic system. If $i\in [q-1]$,
\begin{align*}
\Pi^{\tilde n_i}_\bepsilon\ast \tilde S_i
&\spaceeq
\Pi^{n_i}_\bepsilon\ast (S_i-\ell E_{n_i})
\lemeq{Pi_epsilon_all_one}
\Pi^{n_i}_\bepsilon\ast S_i-\ell,
\intertext{so~\eqref{eqn_1739_06062022_B} holds in this case. Moreover,}
\Pi^{\tilde n_q}_\bepsilon\ast \tilde S_q
&\spaceeq
\Pi^{n_q-1}_\bepsilon\ast (E_1\ast S_q,\dots,E_{n_q-1}\ast S_q)
\lemeq{Pi_epsilon_all_one}
\sum_{b\in [n_q-1]}E_b\ast S_q 
\spaceeq
\bone_{n_q}\ast S_q-\ell\\
&\lemeq{Pi_epsilon_all_one}
\Pi^{n_q}_\bepsilon\ast S_q-\ell,
\end{align*}
so~\eqref{eqn_1739_06062022_B} holds in this case as well. We conclude that $\tilde\cS$ is indeed a realistic system. 

Since $\tilde\bn^T\bone_q=\bn^T\bone_q-1$, we have from the inductive hypothesis that $\tilde{\cS}$ is realisable, so there exists a tensor $\tilde C\in\cT^{\tilde{\bn}}(\Z)$ such that $\Pi^{\tilde{\bn}}_i\ast \tilde{C}=\tilde{S}_i$ for each $i\in [q]$. Define a tensor $C\in\cT^{\bn}(\Z)$ by setting, for each $\bb\in [\bn]$,
\begin{align}
\label{eqn_1230_0507}
E_\bb\ast C
\spaceeq
\left\{
\begin{array}{llll}
\ell&\mbox{ if }\bb=\bn\\
0&\mbox{ if }\bb\neq \bn\;\mbox{ and }\;b_q=n_q\\
E_\bb\ast \tilde{C}&\mbox{ if }b_q\neq n_q.
\end{array}
\right.
\end{align}
(Notice that the last line of the right-hand side of the above expression is well defined as, if $b_q\neq n_q$, then $\bb\in [\tilde{\bn}]$.)
Take $i\in [q]$; we claim that $\Pi^\bn_i\ast C=S_i$. For $a\in [n_i]$, we find
\begin{align*}
E_a\ast \Pi^\bn_i\ast C
&\lemeq{lem_Pi_basic}
\sum_{\substack{\bb\in [\bn]\\b_i=a}}E_\bb\ast C.
\end{align*}
For $i\neq q$, this yields
\begin{align*}
E_a\ast \Pi^\bn_i\ast C
&\spaceeq
\sum_{\substack{\bb\in [\bn]\\b_i=a\\b_q=n_q}}E_\bb\ast C
+
\sum_{\substack{\bb\in [\bn]\\b_i=a\\b_q\neq n_q}}E_\bb\ast C
\equationeq{eqn_1230_0507}
\ell\cdot \delta_{a,n_i}
+
\sum_{\substack{\bb\in [\tilde{\bn}]\\b_i=a}}E_\bb\ast \tilde C
\intertext{(where $\delta_{a,n_i}$ is $1$ if $a=n_i$, $0$ otherwise)}
&\lemeq{lem_Pi_basic}
\ell\cdot \delta_{a,n_i}
+
E_a\ast \Pi^{\tilde{\bn}}_i\ast \tilde{C}
\spaceeq
\ell\cdot \delta_{a,n_i}
+
E_a\ast\tilde{S}_i
\spaceeq
\ell\cdot \delta_{a,n_i}
+
E_a\ast(S_i-\ell E_{n_i})\\
&\spaceeq
E_a\ast S_i.
\end{align*}
For $i=q$, if $a=n_q$ we get
\begin{align*}
E_a\ast \Pi^\bn_q\ast C
&\spaceeq
\sum_{\substack{\bb\in [\bn]\\b_q=n_q}}E_\bb\ast C
\equationeq{eqn_1230_0507}
\ell
\spaceeq
E_a\ast S_q,
\end{align*}
while if $a\neq n_q$ we get 
\begin{align*}
E_a\ast \Pi^\bn_q\ast C
&\spaceeq
\sum_{\substack{\bb\in [\bn]\\b_q=a}}E_\bb\ast C
\equationeq{eqn_1230_0507}
\sum_{\substack{\bb\in [\tilde{\bn}]\\b_q=a}}E_\bb\ast \tilde C
\lemeq{lem_Pi_basic}
E_a\ast\Pi^{\tilde{\bn}}_q\ast\tilde{C}
\spaceeq
E_a\ast\tilde{S}_q\\
&\spaceeq
E_a\ast (E_1\ast S_q,\dots,E_{n_q-1}\ast S_q)
\spaceeq
E_a\ast S_q.
\end{align*} 
It follows that $\Pi^\bn_i\ast C=S_i$, as claimed. Therefore, $\cS$ is a realisable system.
\end{proof}

\begin{proof}[Proof of Theorem~\ref{thm_realistic_equals_realisable}]
Let $\cS=\{S_\bi\}_{\bi\in [q]^p_\rightarrow}$ be a realisable system of shadows; i.e., there exists $C\in\cT^\bn(\Z)$ such that $\Pi^\bn_\bi\ast C= S_\bi$ for each $\bi\in [q]^p_\rightarrow$. For any $\bi,\bj\in [q]^p_\rightarrow$ and $\br,\bs\in [p]^{p-1}_\rightarrow$ such that $\bi_\br=\bj_\bs$, we find
\begin{align*}
\Pi^{\bn_\bi}_\br\ast S_\bi
&\spaceeq
\Pi^{\bn_\bi}_\br\ast(\Pi^\bn_\bi\ast C)
\lemeq{lem_associativity_contraction}
\Pi^{\bn_\bi}_\br\cont{p}\Pi^\bn_\bi\ast C
\lemeq{lem_proj_contr}
\Pi^{\bn}_{\bi_\br}\ast C
\spaceeq
\Pi^{\bn}_{\bj_\bs}\ast C
\lemeq{lem_proj_contr}
\Pi^{\bn_\bj}_\bs\cont{p}\Pi^\bn_\bj\ast C\\
&\spaceeq
\Pi^{\bn_\bj}_\bs\ast S_\bj,
\end{align*}
which shows that $\cS$ is a realistic system. Hence, the ``if'' part of the statement is true. Next, we focus on the ``only if'' part.

We prove the result by nested induction, first on $p$ and second on $\bn^T\bone_q$. For $p=1$, the result follows from Lemma~\ref{lem_realistic_realisable_case_p_1}. Suppose that $p\geq 2$. 
For $\bn^T\bone_q=q$ (which implies $\bn=\bone_q$), the result follows from Lemma~\ref{lem_1_dim_albums_are_truthful}. Suppose that $\bn^T\bone_q\geq q+1$. Using Lemma~\ref{lem_rotating_preserves_truth}, we can safely assume $n_q\geq 2$. 
If $q=1$, then $[q]^p_\to=\emptyset$ and the statement is trivially true, so we can assume $q\geq 2$.
Let $\cS=\{S_\bi\}_{\bi\in [q]^p_\to}$ be a realistic $(p,\bn)$-system of shadows; we need to show that $\cS$ is realisable.

Set
$\hat{\bn}=(n_1,\dots,n_{q-1})\in\N^{q-1}$. For any $\bi\in [q-1]^{p-1}_\rightarrow$, we define $\hat S_\bi\in \cT^{\hat \bn_\bi}(\Z)$ by $E_\ba\ast \hat{S}_\bi= E_{(\ba,n_q)}\ast S_{(\bi,q)}$ for each $\ba\in [\hat{\bn}_\bi]$. Observe that the last expression is well defined, as $\bi\in [q-1]^{p-1}_\rightarrow$ implies that $(\bi,q)\in [q]^{p}_\rightarrow$. 
We claim that the family $\hat\cS=\{\hat{S}_\bi\}_{\bi\in [q-1]^{p-1}_\rightarrow}$ is a realistic $(p-1,\hat\bn)$-system of shadows. Take $\bi,\bj\in [q-1]^{p-1}_\rightarrow$ and $\br,\bs\in [p-1]^{p-2}_\rightarrow$ such that $\bi_\br=\bj_\bs$. For any $\ba\in [\hat{\bn}_{\bi_\br}]$, we find
\begin{align}
\label{eqn_1656_0507_A}
\notag
E_\ba\ast\Pi^{\hat{\bn}_\bi}_\br\ast\hat{S}_\bi
&\lemeq{lem_Pi_basic}
\sum_{\substack{\bb\in [\hat{\bn}_\bi]\\ \bb_\br=\ba}}E_\bb\ast\hat{S}_\bi
\spaceeq
\sum_{\substack{\bb\in [\hat{\bn}_\bi]\\ \bb_\br=\ba}}
E_{(\bb,n_q)}\ast S_{(\bi,q)}
\spaceeq
\sum_{\substack{\bc\in [{\bn}_{(\bi,q)}]\\ \bc_{(\br,p)}=(\ba,n_q)}}
E_{\bc}\ast S_{(\bi,q)}\\
&\lemeq{lem_Pi_basic}
E_{(\ba,n_q)}\ast\Pi^{\bn_{(\bi,q)}}_{(\br,p)}\ast S_{(\bi,q)}
\intertext{and, similarly,}
\label{eqn_1656_0507_B}
E_\ba\ast\Pi^{\hat{\bn}_\bj}_\bs\ast\hat{S}_\bj
&\spaceeq
E_{(\ba,n_q)}\ast\Pi^{\bn_{(\bj,q)}}_{(\bs,p)}\ast S_{(\bj,q)}.
\end{align}
We now use the fact that $\cS$ is a realistic system. In particular, we apply the requirement~\eqref{eqn_condition_balanced_sums} to the tuples $(\bi,q),(\bj,q)\in [q]^p_\rightarrow$ and $(\br,p),(\bs,p)\in [p]^{p-1}_\rightarrow$ (note that $(\bi,q)_{(\br,p)}=(\bi_\br,q)=(\bj_\bs,q)=(\bj,q)_{(\bs,p)}$). Since $(\ba,n_q)\in [\bn_{(\bi,q)_{(\br,p)}}]$, we obtain
\begin{align*}
E_{(\ba,n_q)}\ast\Pi^{\bn_{(\bi,q)}}_{(\br,p)}\ast S_{(\bi,q)}
&\spaceeq
E_{(\ba,n_q)}\ast\Pi^{\bn_{(\bj,q)}}_{(\bs,p)}\ast S_{(\bj,q)}.
\end{align*}
Combining this with~\eqref{eqn_1656_0507_A} and~\eqref{eqn_1656_0507_B} yields
\begin{align*}
E_\ba\ast\Pi^{\hat{\bn}_\bi}_\br\ast\hat{S}_\bi
&\spaceeq
E_\ba\ast\Pi^{\hat{\bn}_\bj}_\bs\ast\hat{S}_\bj.
\end{align*}
We conclude that $\hat\cS$ is a realistic system, as claimed.
It follows from the inductive hypothesis that $\hat\cS$ is realisable, so we can find a tensor $\hat{C}\in\cT^{\hat\bn}(\Z)$ such that $\Pi^{\hat{\bn}}_\bi\ast\hat{C}=\hat{S}_\bi$ for each $\bi\in [q-1]^{p-1}_\rightarrow$. 
Let now $\tilde{\bn}=\bn-E_q\in\N^q$. For any $\bi\in [q]^p_\rightarrow$, define a tensor $\tilde{S}_\bi\in\cT^{\tilde{\bn}_\bi}(\Z)$ as follows:
 If $i_p\neq q$ (in which case $\bi\in [q-1]^p_\rightarrow$) we set $\tilde{S}_\bi=S_\bi-\Pi^{\hat{\bn}}_\bi\ast\hat C$; if $i_p=q$, for $\bb\in [\tilde{\bn}_\bi]$, we set $E_\bb\ast\tilde{S}_\bi=E_\bb\ast S_\bi$ (where the last expression is well defined as $[\tilde{\bn}]\subseteq [\bn]$, so $[\tilde{\bn}_\bi]\subseteq [\bn_\bi]$).
We claim that the family $\tilde\cS=\{\tilde S_\bi\}_{\bi\in [q]^p_\rightarrow}$ is a realistic $(p,\tilde{\bn})$-system of shadows. To show that the claim is true, we shall first prove that the equation
\begin{align}
\label{eqn_1417_08062022}
E_\ba\ast
\Pi^{\tilde\bn_\bi}_\br\ast \tilde S_\bi
\spaceeq
\left\{
\begin{array}{llll}
E_\ba\ast\Pi^{\bn_\bi}_\br\ast S_\bi&\mbox{ if }i_{r_{p-1}}= q\\
E_\ba\ast(\Pi^{\bn_\bi}_\br\ast S_\bi-\hat{S}_{\bi_\br})&\mbox{ otherwise}
\end{array}
\right.
\end{align} 
is satisfied for any $\bi\in [q]^p_\rightarrow$, any $\br\in [p]^{p-1}_\rightarrow$, and any $\ba\in [\tilde{\bn}_{\bi_\br}]$. First, notice that, if $i_p=q$,
\begin{align*}
[\tilde{\bn}_\bi]
&\spaceeq
[\tilde{n}_{i_1}]\times\dots\times [\tilde{n}_{i_{p-1}}]\times [\tilde{n}_{i_p}]
\spaceeq
[{n}_{i_1}]\times\dots\times [n_{i_{p-1}}]\times [{n}_{q}-1]
\spaceeq
\{\bb\in [\bn_\bi]:b_p\neq n_q\}
\end{align*}
while, if $i_p\neq q$, $\tilde{\bn}_\bi=\hat\bn_\bi=\bn_\bi$, so $[\tilde{\bn}_\bi]=[\hat\bn_\bi]=[\bn_\bi]$.
Suppose that $i_{r_{p-1}}=q$. In this case, we have $r_{p-1}=p$ and $i_p=q$. Hence, 
\begin{align*}
E_\ba\ast\Pi^{\tilde\bn_\bi}_\br\ast \tilde S_\bi
&\lemeq{lem_Pi_basic} 
\sum_{\substack{\bb\in [\tilde{\bn}_\bi]\\ \bb_\br=\ba}}
E_\bb\ast\tilde{S}_\bi
\spaceeq
\sum_{\substack{\bb\in [{\bn}_\bi]\\ \bb_\br=\ba\\ b_p\neq n_q}}
E_\bb\ast S_\bi
\spaceeq
\sum_{\substack{\bb\in [{\bn}_\bi]\\ \bb_\br=\ba}}
E_\bb\ast S_\bi
\lemeq{lem_Pi_basic}
E_\ba\ast\Pi^{\bn_\bi}_{\br}\ast S_\bi,
\end{align*}
so~\eqref{eqn_1417_08062022} holds in this case. Suppose now that $i_{r_{p-1}}\neq q$. This can happen either if $i_p\neq q$ (\underline{case \emph{a}}), or if $i_p=q$ and $r_{p-1}\neq p$ (\underline{case \emph{b}}), and it implies that $\bi_\br\in [q-1]^{p-1}_\rightarrow$. In \underline{case \emph{a}}, 
\begin{align*}
\Pi^{\tilde\bn_\bi}_\br\ast \tilde S_\bi
&\spaceeq
\Pi^{\bn_\bi}_\br\ast (S_\bi-\Pi^{\hat\bn}_\bi\ast\hat{C})
\lemeq{lem_associativity_contraction}
\Pi^{\bn_\bi}_\br\ast S_\bi - \Pi^{\bn_\bi}_\br\cont{p}\Pi^{\hat\bn}_\bi\ast\hat{C}
\lemeq{lem_proj_contr}
\Pi^{\bn_\bi}_\br\ast S_\bi - \Pi^{\hat\bn}_{\bi_\br}\ast\hat{C}\\
&\spaceeq
\Pi^{\bn_\bi}_\br\ast S_\bi - \hat{S}_{\bi_\br},
\end{align*}
where the last equality follows from the property of $\hat{C}$. So,~\eqref{eqn_1417_08062022} holds in this case. In \underline{case \emph{b}}, we must have $\br=\ang{p-1}$. Hence,
\begin{align*}
E_\ba\ast\Pi^{\tilde\bn_\bi}_\br\ast \tilde S_\bi
&\lemeq{lem_Pi_basic}
\sum_{\substack{\bb\in [\tilde{\bn}_\bi]\\ \bb_{\ang{p-1}}=\ba}}
E_\bb\ast\tilde{S}_\bi
\spaceeq
\sum_{\substack{\bb\in [{\bn}_\bi]\\ \bb_{\ang{p-1}}=\ba\\ b_p\neq n_q}}
E_\bb\ast S_\bi
\spaceeq
\sum_{\substack{\bb\in [{\bn}_\bi]\\ \bb_{\ang{p-1}}=\ba}}
E_\bb\ast S_\bi
-
E_{(\ba,n_q)}\ast S_\bi\\
&\lemeq{lem_Pi_basic}
E_\ba\ast\Pi^{\bn_\bi}_{\ang{p-1}}\ast S_\bi-E_{(\ba,n_q)}\ast S_\bi
\spaceeq
E_\ba\ast\Pi^{\bn_\bi}_{\ang{p-1}}\ast S_\bi-E_{(\ba,n_q)}\ast S_{(\bi_{\ang{p-1}},q)}\\
&\spaceeq
E_\ba\ast\Pi^{\bn_\bi}_{\ang{p-1}}\ast S_\bi-E_\ba\ast \hat S_{\bi_{\ang{p-1}}}
\lemeq{lem_associativity_contraction}
E_\ba\ast(\Pi^{\bn_\bi}_{\br}\ast S_\bi-\hat S_{\bi_{\br}})
,
\end{align*}
where the penultimate equality comes from the definition of $\hat\cS$ and from the fact that, in this case, $\tilde{\bn}_{\bi_\br}=\hat\bn_{\bi_\br}$, so $\ba\in [\hat{\bn}_{\bi_\br}]$. We conclude that~\eqref{eqn_1417_08062022} also holds in \underline{case \emph{b}}. Using~\eqref{eqn_1417_08062022} and the fact that $\cS$ is a realistic system, we easily conclude that $\tilde{\cS}$ is a realistic system, too. Indeed, take $\bi,\bj\in [q]^p_\rightarrow$ and $\br,\bs\in [p]^{p-1}_\rightarrow$ such that $\bi_\br=\bj_\bs$, and choose $\ba\in [\tilde{\bn}_{\bi_\br}]$. Observe that $i_{r_{p-1}}=j_{s_{p-1}}$. If $i_{r_{p-1}}=q$, we find
\begin{align*}
E_\ba\ast\Pi^{\tilde{\bn}_\bi}_\br\ast\tilde{S}_\bi
&\spaceeq
E_\ba\ast\Pi^{\bn_\bi}_\br\ast S_\bi
\spaceeq
E_\ba\ast\Pi^{\bn_\bj}_\bs\ast S_\bj
\spaceeq
E_\ba\ast\Pi^{\tilde{\bn}_\bj}_\bs\ast\tilde{S}_\bj;
\intertext{otherwise,}
E_\ba\ast\Pi^{\tilde{\bn}_\bi}_\br\ast\tilde{S}_\bi
&\spaceeq
E_\ba\ast(\Pi^{\bn_\bi}_\br\ast S_\bi-\hat{S}_{\bi_\br})
\spaceeq
E_\ba\ast(\Pi^{\bn_\bj}_\bs\ast S_\bj-\hat{S}_{\bj_\bs})
\spaceeq
E_\ba\ast\Pi^{\tilde{\bn}_\bj}_\bs\ast\tilde{S}_\bj.
\end{align*}
It follows that $\tilde{\cS}$ is indeed a realistic system, as claimed.
Since $\tilde{\bn}^T\bone_q=\bn^T\bone_q-1$, we can then apply the inductive hypothesis to deduce that $\tilde\cS$ is realisable, so there exists a tensor $\tilde{C}\in\cT^{\tilde{\bn}}(\Z)$ such that $\Pi^{\tilde{\bn}}_\bi\ast\tilde{C}=\tilde{S}_\bi$ for each $\bi\in [q]^p_\rightarrow$. 

We now define a tensor $C\in\cT^{\bn}(\Z)$ by setting, for each $\bb\in [\bn]$,
\begin{align}
\label{eqn_2124_0507}
E_\bb\ast C
\spaceeq
\left\{
\begin{array}{llll}
E_{\bb_{\langle q-1\rangle}}\ast \hat{C}&\mbox{ if }& b_q=n_q\\[5pt]
E_{\bb}\ast \tilde{C}&\mbox{ if }&b_q\neq n_q.
\end{array}
\right.
\end{align}
Take $\bi\in [q]^p_\rightarrow$ and $\ba\in [\bn_\bi]$. To conclude the proof, we need to show that 
\begin{align}
\label{eqn_goal_1259_07062022}
E_\ba\ast\Pi^\bn_\bi\ast C
\spaceeq
E_\ba\ast S_\bi.
\end{align}
Observe that 
\begin{align}
\label{eqn_1311_07062022}
E_\ba\ast\Pi^\bn_\bi\ast C
&\lemeq{lem_Pi_basic}
\sum_{\substack{\bb\in [\bn]\\ \bb_\bi=\ba}}E_\bb\ast C
\spaceeq
\sum_{\substack{\bb\in [\bn]\\ \bb_\bi=\ba\\ b_q=n_q}}E_\bb\ast C
+
\sum_{\substack{\bb\in [\bn]\\ \bb_\bi=\ba\\ b_q\neq n_q}}E_\bb\ast C
\equationeq{eqn_2124_0507}
\sum_{\substack{\bb\in [\bn]\\ \bb_\bi=\ba\\ b_q=n_q}}E_{\bb_{\ang{q-1}}}\ast \hat C
+
\sum_{\substack{\bb\in [\tilde\bn]\\ \bb_\bi=\ba}}E_\bb\ast \tilde C.
\end{align}
Let us denote the first and the second summand of the rightmost expression in~\eqref{eqn_1311_07062022} by $\alpha$ and $\beta$, respectively.
Suppose first that $i_p=q$. If $a_p\neq n_q$, we see that $\alpha=0$, so
\begin{align*}
E_\ba\ast\Pi^\bn_\bi\ast C
&\equationeq{eqn_1311_07062022}
\sum_{\substack{\bb\in [\tilde\bn]\\ \bb_\bi=\ba}}E_\bb\ast \tilde C
\lemeq{lem_Pi_basic}
E_\ba\ast\Pi^{\tilde\bn}_\bi\ast \tilde{C}
\spaceeq
E_\ba\ast\tilde{S}_\bi
\spaceeq
E_\ba\ast S_\bi; 
\end{align*} 
if $a_p=n_q$, we get $\beta=0$, so  
\begin{align*}
E_\ba\ast\Pi^\bn_\bi\ast C
&\equationeq{eqn_1311_07062022}
\sum_{\substack{\bb\in [\bn]\\ \bb_\bi=\ba\\ b_q=n_q}}E_{\bb_{\ang{q-1}}}\ast \hat C
\spaceeq
\sum_{\substack{\bb\in [\bn]\\ \bb_\bi=\ba}}E_{\bb_{\ang{q-1}}}\ast \hat C
\spaceeq
\sum_{\substack{\bc\in [\hat \bn]\\\bc_{\bi_{\ang{p-1}}}=\ba_{\ang{p-1}}}}E_\bc\ast\hat{C}\\
&\lemeq{lem_Pi_basic}
E_{\ba_{\ang{p-1}}}\ast\Pi^{\hat\bn}_{\bi_{\ang{p-1}}}\ast \hat{C}
\spaceeq
E_{\ba_{\ang{p-1}}}\ast\hat{S}_{\bi_{\ang{p-1}}}
\spaceeq
E_{(\ba_{\ang{p-1}},n_q)}\ast S_{(\bi_{\ang{p-1}},q)}
\spaceeq
E_\ba\ast S_\bi.
\end{align*}
Suppose now that $i_p\neq q$, in which case $\bi\in [q-1]^p_\rightarrow$. We obtain
\begin{align*}
\alpha
&\spaceeq
\sum_{\substack{\bb\in [\bn]\\ \bb_\bi=\ba\\ b_q=n_q}}E_{\bb_{\ang{q-1}}}\ast \hat C
\spaceeq
\sum_{\substack{\bc\in [\hat{\bn}]\\ \bc_{\bi}=\ba}}E_\bc\ast\hat{C}
\lemeq{lem_Pi_basic}
E_\ba\ast\Pi^{\hat{\bn}}_{\bi}\ast\hat C,\\
\beta
&\spaceeq
\sum_{\substack{\bb\in [\tilde\bn]\\ \bb_\bi=\ba}}E_\bb\ast \tilde C
\lemeq{lem_Pi_basic}
E_\ba\ast\Pi^{\tilde{\bn}}_{\bi}\ast\tilde{C} 
\spaceeq
E_\ba\ast\tilde{S}_\bi
\spaceeq
E_\ba\ast(S_\bi-\Pi^{\hat\bn}_\bi\ast\hat{C})
\spaceeq
E_\ba\ast S_\bi - E_\ba\ast\Pi^{\hat\bn}_\bi\ast\hat{C},
\end{align*}
and it follows that
\begin{align*}
E_\ba\ast\Pi^\bn_\bi\ast C
&\equationeq{eqn_1311_07062022}
\alpha+\beta
\spaceeq
E_\ba\ast S_\bi.
\end{align*}
Therefore,~\eqref{eqn_goal_1259_07062022} holds, $\cS$ is realisable, and the proof is concluded.\qedhere
\end{proof}

\subsection{Crystalisation}
\label{subsec_crystalisation}
\noindent One easily derives from Theorem~\ref{thm_realistic_equals_realisable} a \emph{crystalisation} procedure, which turns a given crystal $S$ into a new crystal whose shadow is $S$. This is precisely what is needed to complete the step ($\spadesuit$\,2) of the proof of Theorem~\ref{prop_hollow_shadows_exist}.
\begin{cor}
\label{cor_exogenesis_of_crystals}
Let $n,q\in\N$, let $k\in[q]$, and let $S\in\mathcal{T}^{n\cdot\bone_k}(\Z)$ be a $(k-1)$-crystal. Then there exists a $k$-crystal $C\in\mathcal{T}^{n\cdot\bone_q}(\Z)$ whose $k$-shadow is $S$.
\end{cor}
\begin{proof}
Consider the $(k,n\cdot\bone_q)$-system of shadows $\mathcal{S}=\{S_\bi\}_{\bi\in [q]^k_\to}$ obtained by setting $S_\bi=S$ for each $\bi\in [q]^k_\to$. The fact that $S$ is a $(k-1)$-crystal immediately implies that $\mathcal{S}$ is a realistic system of shadows. Using Theorem~\ref{thm_realistic_equals_realisable}, we deduce that $\mathcal{S}$ is realisable---i.e., there exists a tensor $C\in\mathcal{T}^{n\cdot\bone_q}(\Z)$ such that $\Pi^{n\cdot\bone_q}_\bi\ast C=S$ for each $\bi\in [q]^k_\to$. It follows that $C$ is a $k$-crystal, whose $k$-shadow is $S$.
\end{proof}
Before proceeding to the next steps towards the proof of Theorem~\ref{prop_hollow_shadows_exist}, we illustrate the crystalisation procedure on a concrete example, by showing how to produce a $4$-dimensional $2$-crystal having a given shadow through the construction described in Section~\ref{subsec_system_of_shadows_BODY}.

\begin{example}
\label{example_four_dimensional_crystal}
Throughout this example, we shall indicate the numbers $-2$, $-1$, ${0}$, $1$, $2$, and $3$ by the colours blue, green, light grey, yellow, orange, and red, respectively.

Take $n=3$, $q=4$, and $k=2$ in the statement of Corollary~\ref{cor_exogenesis_of_crystals}. The goal is to build a $2$-crystal $C\in\cT^{3\cdot\bone_4}(\Z)$ whose $2$-shadow is the matrix $\squareTensorA{1}{.25}$\, (which is easily shown to be a $1$-crystal, as the row- and column-sum vectors coincide).
To this end, we consider the $(2,3\cdot\bone_4)$-system of shadows $\cS$ whose members are all equal to 
$\squareTensorA{1}{.25}$\,. $\cS$ is trivially realistic. The goal is to show that it is realisable; indeed, the tensor $C\in\cT^{3\cdot\bone_4}(\Z)$ witnessing this fact would be the crystal we seek.
Following the proof of Theorem~\ref{thm_realistic_equals_realisable}, we create two auxiliary systems of shadows $\hat{\cS}$ and $\tilde{\cS}$. $\hat\cS$ is a $(1,3\cdot\bone_3)$-system---i.e., both the shadows and the tensor that is claimed to realise
them have one fewer dimension than those for the original system $\cS$. In particular, we see from the proof that all members of $\hat\cS$ are the same vector $\lineTensorA{1}{.25}\,$.
Again, it is not hard to verify that $\hat\cS$ is a realistic system. To check that it is realisable, we only need to find a $3$-dimensional tensor of width $3$ such that summing its entries along all three modes yields $\lineTensorA{1}{.25}\,$.
Either by inspection or using the proof of Lemma~\ref{lem_realistic_realisable_case_p_1}, we find that 
\begin{align}
\label{eqn_1246_1902}
\hat{C}\spaceeq\mbox{\cubeTensorA{-3}{.3}}\in\cT^{3\cdot\bone_3}(\Z)
\end{align}
satisfies these conditions.
The second auxiliary system of shadows is the $(2,(3,3,3,2))$-system $\tilde\cS$ defined as follows: $\tilde{S}_{(1,4)}=\tilde{S}_{(2,4)}=\tilde{S}_{(3,4)}=\rectangularTensorA{1}{.25}\,$ (i.e., the matrix obtained by slicing off the rightmost column of $\squareTensorA{1}{.25}\,$); each of the other members of the system is obtained by taking the corresponding matrix in $\cS$ and subtracting from it the projection of $\hat C$ onto the corresponding modes (i.e., $\tilde{S}_\bi=S_\bi-\Pi^{3\cdot\bone_3}_\bi\ast\hat C$).
We see from~\eqref{eqn_1246_1902} that all three projections $\Pi^{3\cdot\bone_3}_{(1,2)}\ast\hat C$, $\Pi^{3\cdot\bone_3}_{(1,3)}\ast\hat C$, and $\Pi^{3\cdot\bone_3}_{(2,3)}\ast\hat C$ are equal to $\squareTensorBNewProjection{1}{.25}\,$.
Hence, we obtain $$\tilde S_{(1,2)}\spaceeq\tilde S_{(1,3)}\spaceeq\tilde S_{(2,3)}\spaceeq\squareTensorA{2}{.4}\,-\,\squareTensorBNewProjection{2}{.4}\spaceeq\squareTensorB{2}{.4}\,.$$ This system is also realistic, and it is such that the sum of the sizes of the modes of the tensor $\tilde{C}$ that is claimed to realise it is strictly smaller than the corresponding quantity for the system $\cS$. At this point, we simply iterate the process, by repeatedly ``slicing'' $\tilde{\cS}$ into a system of $1$-dimensional shadows 
(which we handle through Lemma~\ref{lem_realistic_realisable_case_p_1}) and a smaller system of $2$-dimensional shadows, 
until we end up with a system such that the sizes of all modes are shrunk to $1$, so that the tensor realising it is a single number (see Lemma~\ref{lem_1_dim_albums_are_truthful}).
Throughout this process, Lemma~\ref{lem_rotating_preserves_truth} guarantees that the tensors can be rotated in a way that we slice along the rightmost mode, thus avoiding complications with the orientations of the shadows.
In this way, we find that the system $\tilde{\cS}$ is realised by the tensor $\tilde C$ whose two blocks are $$\cubeTensorB{-3}{.3}\,$$ 
and the all-zero $3\times 3\times 3$ tensor, respectively. Finally, to obtain a tensor $C$ realising the initial system $\cS$ (i.e., a $4$-dimensional $2$-crystal having shadow $\squareTensorA{1}{.25}$\,), we glue together $\tilde{C}$ and $\hat{C}$.
The result is shown in Figure~\ref{fig_four_dimensional_crystal}.
\begin{figure}
\begin{center}
\fourDimCrystal{1}{1}
\end{center}
\caption{A $4$-dimensional $2$-crystal having shadow {\protect\squareTensorA{1}{.25}}\hspace{-.05cm}.}
\label{fig_four_dimensional_crystal}
\end{figure}
\end{example}

\subsection{Quartzes}
\label{subsec_quartzes_BODY}
The crystalisation procedure destroys hollowness: Even when the crystal $S$ in the statement of Corollary~\ref{cor_exogenesis_of_crystals} is hollow, the new crystal $C$ resulting from the crystalisation is not hollow in general---as it is clear from Example~\ref{example_four_dimensional_crystal}.
There does not appear to be
a natural way of modifying the inductive construction in Section~\ref{subsec_system_of_shadows_BODY}
to require that hollowness be preserved along the process. Hence, to achieve hollowness, we employ a second, separate procedure---step ($\spadesuit$\,4)---which consists in applying multiple \emph{local modifications} to the crystal resulting from step ($\spadesuit$\,2) (after expanding it with layers of zeros in step ($\spadesuit$\,3)). These modifications are associated with certain transparent crystals defined next.
\begin{defn}
\label{defn_quartzes}
Let $k,n\in\N$, and let $\ba,\bb\in [n]^k$ be such that $a_i\neq b_i$ for each $i\in [k]$. Given $\bz\in \{0,1\}^k$, let $h(\bz;\ba,\bb)$ be the tuple in $[n]^k$ whose $i$-th entry is $a_i$ if $z_i=0$, $b_i$ otherwise. The \emph{quartz} $Q_{\ba,\bb}$ is the tensor in $\mathcal{T}^{n\cdot\bone_k}(\Z)$ defined by $Q_{\ba,\bb}=\sum_{\bz\in\{0,1\}^k}(-1)^{\bz^T\bone_k}E_{h(\bz;\ba,\bb)}$. Equivalently, $E_{h(\bz;\ba,\bb)}\ast Q_{\ba,\bb}=(-1)^{\bz^T\bone_k}$ for each $\bz\in \{0,1\}^k$, and all other entries are zero. 
\end{defn}
Let the symbol ``$\odot$'' indicate the entrywise multiplication of tuples having the same length.
\begin{rem}
\label{remark_on_quartzes}
It is straightforward to check that, for any two tuples $\bz,\hat\bz\in\{0,1\}^k$, $\bz=\hat{\bz}$ if and only if $h(\bz;\ba,\bb)=h(\hat\bz;\ba,\bb)$.
We can write 
\begin{align}
\label{eqn_1835_4_oct}
h(\bz;\ba,\bb)=(\bone_k-\bz)\odot\ba+\bz\odot\bb.
\end{align}
Notice that the operation of tuple projection distributes over ``$\odot$'', in the sense that $(\bu\odot\bv)_\bi=\bu_\bi\odot\bv_\bi$. Hence, for any $\ell\in\N$ and any  $\bj\in [k]^\ell$, 
\begin{align}
\label{eqn_1817_60oct}
    [h(\bz;\ba,\bb)]_\bj
    &\equationeq{eqn_1835_4_oct}
    [(\bone_k-\bz)\odot\ba+\bz\odot\bb]_\bj
    \spaceeq
    (\bone_\ell-\bz_\bj)\odot\ba_\bj+\bz_\bj\odot\bb_\bj
    \equationeq{eqn_1835_4_oct}
    h(\bz_\bj;\ba_\bj,\bb_\bj).
\end{align}
\end{rem}
\noindent We will need the following simple lemma on crystals.
\begin{lem}
\label{lem_conservation_sum_crystals}
Let $q,n\in\N$ and $k\in\{0,\dots,q\}$,
let $C\in\cT^{n\cdot\bone_q}(\Z)$ be a $k$-crystal, and let $S$ be its $k$-shadow. Then $\Pi^{n\cdot\bone_q}_\bepsilon\ast C=\Pi^{n\cdot\bone_k}_\bepsilon\ast S$. In particular, $C$ is affine if and only if $S$ is affine.
\end{lem}
\begin{proof}
Observe that $\ang{k}\in [q]^k_\to$ and $\ang{k}_\bepsilon=\bepsilon$. We obtain
\begin{align*}
    \Pi^{n\cdot\bone_q}_\bepsilon\ast C
    \spaceeq
    \Pi^{n\cdot\bone_q}_{\ang{k}_\bepsilon}\ast C
    \lemeq{lem_proj_contr}
    \left(\Pi^{n\cdot\bone_k}_\bepsilon\cont{k}\Pi^{n\cdot\bone_q}_\ang{k}\right)\ast C
    \lemeq{lem_associativity_contraction}
    \Pi^{n\cdot\bone_k}_\bepsilon\ast\left(\Pi^{n\cdot\bone_q}_\ang{k}\ast C\right)
    \spaceeq
    \Pi^{n\cdot\bone_k}_\bepsilon\ast S,
\end{align*}
as required. Then, the last part of the statement directly follows from the definition of an affine tensor (Definition~\ref{defn_affine_tensors}).
\end{proof}

\noindent The next proposition collects certain properties of quartzes that shall be useful later.

\begin{prop}
\label{prop_basic_stuff_quartzes}
Let $k,n\in\N$, and let $\ba,\bb\in [n]^k$ be such that $a_i\neq b_i$ for each $i\in [k]$. Then 
\begin{enumerate}[(i)]
    \item $\supp(Q_{\ba,\bb})=\{a_1,b_1\}\times\{a_2,b_2\}\times\dots\times\{a_k,b_k\}$.
    \label{lem_part_supp_quartzes}
    \item $E_\ba\ast Q_{\ba,\bb}=1$.
    \label{lem_part_isentryone}
    \item $\Pi^{n\cdot\bone_k}_\bell\ast Q_{\ba,\bb}=Q_{\ba_\bell,\bb_\bell}$ for any $\bell\in [k]^k$ such that $|\bell|=k$.
    \label{lem_part_inveriance_under_tuple_projection}
    \item $Q_{\ba,\bb}$ is a $(k-1)$-crystal, and its $(k-1)$-shadow is the all-zero tensor in $\mathcal{T}^{n\cdot\bone_{k-1}}(\Z)$.
    \label{lem_part_iscrystal}
    \item $\Pi^{n\cdot\bone_k}_\bepsilon\ast Q_{\ba,\bb}=0$.
    \label{lem_part_issumzero}
\end{enumerate}
\end{prop}
\begin{proof}
To prove~\eqref{lem_part_supp_quartzes}, take $S=\{a_1,b_1\}\times\dots\times\{a_k,b_k\}\subseteq[n]^k$. The map $\bz\mapsto h(\bz;\ba,\bb)$ yields a bijection between $\{0,1\}^k$ and $S$. Hence,
\begin{align*}
    \supp(Q_{\ba,\bb})
    \spaceeq
    \bigcup_{\bz\in\{0,1\}^k}\supp(E_{h(\bz;\ba,\bb)})
    \spaceeq
    \bigcup_{\bz\in\{0,1\}^k}\{{h(\bz;\ba,\bb)}\}
    \spaceeq
    \bigcup_{\bs\in S}\{\bs\}
    \spaceeq
    S.
\end{align*}

To prove~\eqref{lem_part_isentryone}, observe that $\ba=h(\bzero_k;\ba,\bb)$, whence we find
\begin{align*}
    E_\ba\ast Q_{\ba,\bb}
    &\spaceeq
    \sum_{\bz\in\{0,1\}^k}(-1)^{\bz^T\bone_k}E_\ba\ast E_{h(\bz;\ba,\bb)}
    \spaceeq
    (-1)^{\bzero_k^T\bone_k}
    \spaceeq
    1.
\end{align*}
%
\iffalse
To prove~\eqref{lem_part_isreverse}, observe first, using~\eqref{eqn_1835_4_oct}, that
\begin{align*}
    h(\bz;\bb,\ba)
    &\spaceeq
    (\bone_k-\bz)\odot\bb+\bz\odot\ba
    \spaceeq
    h(\bone_k-\bz;\ba,\bb).
\end{align*}
Therefore, we obtain
\begin{align*}
    Q_{\bb,\ba}
    &\spaceeq
    \sum_{\bz\in\{0,1\}^k}(-1)^{\bz^T\bone_k}E_{h(\bz;\bb,\ba)}
    \spaceeq
    \sum_{\bz\in\{0,1\}^k}(-1)^{\bz^T\bone_k}E_{h(\bone_k-\bz;\ba,\bb)}\\
    &\spaceeq
    \sum_{\bz\in\{0,1\}^k}(-1)^{(\bone_k-\bz)^T\bone_k}E_{h(\bz;\ba,\bb)}
    \spaceeq
    \sum_{\bz\in\{0,1\}^k}(-1)^{k-\bz^T\bone_k}E_{h(\bz;\ba,\bb)}\\
    &\spaceeq
    (-1)^k\sum_{\bz\in\{0,1\}^k}(-1)^{\bz^T\bone_k}E_{h(\bz;\ba,\bb)}
    \spaceeq
    (-1)^kQ_{\ba,\bb}.
\end{align*}
%
\fi

To prove~\eqref{lem_part_inveriance_under_tuple_projection}, 
observe that
\begin{align}
\label{eqn_1750_6oct_A}
\notag
   Q_{\ba_\bell,\bb_\bell}
   &\spaceeq
   \sum_{\bz\in\{0,1\}^k}(-1)^{\bz^T\bone_k}E_{h(\bz;\ba_\bell,\bb_\bell)}
   \spaceeq
   \sum_{\bz\in\{0,1\}^k}(-1)^{\bz_\bell^T\bone_k}E_{h(\bz_\bell;\ba_\bell,\bb_\bell)}\\
   &\spaceeq
   \sum_{\bz\in\{0,1\}^k}(-1)^{\bz^T\bone_k}E_{h(\bz_\bell;\ba_\bell,\bb_\bell)},
\end{align}
where the second equality is obtained by noting that summing over $\bz$ is equivalent to summing over $\bz_\bell$, since $|\bell|=k$. On the other hand, letting $\bj\in [k]^k$ be the tuple for which $\bell_\bj=\bj_\bell=\ang{k}$,
\begin{align}
\label{eqn_1750_6oct_B}
    \notag
    \Pi^{n\cdot\bone_k}_\bell\ast Q_{\ba,\bb}
    &\spaceeq
    \sum_{\bc\in [n]^k}(E_\bc\ast\Pi^{n\cdot\bone_k}_\bell\ast Q_{\ba,\bb})E_\bc
    \lemeq{lem_Pi_basic}
    \sum_{\bc\in [n]^k}\bigg(\sum_{\substack{\bd\in [n]^k\\\bd_\bell=\bc}}E_\bd\ast Q_{\ba,\bb}\bigg)E_\bc\\
    \notag
    &\spaceeq
    \sum_{\bc\in [n]^k}(E_{\bc_\bj}\ast Q_{\ba,\bb}) E_\bc
    \spaceeq
    \sum_{\bc\in [n]^k}(E_{\bc}\ast Q_{\ba,\bb}) E_{\bc_\bell}\\
    \notag
    &\spaceeq
    \sum_{\bc\in [n]^k}\sum_{\bz\in\{0,1\}^k}(-1)^{\bz^T\bone_k}(E_\bc\ast E_{h(\bz;\ba,\bb)}) E_{\bc_\bell}\\
    \notag
    &\spaceeq
    \sum_{\bz\in\{0,1\}^k}(-1)^{\bz^T\bone_k}\sum_{\bc\in [n]^k}(E_\bc\ast E_{h(\bz;\ba,\bb)}) E_{\bc_\bell}
    \spaceeq
    \sum_{\bz\in\{0,1\}^k}(-1)^{\bz^T\bone_k}E_{[h(\bz;\ba,\bb)]_\bell}\\
    &\equationeq{eqn_1817_60oct}
    \sum_{\bz\in\{0,1\}^k}(-1)^{\bz^T\bone_k}E_{h(\bz_\bell;\ba_\bell,\bb_\bell)}.
\end{align}
Combining~\eqref{eqn_1750_6oct_A} and~\eqref{eqn_1750_6oct_B}, we obtain $\Pi^{n\cdot\bone_k}_\bell\ast Q_{\ba,\bb}=Q_{\ba_\bell,\bb_\bell}$. 

To prove~\eqref{lem_part_iscrystal}, observe that, for any $\bc\in [n]^{k-1}$,
\begin{align}
\label{eqn_1825_6oct}
    \notag
    E_\bc\ast\Pi^{n\cdot\bone_k}_{\ang{k-1}}\ast Q_{\ba,\bb}
    &\lemeq{lem_Pi_basic}
    \sum_{\substack{\bd\in [n]^k\\\bd_{\ang{k-1}}=\bc}}E_\bd\ast Q_{\ba,\bb}
    \spaceeq
    \sum_{d\in [n]}E_{(\bc,d)}\ast Q_{\ba,\bb}\\
    &\spaceeq
    \sum_{\bz\in\{0,1\}^k}(-1)^{\bz^T\bone_k}\sum_{d\in [n]}E_{(\bc,d)}\ast E_{h(\bz;\ba,\bb)}.
\end{align}
In order for a tuple $\bz\in\{0,1\}^k$ to give a nonzero contribution to the sum in the right-hand side of~\eqref{eqn_1825_6oct}, we must have that $(\bc,d)=h(\bz;\ba,\bb)$ for some $d\in [n]$, which implies that
\begin{align*}
    \bc
    \spaceeq
    (\bc,d)_{\ang{k-1}}
    \spaceeq
    [h(\bz;\ba,\bb)]_{\ang{k-1}}
    \equationeq{eqn_1817_60oct}
    h(\bz_{\ang{k-1}};\ba_{\ang{k-1}},\bb_{\ang{k-1}}).
\end{align*}
In particular, $\bz_{\ang{k-1}}=\tilde\bz$ for some $\tilde\bz\in \{0,1\}^{k-1}$ such that $\bc=h(\tilde\bz;\ba_{\ang{k-1}},\bb_{\ang{k-1}})$. Then, it follows from Remark~\ref{remark_on_quartzes} that such tuple $\tilde\bz$ is unique. Notice that $h((\tilde\bz,0);\ba,\bb)=(\bc,a_k)$ and $h((\tilde\bz,1);\ba,\bb)=(\bc,b_k)$.
As a consequence, we can simplify~\eqref{eqn_1825_6oct} to yield
\begin{align*}
    E_\bc\ast\Pi^{n\cdot\bone_k}_{\ang{k-1}}\ast Q_{\ba,\bb}
    &\spaceeq
    \sum_{z\in\{0,1\}}(-1)^{(\tilde\bz,z)^T\bone_k}\sum_{d\in [n]}E_{(\bc,d)}\ast E_{h((\tilde\bz,z);\ba,\bb)}\\
    &\spaceeq
    (-1)^{(\tilde\bz,0)^T\bone_k}\sum_{d\in [n]}E_{(\bc,d)}\ast E_{h((\tilde\bz,0);\ba,\bb)}
    +
    (-1)^{(\tilde\bz,1)^T\bone_k}\sum_{d\in [n]}E_{(\bc,d)}\ast E_{h((\tilde\bz,1);\ba,\bb)}\\
    &\spaceeq
    (-1)^{(\tilde\bz,0)^T\bone_k}\sum_{d\in [n]}E_{(\bc,d)}\ast E_{(\bc,a_k)}
    +
    (-1)^{(\tilde\bz,1)^T\bone_k}\sum_{d\in [n]}E_{(\bc,d)}\ast E_{(\bc,b_k)}\\
    &\spaceeq
    (-1)^{(\tilde\bz,0)^T\bone_k}+(-1)^{(\tilde\bz,1)^T\bone_k}
    \spaceeq
    (-1)^{\tilde\bz^T\bone_{k-1}}-(-1)^{\tilde\bz^T\bone_{k-1}}
    \spaceeq 0.
\end{align*}
It follows that $\Pi^{n\cdot\bone_k}_{\ang{k-1}}\ast Q_{\ba,\bb}$ is the all-zero tensor. Take now $\bi\in [k]^{k-1}_\to$, and let $p$ be the unique element of $[k]\setminus\set(\bi)$. Consider the tuple $\bell=(\bi,p)\in [k]^k$, and notice that $|\bell|=k$ and $\bi=\bell_{\ang{k-1}}$. Hence,
\begin{align*}
    \Pi^{n\cdot\bone_k}_\bi\ast Q_{\ba,\bb}
    &\spaceeq
    \Pi^{n\cdot\bone_k}_{\bell_{\ang{k-1}}}\ast Q_{\ba,\bb}
    \lemeq{lem_proj_contr}
    \left(\Pi^{n\cdot\bone_k}_{\ang{k-1}}\cont{k}\Pi^{n\cdot\bone_k}_\bell\right)\ast Q_{\ba,\bb}
    \lemeq{lem_associativity_contraction}
    \Pi^{n\cdot\bone_k}_{\ang{k-1}}\ast\left(\Pi^{n\cdot\bone_k}_\bell\ast Q_{\ba,\bb}\right)\\
    &\propparteq{prop_basic_stuff_quartzes}{lem_part_inveriance_under_tuple_projection}
    \Pi^{n\cdot\bone_k}_{\ang{k-1}}\ast Q_{\ba_\bell,\bb_\bell},
\end{align*}
which is the all-zero tensor as proved above. This shows that $Q_{\ba,\bb}$ is a $(k-1)$-crystal having the all-zero tensor as its $(k-1)$-shadow.

Finally,~\eqref{lem_part_issumzero} directly follows from~\eqref{lem_part_iscrystal} by applying Lemma~\ref{lem_conservation_sum_crystals}.
\end{proof}

\subsection{Crystals with hollow shadows}
\label{subsec_hollow_shadowed_crystals_BODY}

\noindent We now have all the ingredients for implementing the steps ($\spadesuit$\,1)--($\spadesuit$\,5), thus completing the proof of Theorem~\ref{prop_hollow_shadows_exist}. Once that is established, the existence of hollow-shadowed crystals of quadratic width (Theorem~\ref{thm_existence_crystals_with_hollow_shadow}) can be easily derived.

\begin{proof}[Proof of Theorem~\ref{prop_hollow_shadows_exist}]
We use induction over $k$. For $k=1$, the tensor $C=1$ works. For the inductive step, suppose that $k\geq 2$.
Let $\hat{n}=\frac{k^2-k}{2}$ and $n=\hat{n}+k=\frac{k^2+k}{2}$. By the inductive hypothesis, we find a hollow affine $(k-2)$-crystal $U\in\mathcal{T}^{\hat{n}\cdot\bone_{k-1}}(\Z)$ ($\spadesuit$\,1). Using Corollary~\ref{cor_exogenesis_of_crystals}, we deduce that there exists a (not necessarily hollow) $(k-1)$-crystal $V\in\cT^{\hat{n}\cdot\bone_k}(\Z)$ whose $(k-1)$-shadow is $U$ ($\spadesuit$\,2). By Lemma~\ref{lem_conservation_sum_crystals}, $V$ is affine, too.
Consider now the tensor $W\in\cT^{n\cdot\bone_k}(\Z)$ defined by setting, for each $\ba\in [n]^k$, $E_\ba\ast W=E_\ba\ast V$ if $\set(\ba)\subseteq [\hat{n}]$, $E_\ba\ast W=0$ otherwise; i.e., $W$ is obtained by padding $V$ with $k$ layers of zeros on each mode ($\spadesuit$\,3). Similarly, define $Z\in\cT^{n\cdot\bone_{k-1}}(\Z)$ by setting, for each $\ba\in [n]^{k-1}$, $E_\ba\ast Z=E_\ba\ast U$ if $\set(\ba)\subseteq [\hat{n}]$, $E_\ba\ast Z=0$ otherwise. Observe that $\supp(U)=\supp(Z)$, so $U$ being hollow implies $Z$ being hollow as well. We claim that $W$ is a $(k-1)$-crystal whose $(k-1)$-shadow is $Z$. Indeed, for any $\bi\in [k]^{k-1}_\to$ and $\ba\in [n]^{k-1}$,
\begin{align*}
    E_\ba\ast\Pi^{n\cdot\bone_k}_\bi\ast W
    &\lemeq{lem_Pi_basic}
    \sum_{\substack{\bb\in [n]^k\\\bb_\bi=\ba}}E_\bb\ast W
    =
    \sum_{\substack{\bb\in [\hat{n}]^k\\\bb_\bi=\ba}}E_\bb\ast V
    \lemeq{lem_Pi_basic}
    \left\{
    \begin{array}{cc}
         E_\ba\ast\Pi^{\hat{n}\cdot\bone_k}_\bi\ast V &  \mbox{ if }\set(\ba)\subseteq [\hat{n}]\\
         0 & \mbox{ otherwise} 
    \end{array}
    \right.\\
    &=
    \left\{
    \begin{array}{cc}
         E_\ba\ast U &  \mbox{ if }\set(\ba)\subseteq [\hat{n}]\\
         0 & \mbox{ otherwise} 
    \end{array}
    \right.
    =
    E_\ba\ast Z,
\end{align*}
so $\Pi^{n\cdot\bone_k}_\bi\ast W=Z$, as wanted. Clearly, the padding operation does not change the sum of the entries in the tensor, so $W$ is affine. Consider the tuple $\by=(\hat{n}+1,\hat{n}+2,\dots,n)\in [n]^k$, and define ($\spadesuit$\,4) the tensor 
\begin{align}
\label{eqn_2106_15_sept}
    C
    &=
    W-\sum_{\bd\in [\hat{n}]^k}(E_\bd\ast W)Q_{\bd,\by}.
\end{align}
Note that $C\in\cT^{n\cdot\bone_k}(\Z)$.
We shall prove that $C$ is a hollow affine $(k-1)$-crystal. Recall that $W$ is an affine $(k-1)$-crystal. 
Since tensor projection is a linear operation, crystals are preserved under linear combinations. Hence,
it follows from Proposition~\ref{prop_basic_stuff_quartzes}\eqref{lem_part_iscrystal} that $C$ is a $(k-1)$-crystal, too, having the same $(k-1)$-shadow as $W$---namely, $Z$. Similarly, $C$ is affine by virtue of Proposition~\ref{prop_basic_stuff_quartzes}\eqref{lem_part_issumzero}. Hence, we are left to show that $C$ is hollow. To this end, we show that no tuple $\bb\in [n]^k$ is a tie for $C$. This is proved by induction over the quantity $\ell(\bb)=|\{i\in [k]:b_i>\hat{n}\}|$. For the basis of the induction, suppose that $\ell(\bb)=0$ (which means that $\bb\in [\hat{n}]^k$). 
Observe that the choice of $\by$ guarantees that $\set(\by)$ is disjoint from $\set(\bd)$ for each $\bd\in [\hat{n}]^k$. 
We find
\begin{align*}
    E_\bb\ast C
    &\equationeq{eqn_2106_15_sept}
    E_\bb\ast W-
    \sum_{\bd\in [\hat{n}]^k}(E_\bd\ast W)(E_\bb\ast Q_{\bd,\by})
    \propparteq{prop_basic_stuff_quartzes}{lem_part_supp_quartzes}
    E_\bb\ast W-
    (E_\bb\ast W)(E_\bb\ast Q_{\bb,\by})\\
    &\propparteq{prop_basic_stuff_quartzes}{lem_part_isentryone}
    E_\bb\ast W-E_\bb\ast W
    \spaceeq
    0,
\end{align*}
which means, in particular, that $\bb$ is not a tie for $C$. 
Before dealing with the inductive step, we establish the following fact:
\begin{align}
\label{eqn_claim_1508_29sep}
\mbox{\emph{If $\bc\in\supp(C)$ and $c_i>\hat{n}$ for some $i\in [k]$, then $c_i=\hat{n}+i$.}}
\end{align}
To prove~\eqref{eqn_claim_1508_29sep}, observe that $\set(\bc)\not\subseteq [\hat{n}]$, so $\bc\not\in\supp(W)$. Therefore, it follows from~\eqref{eqn_2106_15_sept} that $\bc\in\supp(Q_{\bd,\by})$ for some $\bd\in [\hat{n}]^k$. 
Using~Proposition~\ref{prop_basic_stuff_quartzes}\eqref{lem_part_supp_quartzes}, we conclude that $c_i=y_i=\hat{n}+i$, as claimed.

Take now $\bb\in [n]^k$ with $\ell(\bb)\geq 1$, and let $j\in [k]$ be such that $b_j>\hat{n}$. Suppose, for the sake of contradiction, that $\bb$ is a tie for $C$; i.e., $|\bb|< k$ and $\bb\in\supp(C)$. Let $\alpha<\beta\in [k]$ be such that $b_\alpha=b_\beta$. Notice that $b_\alpha=b_\beta\in [\hat{n}]$ as, otherwise,~\eqref{eqn_claim_1508_29sep} would yield $b_\alpha=\hat{n}+\alpha\neq\hat{n}+\beta=b_\beta$, a contradiction. In particular, this means that $j\not\in\{\alpha,\beta\}$. Define $\tilde\alpha=\alpha$ if $\alpha<j$, and $\tilde\alpha=\alpha-1$ if $\alpha>j$. Similarly, define $\tilde\beta=\beta$ if $\beta<j$, and $\tilde\beta=\beta-1$ if $\beta>j$. Consider also the tuple $\bi\in [k]^{k-1}_\to$ obtained by removing the $j$-th element from $\ang{k}$, and observe that $b_{i_{\tilde\alpha}}=b_\alpha$ and $b_{i_{\tilde\beta}}=b_\beta$, so $b_{i_{\tilde\alpha}}=b_{i_{\tilde\beta}}$. We note that $\tilde\alpha\neq\tilde\beta$. Indeed, $\tilde\alpha=\tilde\beta$ would imply that $\tilde\alpha=\alpha$ and $\tilde\beta=\beta-1$, from which it would follow that $\alpha<j<\beta$ and that $\alpha=\beta-1$, a contradiction. 
As a consequence, $|\bb_\bi|<k-1$. Since $Z$ is hollow, it follows that $\bb_\bi\not\in\supp(Z)$. 
For any $a\in [n]$, let $\bb^{(a)}$ be the tuple in $[n]^k$ obtained by replacing the $j$-th element of $\bb$ with $a$. We find
\begin{align}
\label{eqn_1634_29sep}
\notag
    0
    &\spaceeq
    E_{\bb_\bi}\ast Z
    \spaceeq
    E_{\bb_\bi}\ast\Pi^{n\cdot\bone_k}_\bi\ast C
    \lemeq{lem_Pi_basic}
    \sum_{\substack{\ba\in [n]^k\\\ba_\bi=\bb_\bi}}E_\ba\ast C
    \spaceeq
    \sum_{a\in [n]}E_{\bb^{(a)}}\ast C\\
    &\spaceeq
    \sum_{a\in [\hat{n}]}E_{\bb^{(a)}}\ast C
    +
    \sum_{a\in [n]\setminus [\hat{n}]}E_{\bb^{(a)}}\ast C
\end{align}
where the second equality follows from the fact that $Z$ is the $(k-1)$-shadow of $C$.
If $a\in [\hat{n}]$, $\ell(\bb^{(a)})=\ell(\bb)-1$. Moreover, using that $j\not\in\{\alpha,\beta\}$, we have $b^{(a)}_\alpha=b_\alpha=b_\beta=b^{(a)}_\beta$, which means that $|\bb^{(a)}|<k$. Using the inductive hypothesis, we deduce that $\bb^{(a)}\not\in\supp(C)$, so $E_{\bb^{(a)}}\ast C=0$. If $a\in [n]\setminus [\hat{n}]$ and $\bb^{(a)}\in\supp(C)$, applying~\eqref{eqn_claim_1508_29sep} twice yields $a=\hat{n}+j=b_j$, which implies that $\bb^{(a)}=\bb$. Therefore, it follows from~\eqref{eqn_1634_29sep} that $E_\bb\ast C=0$, thus contradicting our assumptions. This establishes that $C$ is hollow ($\spadesuit$\,5) and concludes the proof of the theorem.
\end{proof}

\begin{proof}[Proof of Theorem~\ref{thm_existence_crystals_with_hollow_shadow}]
Using Theorem~\ref{prop_hollow_shadows_exist}, we find a hollow affine $(k-1)$-crystal $\hat{C}\in\mathcal{T}^{\frac{k^2+k}{2}\cdot\bone_k}(\Z)$. Applying Corollary~\ref{cor_exogenesis_of_crystals}, we find a $k$-crystal $C\in\mathcal{T}^{\frac{k^2+k}{2}\cdot\bone_q}(\Z)$ whose $k$-shadow is $\hat{C}$. The fact that $C$ is affine directly follows from Lemma~\ref{lem_conservation_sum_crystals}.
\end{proof}

\section{Fooling the hierarchy}
\label{sec_fooling_BA}
In this section, we translate the hollow-shadowed crystals built in Section~\ref{sec_crystals} back to the algorithmic framework. This results in a proof of Theorem~\ref{acceptance_prop_BA_big_enough}, which establishes that any loopless digraph is accepted by any level of the $\BA$ hierarchy applied to AGC, \emph{provided that} the number of colours is large enough. Then, we prove two results on the $\BA$ hierarchy (Propositions~\ref{prop_reduction_line_digraph_BA} and~\ref{prop_BA_preserves_homo}, both consequences of more general results on linear minions) that are able to ``boost'' Theorem~\ref{acceptance_prop_BA_big_enough} by relaxing the requirement on the number of colours. These are the last ingredients needed to establish that the family of shift digraphs introduced in Section~\ref{subsec_fooling_BA} provides fooling instances for \emph{all} levels of the $\BA$ hierarchy applied to AGC for \emph{all} numbers of colours, and to finally validate the proof of Theorem~\ref{thm_BLPAIPk_no_solves_AGC} presented in Section~\ref{subsec_fooling_BA}.

\begin{thm*}[Theorem~\ref{acceptance_prop_BA_big_enough} restated]
Let $2\leq k\in\N$ and let $\X$ be a loopless digraph. Then $\BA^k(\X,\K_{(k^2+k)/2})=\YES$.
\end{thm*}
\begin{proof}
We can assume that $\Vset(\X)=[q]$ for some $q\in\N$. 
Moreover, by possibly adding isolated vertices to $\X$, we can assume that $q> k$.
Set $n=\frac{k^2+k}{2}$.
Applying Theorem~\ref{thm_existence_crystals_with_hollow_shadow}, we construct an affine $k$-crystal $C\in\mathcal{T}^{n\cdot\bone_q}(\Z)$ whose $k$-shadow $S\in\cT^{n\cdot\bone_k}(\Z)$ is hollow. We claim that the map
\begin{align*}
    \zeta:\Vset(\X)^k&\to\cT^{n\cdot\bone_k}(\Z)\\
    \bx&\mapsto\Pi^{n\cdot\bone_q}_\bx\ast C
\end{align*}
yields a $k$-tensorial homomorphism from $\Xk$ to $\freeZ(\K_n^\tensor{k})$.

First of all, we need to check that $\zeta(\bx)\in\Vset(\freeZ(\K_n^\tensor{k}))=\Zaff^{(n^k)}$ for each $\bx\in \Vset(\X)^k$. This easily follows from the facts that $C$ has integer entries and %
\begin{align*}
    \Pi^{n\cdot\bone_k}_\bepsilon\ast\zeta(\bx)
    &\spaceeq
    \Pi^{n\cdot\bone_k}_\bepsilon\ast\left(\Pi^{n\cdot\bone_q}_\bx\ast C\right)
    \lemeq{lem_associativity_contraction}
    \left(\Pi^{n\cdot\bone_k}_\bepsilon\ast\Pi^{n\cdot\bone_q}_\bx\right)\ast C
    \lemeq{lem_proj_contr}
    \Pi^{n\cdot\bone_q}_{\bx_\bepsilon}\ast C\\
    &\spaceeq
    \Pi^{n\cdot\bone_q}_\bepsilon\ast C
    \spaceeq
    1,
\end{align*}
where the last equality holds since $C$ is affine.

We now check that $\zeta$ sends hyperedges of $\Xk$ to hyperedges of $\freeZ(\K_n^\tensor{k})$. Take $\bx=(x_1,x_2)\in\Eset(\X)$, so that $\bx^\tensor{k}\in\Eset(\Xk)$. To prove that $\zeta(\bx^\tensor{k})\in\Eset(\freeZ(\K_n^\tensor{k}))$, we need to find some ${\bq}\in\Zaff^{(|\Eset(\K_n)|)}=\Zaff^{(n^2-n)}$ for which $\zeta(\bx_\bi)={\bq}_{/\pi_\bi}$ for each $\bi\in [2]^k$. By Proposition~\ref{prop_crystals_smaller_dimension} we have that $C$ is a $2$-crystal; let $\tilde{S}$ be its $2$-shadow.
Consider the tuple $\balpha$ defined by $\balpha=(1,2)$ if $x_1<x_2$, $\balpha=(2,1)$ if $x_1>x_2$ (notice that $x_1\neq x_2$ as $\X$ is loopless). 
Observe that $\bx_\balpha\in [q]^2_\to$ and $\balpha_\balpha=(1,2)$. We consider the vector ${\bq}\in\cT^{n^2-n}(\Z)$ whose $\ba$-th entry is $E_\ba\ast\Pi^{n\cdot\bone_2}_\balpha\ast \tilde S$ for any $\ba\in\Eset(\K_n)$. Observe that
\begin{align}
\label{eqn_1805_30sep}
\notag
    \tilde{S}
    &\spaceeq
    \Pi^{n\cdot\bone_q}_{\ang{2}}\ast C
    \spaceeq
    \Pi^{n\cdot\bone_q}_{\ang{k}_{\ang{2}}}\ast C
    \lemeq{lem_proj_contr}
    \left(\Pi^{n\cdot\bone_k}_{\ang{2}}\cont{k}\Pi^{n\cdot\bone_q}_{\ang{k}}\right)\ast C
    \lemeq{lem_associativity_contraction}
    \Pi^{n\cdot\bone_k}_{\ang{2}}\ast\left(\Pi^{n\cdot\bone_q}_{\ang{k}}\ast C\right)\\
    &\spaceeq
    \Pi^{n\cdot\bone_k}_{\ang{2}}\ast S,
\end{align}
where the first and fifth equalities come from the fact that $\tilde{S}$ and $S$ are the $2$-shadow and the $k$-shadow of $C$, respectively, while the second equality holds since $\ang{k}_{\ang{2}}=\ang{2}$.
Therefore, for any $a\in [n]$,
\begin{align}
\label{eqn_1448_30sep}
\notag
    E_{(a,a)}\ast\Pi^{n\cdot\bone_2}_\balpha\ast \tilde S
    &\equationeq{eqn_1805_30sep}
    E_{(a,a)}\ast\Pi^{n\cdot\bone_2}_\balpha\ast\left(\Pi^{n\cdot\bone_k}_{\ang{2}}\ast S\right)
    \lemeq{lem_associativity_contraction}
    E_{(a,a)}\ast\left(\Pi^{n\cdot\bone_2}_\balpha\cont{2}\Pi^{n\cdot\bone_k}_{\ang{2}}\right)\ast S\\
    \lemeq{lem_proj_contr}
    &E_{(a,a)}\ast\Pi^{n\cdot\bone_k}_{\ang{2}_\balpha}\ast S
    \spaceeq
    E_{(a,a)}\ast\Pi^{n\cdot\bone_k}_{\balpha}\ast S
    \lemeq{lem_Pi_basic}
    \sum_{\substack{\bb\in [n]^k\\\bb_\balpha=(a,a)}}E_\bb\ast S
    \spaceeq
    0,
\end{align}
where the fourth equality holds since $\ang{2}_\balpha=\balpha$, and the sixth follows from the fact that $S$ is hollow. Hence, we find
\begin{align*}
    \sum_{\ba\in\Eset(\K_n)}E_\ba\ast {\bq}
    &\spaceeq
    \sum_{\ba\in\Eset(\K_n)}E_\ba\ast\Pi^{n\cdot\bone_2}_\balpha\ast \tilde S
    \equationeq{eqn_1448_30sep}
    \sum_{\ba\in [n]^2}E_\ba\ast\Pi^{n\cdot\bone_2}_\balpha\ast \tilde S
    \lemeq{Pi_epsilon_all_one}
    \Pi^{n\cdot\bone_2}_\bepsilon\ast\Pi^{n\cdot\bone_2}_\balpha\ast \tilde S\\
    &\lemeq{lem_proj_contr}
    \Pi^{n\cdot\bone_2}_{\balpha_\bepsilon}\ast \tilde S
    \spaceeq
    \Pi^{n\cdot\bone_2}_\bepsilon\ast \tilde S
    \lemeq{lem_conservation_sum_crystals}
    1,
\end{align*}
whence it follows that ${\bq}\in\Zaff^{(n^2-n)}$.
Given $\bi\in [2]^k$,
we have
\begin{align*}
    \zeta(\bx_\bi)
    &\spaceeq
    \Pi^{n\cdot\bone_q}_{\bx_\bi}\ast C
    \spaceeq
    \Pi^{n\cdot\bone_q}_{\bx_{\balpha_{\balpha_\bi}}}\ast C
    \lemeq{lem_proj_contr}
    \Pi^{n\cdot\bone_2}_\bi\cont{2}\left(\Pi^{n\cdot\bone_2}_\balpha\cont{2}\Pi^{n\cdot\bone_q}_{\bx_\balpha}\right)\ast C\\
    &\lemeq{lem_associativity_contraction}
    \Pi^{n\cdot\bone_2}_\bi\ast\left(\Pi^{n\cdot\bone_2}_\balpha\ast\left(\Pi^{n\cdot\bone_q}_{\bx_\balpha}\ast C\right)\right)
    \spaceeq
    \Pi^{n\cdot\bone_2}_\bi\ast\left(\Pi^{n\cdot\bone_2}_\balpha\ast\tilde S\right).
\end{align*}
It follows that, for any $\ba\in [n]^k$,
\begin{align*}
    E_\ba\ast\zeta(\bx_\bi)
    &\spaceeq
    E_\ba\ast\left(\Pi^{n\cdot\bone_2}_\bi\ast\left(\Pi^{n\cdot\bone_2}_\balpha\ast\tilde S\right)\right)
    \lemeq{lem_associativity_contraction}
    E_\ba\ast \Pi^{n\cdot\bone_2}_\bi\ast\Pi^{n\cdot\bone_2}_\balpha\ast\tilde S
    \lemeq{lem_Pi_basic}
    \sum_{\substack{\bb\in [n]^2\\\bb_\bi=\ba}}E_\bb\ast\Pi^{n\cdot\bone_2}_\balpha\ast\tilde S\\
    &\equationeq{eqn_1448_30sep}
    \sum_{\substack{\bb\in \Eset(\K_n)\\\bb_\bi=\ba}}E_\bb\ast\Pi^{n\cdot\bone_2}_\balpha\ast\tilde S
    \spaceeq
    \sum_{\substack{\bb\in \Eset(\K_n)\\\bb_\bi=\ba}}E_\bb\ast {\bq}
    \lemeq{lem_calculation_rule_P}
    E_\ba\ast P_{\pi_\bi}\ast {\bq}
    \spaceeq
    E_\ba\ast {\bq}_{/\pi_\bi},
\end{align*}
which concludes the proof that $\zeta(\bx_\bi)={\bq}_{/\pi_\bi}$. Hence, $\zeta$ is a homomorphism.

To check that $\zeta$ is $k$-tensorial, simply notice that, for any $\bx\in \Vset(\X)^k$ and $\bi\in [k]^k$,
\begin{align}
\label{eqn_1933_30sep}
    \zeta(\bx_\bi)
    &\spaceeq
    \Pi^{n\cdot\bone_q}_{\bx_\bi}\ast C
    \lemeq{lem_proj_contr}
    \left(\Pi^{n\cdot\bone_k}_\bi\cont{k}\Pi^{n\cdot\bone_q}_\bx\right)\ast C
    \lemeq{lem_associativity_contraction}    
    \Pi^{n\cdot\bone_k}_\bi\ast\left(\Pi^{n\cdot\bone_q}_\bx\ast C\right)
    \spaceeq
    \Pi^{n\cdot\bone_k}_\bi\ast\zeta(\bx).
\end{align}

Take now $\bx\in \Vset(\X)^k$ and $\ba\in [n]^k$, and suppose that $\ba\not\prec\bx$. If we manage to show that $E_\ba\ast\zeta(\bx)=0$, we may apply Theorem~\ref{thm_decoupling_BLPAIP_for_cliques} and conclude that $\BA^k(\X,\K_n)=\YES$, as desired. 
Choose $u,v\in [k]$ for which $a_u=a_v$ and $x_u\neq x_v$. Using that $q> k$, we find $\by\in [q]^k_\to$ and $\bi\in [k]^k$ for which $\bx=\by_\bi$. We obtain
\begin{align}
\label{eqn_1948_30sep}
\notag
    E_\ba\ast\zeta(\bx)
    &\spaceeq
    E_\ba\ast\zeta(\by_\bi)
    \equationeq{eqn_1933_30sep}
    E_\ba\ast\Pi^{n\cdot\bone_k}_\bi\ast\zeta(\by)
    \spaceeq
    E_\ba\ast\Pi^{n\cdot\bone_k}_\bi\ast\left(\Pi^{n\cdot\bone_q}_\by\ast C\right)\\
    &\spaceeq
    E_\ba\ast\Pi^{n\cdot\bone_k}_\bi\ast S
    \lemeq{lem_Pi_basic}
    \sum_{\substack{\bb\in [n]^k\\\bb_\bi=\ba}}E_\bb\ast S.
\end{align}
Suppose that $\bb\in [n]^k$ satisfies $\bb_\bi=\ba$. Then, $b_{i_u}=a_u=a_v=b_{i_v}$. On the other hand, $y_{i_u}=x_u\neq x_v=y_{i_v}$, which implies that $i_u\neq i_v$. As a consequence, $|\bb|<k$. Since $S$ is hollow, we deduce that $\bb\not\in\supp(S)$. Hence, it follows from~\eqref{eqn_1948_30sep} that $E_\ba\ast\zeta(\bx)=0$, as wanted.
\end{proof}

Our next goal is to prove Proposition~\ref{prop_reduction_line_digraph_BA}, which states that $\BA^k$ acceptance is preserved under the line digraph operator introduced in Section~\ref{subsec_fooling_BA}, at the cost of halving the level.
In fact, we shall prove that result in the more general setting of arbitrary conic minions, as stated in Proposition~\ref{prop_reduction_line_digraph}. We need a property of conic minions from~\cite[Proposition~38]{cz23soda:minions}, formally stated below in Proposition~\ref{prop_partial_homo}: Each relaxation hierarchy built on this type of minions only gives a nonzero
weight to those assignments that yield partial homomorphisms. In other words, each such hierarchy enforces consistency. 
\begin{prop}[\cite{cz23soda:minions}]
\label{prop_partial_homo}
Let $\Mminion$ be a conic minion of depth $d$, let $2\leq k\in\N$, let $\X,\A$ be digraphs, and let $\xi:\X^{\tensor{k}}\to\mathbb{F}_{\Mminion}(\A^{\tensor{k}})$ be a $k$-tensorial homomorphism. Take $\bx\in \Vset(\X)^k$, $\ba\in \Vset(\A)^k$, and $\bi\in [k]^2$. If $\bx_\bi\in \Eset(\X)$ and $\ba_\bi\not\in \Eset(\A)$, then $E_\ba\ast\xi(\bx)=\bzero_d$.
\end{prop}
\begin{prop}
\label{prop_reduction_line_digraph}
Let $\Mminion$ be a conic minion, let $2\leq k\in\N$, let $\X,\A$ be 
digraphs, and suppose that there exists a $(2k)$-tensorial homomorphism from $\X^\tensor{2k}$ to $\freeM(\A^\tensor{2k})$.
Then there exists a $k$-tensorial homomorphism from $(\delta\X)^\tensor{k}$ to $\freeM((\delta\A)^\tensor{k})$.
\end{prop}
\begin{proof}
As usual, we let $n=|\Vset(\A)|$; moreover, we let $m=|\Eset(\A)|$.
Take a $(2k)$-tensorial homomorphism $\xi:\X^\tensor{2k}\to\freeM(\A^\tensor{2k})$, 
whose existence is guaranteed by the hypothesis. 

Suppose first that $\Eset(\delta \A)=\emptyset$. We claim that, in this case, $\Eset(\delta\X)=\emptyset$. Otherwise, take some element $((x,y),(y,z))\in\Eset(\delta\X)$, and consider a tuple $\bw\in\Vset(\X)^{2k}$ satisfying $\bw_{\ang{3}}=(x,y,z)$ (where we have used that $k\geq 2$). Since the minion $\Mminion$ is conic, $\xi(\bw)$ is not all zero. Hence, there exists some $\ba\in\Vset(\A)^{2k}$ such that $E_\ba\ast\xi(\bw)\neq\bzero_d$, where $d$ is the depth of $\Mminion$. Applying Proposition~\ref{prop_partial_homo} to the cases $\bi=(1,2)$ and $\bi=(2,3)$, we deduce that $\ba_{(1,2)}\in\Eset(\A)$ and $\ba_{(2,3)}\in\Eset(\A)$. But this means that $((a_1,a_2),(a_2,a_3))\in\Eset(\delta\A)$, a contradiction. Hence, as claimed, $\Eset(\delta\X)=\emptyset$. As a consequence, any map from $\Eset(\X)$ to $\Eset(\A)$ yields a homomorphism $\delta\X\to\delta\A$. Thus, it follows from the completeness of minion tests (see~\cite[Proposition~13]{cz23soda:minions}) that a $k$-tensorial homomorphism from $(\delta\X)^\tensor{k}$ to $\freeM((\delta\A)^\tensor{k})$ exists for any $k\in\N$.

Suppose now that $\Eset(\delta\A)\neq\emptyset$.
Fix $\textbf{t}=(\textbf{e},\textbf{f})\in\Eset(\delta\A)$, where $\textbf{e},\textbf{f}\in \Eset(\A)$, and consider the maps 
\begin{align*}
    \alpha:\Vset(\A)^2&\to \Eset(\A),
    &&&
    \beta:\Vset(\A)^{2k}&\to \Eset(\A)^k
    \\
    (a,b)&\mapsto \left\{
    \begin{array}{cl}
         (a,b)&\mbox{ if }(a,b)\in \Eset(\A)  \\
         \textbf{e}&\mbox{ otherwise} 
    \end{array}\right.
    &&&
    \ba&\mapsto (\alpha(\ba_{(1,2)}),\alpha(\ba_{(3,4)}),\dots,\alpha(\ba_{(2k-1,2k)})).
\end{align*}
Consider also the map $\gamma:\Eset(\X)^k\to \Vset(\X)^{2k}$ sending a tuple $(\bx^{(1)},\bx^{(2)},\dots,\bx^{(k)})$ of edges of $\X$ to the tuple $(x^{(1)}_1,x^{(1)}_2,x^{(2)}_1,x^{(2)}_2,\dots,x^{(k)}_1,x^{(k)}_2)$ of vertices of $\X$, where $\bx^{(i)}=(x^{(i)}_1,x^{(i)}_2)$ for each $i\in [k]$.
We define the map $\vartheta:\Eset(\X)^k\to \Mminion^{(m^k)}$ by setting $\bx\mapsto\xi(\gamma(\bx))_{/\beta}$ for each $\bx\in \Eset(\X)^k$. 
(Observe that the definition of $\vartheta$ is independent of the choice of $\be\in\Eset(\A)$, because of Proposition~\ref{prop_partial_homo}.)
The result would follow if we prove that $\vartheta$ yields a $k$-tensorial homomorphism from $(\delta\X)^\tensor{k}$ to $\freeM((\delta\A)^\tensor{k})$.
Observe first that $\Vset((\delta\X)^\tensor{k})=\Vset(\delta\X)^k=\Eset(\X)^k$ and $\Vset(\freeM((\delta\A)^\tensor{k}))=\Mminion^{(|\Vset((\delta\A)^\tensor{k})|)}=\Mminion^{(|\Vset(\delta\A)^k|)}=\Mminion^{(|\Eset(\A)^k|)}=\Mminion^{(m^k)}$, so the domain and codomain of $\vartheta$ are correct.
Take $\bv=((x,y),(y,z))\in \Eset(\delta\X)$ (so both $(x,y)$ and $(y,z)$ belong to $\Eset(\X)$) and consider the tensor $\bv^\tensor{k}\in \Eset((\delta\X)^\tensor{k})$. To conclude that $\vartheta$ is a homomorphism, we need to show that $\vartheta(\bv^\tensor{k})\in\Eset(\freeM((\delta\A)^\tensor{k}))$; i.e., we need to find some $Q\in\Mminion^{(|\Eset(\delta\A)|)}$ such that $\vartheta(\bv_\bi)=Q_{/\pi_\bi}$ for each $\bi\in [2]^k$, where $\pi_\bi:\Eset(\delta\A)\to\Vset(\delta\A)^k=\Eset(\A)^k$ is the map sending $\bd\in\Eset(\delta\A)$ to $\bd_\bi$. Using that $k\geq 2$, we can consider a tuple $\bw\in \Vset(\X)^{2k}$ satisfying $\bw_{\ang{3}}=(x,y,z)$.
Consider the set $S=\{\ba\in \Vset(\A)^{2k}:\ba_{(\ell,\ell+1)}\in\Eset(\A) \mbox{ for }\ell\in [2]\}$. It follows directly from Proposition~\ref{prop_partial_homo} that 
\begin{align}
\label{eqn_halloween_22_32}
\{\ba\in\Vset(\A)^{2k}:E_\ba\ast\xi(\bw)\neq\bzero_d\}\;\;\subseteq\;\; S.
\end{align}
Take the function
\begin{align*}
    \tau:\Vset(\A)^{2k}&\to \Eset(\delta\A)\\
    \ba&\mapsto\left\{
    \begin{array}{cl}
        (\ba_{(1,2)},\ba_{(2,3)}) &  \mbox{ if }\ba\in S\\
         \textbf{t} & \mbox{ otherwise.} 
    \end{array}
    \right.
\end{align*}
We define $Q=\xi(\bw)_{/\tau}$. (Note that $Q$ does not depend on the choice of $\textbf{t}$, because of Proposition~\ref{prop_partial_homo}.) Let $\bi\in [2]^k$; we need to show that $\vartheta(\bv_\bi)=Q_{/\pi_\bi}$. Consider the tuple $\bj\in [3]^{2k}$ defined by $j_{2\ell-1}=i_\ell$, $j_{2\ell}=i_\ell+1$ for each $\ell\in [k]$, and notice that $\gamma(\bv_\bi)=\bw_\bj$. It follows that
\begin{align}
\label{eqn_26sep_1151_a}
\notag
    \vartheta(\bv_\bi)
    &\spaceeq
    \xi(\gamma(\bv_\bi))_{/\beta}
    \spaceeq
    \xi(\bw_\bj)_{/\beta}
    \spaceeq
    P_\beta\cont{2k}\xi(\bw_\bj)
    \spaceeq
    P_\beta\cont{2k}\left(\Pi^{n\cdot\bone_{2k}}_\bj\cont{2k}\xi(\bw)\right)\\
    &\lemeq{lem_associativity_contraction}
    P_\beta\cont{2k}\Pi^{n\cdot \bone_{2k}}_{\bj}\cont{2k} \xi(\bw),
    \intertext{where the fourth equality follows from the fact that $\xi$ is $(2k)$-tensorial,
    while}
\label{eqn_26sep_1151_b}
    Q_{/\pi_\bi}
    &\spaceeq
    (\xi(\bw)_{/\tau})_{/\pi_\bi}
    \equationeq{eqn_composition_minors}
    \xi(\bw)_{/\pi_\bi\circ\tau}
    \spaceeq
    P_{\pi_\bi\circ\tau}\cont{2k}\xi(\bw).
\end{align}
Consider the function $\rho:\Vset(\A)^{2k}\to\Vset(\A)^{2k}$ defined by $\bc\mapsto\bc_\bj$ for each $\bc\in\Vset(\A)^{2k}$.
Observe that the functions $\beta\circ\rho$ and $\pi_\bi\circ\tau$ coincide on the set $S\subseteq \Vset(\A)^{2k}$. Indeed, for any $\bc\in S$,
\begin{align}
\label{eqn_1237_26sep}
\notag
    \beta\circ\rho(\bc)
    &\spaceeq
    \beta(\bc_\bj)
    \spaceeq
    \beta((c_{i_1},c_{i_1+1},c_{i_2},c_{i_2+1},\dots,c_{i_k},c_{i_k+1}))\\
\notag
    &\spaceeq
    ((c_{i_1},c_{i_1+1}),(c_{i_2},c_{i_2+1}),\dots,(c_{i_k},c_{i_k+1}))\\
    &\spaceeq
    (\bc_{(1,2)},\bc_{(2,3)})_\bi
    \spaceeq
    (\tau(\bc))_\bi
    \spaceeq
    \pi_\bi\circ\tau(\bc).
\end{align}
For $\ba\in \Eset(\A)^k$, we find
{
\allowdisplaybreaks
\begin{align*}
    E_\ba\ast\vartheta(\bv_\bi)
    &\equationeq{eqn_26sep_1151_a}
    E_\ba\ast\left(P_\beta\cont{2k}\Pi^{n\cdot \bone_{2k}}_{\bj}\cont{2k} \xi(\bw)\right)
    \lemeq{lem_associativity_contraction}
    E_\ba\ast P_\beta\ast\Pi^{n\cdot \bone_{2k}}_{\bj}\ast \xi(\bw)\\
    &\lemeq{lem_calculation_rule_P}
    \sum_{\bb\in \beta^{-1}(\ba)}E_\bb\ast\Pi^{n\cdot \bone_{2k}}_{\bj}\ast \xi(\bw)
    \lemeq{lem_Pi_basic}
    \sum_{\bb\in \beta^{-1}(\ba)}\sum_{\substack{\bc\in \Vset(\A)^{2k}\\\bc_\bj=\bb}}E_\bc\ast \xi(\bw)
    \spaceeq
    \sum_{\substack{\bc\in \Vset(\A)^{2k}\\\beta(\bc_\bj)=\ba}}E_\bc\ast \xi(\bw)\\
    &\equationeq{eqn_halloween_22_32}
    \sum_{\substack{\bc\in S\\\beta(\bc_\bj)=\ba}}E_\bc\ast \xi(\bw)
    \spaceeq
    \sum_{\substack{\bc\in S\\\beta\circ\rho(\bc)=\ba}}E_\bc\ast \xi(\bw)
    \equationeq{eqn_1237_26sep}
    \sum_{\substack{\bc\in S\\\pi_\bi\circ\tau(\bc)=\ba}}E_\bc\ast \xi(\bw)\\
    &\equationeq{eqn_halloween_22_32}
    \sum_{\substack{\bc\in \Vset(\A)^{2k}\\\pi_\bi\circ\tau(\bc)=\ba}}E_\bc\ast \xi(\bw)
    \lemeq{lem_calculation_rule_P}
    E_\ba\ast P_{\pi_\bi\circ\tau}\ast\xi(\bw)
    \lemeq{lem_associativity_contraction}
    E_\ba\ast \left(P_{\pi_\bi\circ\tau}\cont{2k}\xi(\bw)\right)
    \equationeq{eqn_26sep_1151_b}
    E_\ba\ast Q_{/\pi_\bi},
\end{align*}
}
which concludes the proof that $\vartheta(\bv_\bi)=Q_{/\pi_\bi}$, thus establishing that $\vartheta$ is a homomorphism.

We are left to prove that $\vartheta$ is $k$-tensorial. To that end, consider some tuples $\bx\in\Eset(\X)^k$ and $\bi\in [k]^k$. We need to show that $\vartheta(\bx_\bi)=\Pi^{m\cdot\bone_k}_\bi\cont{k}\vartheta(\bx)$. Consider the tuple $\bj\in [2k]^{2k}$ defined by $j_{2\ell-1}=2 i_\ell-1$, $j_{2\ell}=2 i_\ell$ for each $\ell\in[k]$, and observe that $\gamma(\bx_\bi)=(\gamma(\bx))_\bj$. Therefore, 
\begin{align*}
    \vartheta(\bx_\bi)
    &\spaceeq
    \xi(\gamma(\bx_\bi))_{/\beta}
    \spaceeq
    \xi((\gamma(\bx))_{\bj})_{/\beta}
    \spaceeq
    P_\beta\cont{2k}\xi((\gamma(\bx))_{\bj})
    \spaceeq
    P_\beta\cont{2k}\left(\Pi^{n\cdot\bone_{2k}}_\bj\cont{2k} \xi(\gamma(\bx))\right)\\
    &\lemeq{lem_associativity_contraction}
    P_\beta\cont{2k}\Pi^{n\cdot\bone_{2k}}_\bj\cont{2k} \xi(\gamma(\bx)),
\end{align*}
where the fourth equality follows from the fact that $\xi$ is $(2k)$-tensorial. Moreover,
\begin{align*}
    \Pi^{m\cdot\bone_k}_\bi\cont{k}\vartheta(\bx)
    &\spaceeq
    \Pi^{m\cdot\bone_k}_\bi\cont{k}\xi(\gamma(\bx))_{/\beta}
    \spaceeq
     \Pi^{m\cdot\bone_k}_\bi\cont{k}\left(P_\beta\cont{2k}\xi(\gamma(\bx))\right)
     \lemeq{lem_associativity_contraction}
     \Pi^{m\cdot\bone_k}_\bi\cont{k}P_\beta\cont{2k}\xi(\gamma(\bx)).
\end{align*}
The claim would then follow if we show that the two tensors $P_\beta\cont{2k}\Pi^{n\cdot\bone_{2k}}_\bj$ and $\Pi^{m\cdot\bone_k}_\bi\cont{k}P_\beta$ coincide. To that end, observe first that the identity $\beta(\bc_\bj)=(\beta(\bc))_\bi$ holds for any $\bc\in \Vset(\A)^{2k}$. Hence, for each $\ba\in\Eset(\A)^k$, we have
\begin{align*}
    E_\ba\ast\left(P_\beta\cont{2k}\Pi^{n\cdot\bone_{2k}}_\bj\right)
    &\lemeq{lem_associativity_contraction}
    E_\ba\ast P_\beta\ast\Pi^{n\cdot\bone_{2k}}_\bj
    \lemeq{lem_calculation_rule_P}
    \sum_{\bb\in\beta^{-1}(\ba)}E_\bb\ast \Pi^{n\cdot\bone_{2k}}_\bj
    \lemeq{lem_Pi_basic}
    \sum_{\bb\in\beta^{-1}(\ba)}
    \sum_{\substack{\bc\in \Vset(\A)^{2k}\\\bc_\bj=\bb}}E_\bc\\
    &\spaceeq
    \sum_{\substack{\bc\in \Vset(\A)^{2k}\\\beta(\bc_\bj)=\ba}}E_\bc
    \spaceeq
    \sum_{\substack{\bc\in \Vset(\A)^{2k}\\(\beta(\bc))_\bi=\ba}}E_\bc
    \spaceeq
    \sum_{\substack{\bb\in\Eset(\A)^k\\\bb_\bi=\ba}}\sum_{\bc\in\beta^{-1}(\bb)}E_\bc\\
    &\lemeq{lem_calculation_rule_P}
    \sum_{\substack{\bb\in\Eset(\A)^k\\\bb_\bi=\ba}}E_\bb\ast P_\beta
    \lemeq{lem_Pi_basic}
    E_\ba\ast \Pi^{m\cdot\bone_k}_\bi\ast P_\beta
    \lemeq{lem_associativity_contraction}
    E_\ba\ast \left(\Pi^{m\cdot\bone_k}_\bi\cont{k} P_\beta\right).
\end{align*}
It follows that $P_\beta\cont{2k}\Pi^{n\cdot\bone_{2k}}_\bj=\Pi^{m\cdot\bone_k}_\bi\cont{k}P_\beta$,
as desired.
\end{proof}

\begin{prop*}[Proposition~\ref{prop_reduction_line_digraph_BA} restated]
Let $2\leq k\in\N$, let $\X,\A$ be digraphs, and suppose that $\BA^{2k}(\X,\A)=\YES$.
Then $\BA^{k}(\delta\X,\delta\A)=\YES$.
\end{prop*}
\begin{proof}
The result immediately follows from Proposition~\ref{prop_reduction_line_digraph} and Theorem~\ref{thm_acceptance_BA_hierarchy_general} and from the fact that $\BAminion$ is a conic minion (cf.~Example~\ref{Qconv_Zaff_conic_not}). 
\end{proof}

We next show that acceptance of hierarchies of relaxations built on linear minions is preserved under homomorphisms of the template. Proposition~\ref{prop_BA_preserves_homo}---the last missing piece in the proof of Theorem~\ref{thm_BLPAIPk_no_solves_AGC}---will then follow as a corollary.

\begin{prop}
\label{prop_acceptance_preserved_under_homo_linear_minions}
Let $\Mminion$ be a linear minion, let $k\in\N$, let $\X,\A,\B$ be digraphs such that $\A\to\B$, and suppose that there exists a $k$-tensorial homomorphism $\Xk\to\freeM(\Ak)$. Then there exists a $k$-tensorial homomorphism $\Xk\to\freeM(\Bk)$.
\end{prop}
\begin{proof}
Let $f:\A\to\B$ be a homomorphism, and consider the functions $g:\Vset(\A)^k\to\Vset(\B)^k$ defined by $(a_1,\dots,a_k)\mapsto (f(a_1),\dots,f(a_k))$ and $h:\Eset(\A)\to\Eset(\B)$ defined by $(a_1,a_2)\mapsto (f(a_1),f(a_2))$. (Notice that $h$ is well defined as $f$ is a homomorphism.) 
Suppose, without loss of generality, that $\Vset(\A)=[n]$ and $\Vset(\B)=[p]$ for some $n,p\in\N$.
Let $\xi$ be a $k$-tensorial homomorphism from $\Xk$ to $\freeM(\Ak)$, and consider the function 
\begin{align*}
    \vartheta:\Vset(\X)^k&\to\Mminion^{(p^k)}.\\
    \bx&\mapsto\xi(\bx)_{/g}
\end{align*}
We claim that $\vartheta$ is a $k$-tensorial homomorphism from $\Xk$ to $\freeM(\Bk)$.

To show that $\vartheta$ is a homomorphism, take $\bx\in\Eset(\X)$, so $\bx^\tensor{k}\in\Eset(\Xk)$. Since $\xi$ is a homomorphism, $\xi(\bx^\tensor{k})\in\Eset(\freeM(\Ak))$, so there exists $Q\in\Mminion^{(|\Eset(\A)|)}$ such that $\xi(\bx_\bi)=Q_{/\pi^\A_\bi}$ for each $\bi\in [2]^k$---where the superscript ``$\A$'' indicates that $\pi_\bi$ is defined for the digraph $\A$; i.e., $\pi^\A_\bi:\Eset(\A)\to\Vset(\A)^k$ is the function given by $\ba\mapsto\ba_\bi$. Define $W=Q_{/h}\in\Mminion^{(|\Eset(\B)|)}$. Given $\bi\in [2]^k$, let $\pi^\B_\bi:\Eset(\B)\to\Vset(\B)^k$ be the function given by $\bb\mapsto\bb_\bi$. Note that 
$g\circ\pi^\A_\bi=\pi^\B_\bi\circ h$.
Indeed, for any $\ba\in\Eset(\A)$, we have
\begin{align*}
    g(\pi^\A_\bi(\ba))
    \spaceeq
    g(\ba_\bi)
    \spaceeq
    (f(a_{i_1}),\dots,f(a_{i_k}))
    \spaceeq
    (f(a_1),f(a_2))_\bi
    \spaceeq
    (h(\ba))_\bi
    \spaceeq
    \pi^\B_\bi(h(\ba)).
\end{align*}
Therefore, we find
\begin{align*}
    \vartheta(\bx_\bi)
    \spaceeq
    \xi(\bx_\bi)_{/g}
    \spaceeq
    (Q_{/\pi^\A_\bi})_{/g}
    \equationeq{eqn_composition_minors}
    Q_{/g\circ\pi^\A_\bi}
    \spaceeq 
    Q_{/\pi^\B_\bi\circ h}
    \equationeq{eqn_composition_minors}
    (Q_{/h})_{/\pi^\B_\bi}
    \spaceeq
    W_{/\pi^\B_\bi}.
\end{align*}
It follows that $\vartheta(\bx^\tensor{k})\in\Eset(\freeM(\Bk))$, so $\vartheta$ is a homomorphism.

To show that $\vartheta$ is $k$-tensorial, take $\bx\in\Vset(\X)^k$ and $\bi\in[k]^k$.
Using that $\xi$ is $k$-tensorial, we find
\begin{align*}
    \vartheta(\bx_\bi)
    &\spaceeq
    \xi(\bx_\bi)_{/g}
    \spaceeq
    \left(\Pi^{n\cdot\bone_k}_\bi\cont{k}\xi(\bx)\right)_{/g}
    \spaceeq
    P_g\cont{k}\left(\Pi^{n\cdot\bone_k}_\bi\cont{k}\xi(\bx)\right)
    \lemeq{lem_associativity_contraction}
    P_g\cont{k}\Pi^{n\cdot\bone_k}_\bi\cont{k}\xi(\bx),
\end{align*}
while
\begin{align*}
    \Pi^{p\cdot\bone_k}_\bi\cont{k}\vartheta(\bx)
    &\spaceeq
    \Pi^{p\cdot\bone_k}_\bi\cont{k}\xi(\bx)_{/g}
    \spaceeq
    \Pi^{p\cdot\bone_k}_\bi\cont{k}\left(P_g\cont{k}\xi(\bx)\right)
    \lemeq{lem_associativity_contraction}
    \Pi^{p\cdot\bone_k}_\bi\cont{k}P_g\cont{k}\xi(\bx).
\end{align*}
Therefore, to obtain $\vartheta(\bx_\bi)= \Pi^{p\cdot\bone_k}_\bi\cont{k}\vartheta(\bx)$ and thus conclude that $\vartheta$ is $k$-tensorial, it suffices to prove that $P_g\cont{k}\Pi^{n\cdot\bone_k}_\bi=\Pi^{p\cdot\bone_k}_\bi\cont{k}P_g$. 
To that end, we apply a similar argument to the one used at the end of the proof of Proposition~\ref{prop_reduction_line_digraph}.
Notice that both these tensors belong to $\cT^{(p\cdot\bone_k,n\cdot\bone_k)}(\Q)$. Given $\ba\in\Vset(\A)^k$ and $\bb\in\Vset(\B)^k$, we find
\begin{align*}
    E_\bb\ast\left(P_g\cont{k}\Pi^{n\cdot\bone_k}_\bi\right)\ast E_\ba
    &\lemeq{lem_associativity_contraction}
    E_\bb\ast P_g\ast\Pi^{n\cdot\bone_k}_\bi\ast E_\ba
    \lemeq{lem_calculation_rule_P}
    \sum_{\bc\in g^{-1}(\bb)}E_\bc\ast\Pi^{n\cdot\bone_k}_\bi\ast E_\ba\\
    &\spaceeq
    \left\{\begin{array}{cl}
         1&\mbox{ if }\ba_\bi\in g^{-1}(\bb)   \\
         0&\mbox{ otherwise} 
    \end{array}\right.
    \spaceeq
    \left\{\begin{array}{cl}
         1&\mbox{ if }g(\ba_\bi)=\bb   \\
         0&\mbox{ otherwise}, 
    \end{array}\right.
\end{align*}
while
\begin{align*}
    E_\bb\ast\left(\Pi^{p\cdot\bone_k}_\bi\cont{k}P_g\right)\ast E_\ba
    &\lemeq{lem_associativity_contraction}
    E_\bb\ast\Pi^{p\cdot\bone_k}_\bi\ast P_g\ast E_\ba
    \lemeq{lem_Pi_basic}
    \sum_{\substack{\bd\in\Vset(\B)^k\\\bd_\bi=\bb}}E_\bd\ast P_g\ast E_\ba\\
    &\spaceeq
     \left\{\begin{array}{cl}
         1&\mbox{ if }(g(\ba))_\bi=\bb   \\
         0&\mbox{ otherwise.} 
    \end{array}\right.
\end{align*}
Since $g(\ba_\bi)=(g(\ba))_\bi$, the two expressions above coincide, thus implying that $P_g\cont{k}\Pi^{n\cdot\bone_k}_\bi=\Pi^{p\cdot\bone_k}_\bi\cont{k}P_g$, as required.
\end{proof}

\begin{prop*}[Proposition~\ref{prop_BA_preserves_homo} restated]
Let $2\leq k\in\N$, let $\X,\A,\B$ be digraphs such that $\A\to\B$, and suppose that $\BA^k(\X,\A)=\YES$. Then $\BA^k(\X,\B)=\YES$.
\end{prop*}
\begin{proof}
By Theorem~\ref{thm_acceptance_BA_hierarchy_general}, $\BA^k(\X,\A)=\YES$ implies the existence of a $k$-tensorial homomorphism from $\Xk$ to $\freeBA(\Ak)$. By Proposition~\ref{prop_acceptance_preserved_under_homo_linear_minions}, it follows that there exists a $k$-tensorial homomorphism from $\Xk$ to $\freeBA(\Bk)$. Again by Theorem~\ref{thm_acceptance_BA_hierarchy_general}, we conclude that $\BA^k(\X,\B)=\YES$. 
\end{proof}

\section*{Acknowledgements} 
The authors are grateful to Jakub Opr\v{s}al, who suggested to us a connection
between the line digraph reduction and the Sherali--Adams LP hierarchy for
approximate graph colouring. 
We also thank the anonymous reviewers of the two extended
abstracts~\cite{cz23soda:aip,cz23stoc:ba} and this full version. In particular,
one reviewer suggested a simpler proof of Lemma~\ref{lem_involved_1511_1718}.

{\small
\bibliographystyle{plainurl}
\bibliography{cz}
}

\end{document}